\documentclass{article}

\RequirePackage[OT1]{fontenc} \RequirePackage{amsthm,amsmath} \RequirePackage[authoryear]{natbib}

\usepackage{amsfonts, amssymb, graphicx, mathtools, epstopdf,amsmath} 
\usepackage{accents, appendix, xcolor}
\usepackage{enumitem,}

\usepackage[colorlinks, citecolor=blue, urlcolor=blue]{hyperref}

\usepackage{caption}
\usepackage{subcaption, multirow, bigfoot}

\usepackage{multibib}

\newcites{appendix}{References}

\theoremstyle{definition}

\newtheorem{example}{Example}

\theoremstyle{plain}
\newtheorem{theorem}{Theorem}
\newtheorem{assumption}{Assumption}
\newtheorem{proposition}{Proposition}
\newtheorem{corollary}{Corollary}
\newtheorem{lemma}{Lemma} 

\theoremstyle{remark}

\newcommand{\Var}{{\operatorname{Var}\, }}

\newcommand{\Iscr}{{\mathcal I}}

\newcommand{\Haj}{{\textup{Haj}}}
\newcommand{\Cov}{{\operatorname{Cov}}}

\newcommand{\basic}{{\textup{basic}}}
\newcommand{\ind}{\perp\!\!\!\!\perp} 

\title{\normalsize Agnostic Characterization of Interference in Randomized Experiments}
\author{\normalsize David Choi}
\date{\normalsize \today}

\begin{document}
\maketitle

\abstract{
We give an approach for characterizing interference by lower bounding the number of units whose outcome depends on selected groups of treated individuals, such as depending on the treatment of others, or others who are at least a certain distance away. The approach is applicable to randomized experiments with binary-valued outcomes. Asymptotically conservative point estimates and one-sided confidence intervals may be constructed with no assumptions beyond the known randomization design, allowing the approach to be used when interference is poorly understood, or when an observed network might only be a crude proxy for the underlying social mechanisms. Point estimates are equal to H\'{a}jek-weighted comparisons of units with differing levels of treatment exposure. Empirically, we find that the width of our interval estimates is competitive with (and often smaller than) those of the EATE, an assumption-lean treatment effect, suggesting that the proposed estimands may be intrinsically easier to estimate than treatment effects.
 }

\section{Introduction}

There is growing interest in the analysis of randomized experiments where participants may affect each other (``interference between units''), for example by transmission of disease, sharing of information, peer influence, or economic competition. Such mechanisms often play fundamental roles in our understanding of social behaviors, and their empirical characterization may be of scientific interest.

We believe that a better understanding of interference may not only be of scientific interest, but also of immense practical importance. One of the most well-known findings in the history of public health--that cholera could be transmitted through shared water utilities--was a particularly impactful characterization of interference \citep[Ch. 3]{freedman2010statistical}. Conversely, \cite{list2022voltage} and \cite{list2023voltage} discuss fields experiments at Uber in which the results were misinterpreted because interference was ignored, leading to decisions that harmed the company.\footnote{\cite[pp. 93-94]{list2022voltage} reports that giving a rider a promotional coupon increased their Uber usage in a field study, but the possibility of negative effects on other riders (by increasing their waiting time for rides) was ignored. The promotion was subsequently expanded to cover an entire city, leading to long wait times followed by substantial reduction in Uber usage. Similarly, \cite{list2023voltage} reports that allowing a driver to receive tips increased their productivity in a field experiment, but the possibility of negative effects on the productivity of other drivers (who competed for the same pool of riders) was again ignored. Tipping was subsequently allowed for all drivers, but no gains in driver productivity were realized.} The 2007 subprime mortgage crisis and subsequent ``Great Recession'' have been attributed in part to questionable modeling of causal dependencies between mortgage defaults, which led to catastrophic underestimation of credit risk \citep{salmon2012formula}.\footnote{Interference between mortgages arises if housing prices can affect--and are affected by--the rate of mortgage default.} 


In this work, we give a general approach for characterizing whether interference-causing mechanisms are widely prevalent in an experimental setting. The approach is applicable to experiments with binary-valued outcomes. It can be used to explore questions such as
\begin{enumerate}  
  \item How many units are affected by (i.e., their outcome depends on) the treatment of others? of distant others?  
  \item How many units, when receiving a certain level of direct treatment, are affected by the treatment of others? For example, if a vaccine grants immunity then the disease status of a unit receiving the placebo (but not necessarily when receiving the vaccine), might depend on the vaccination status of others.
\end{enumerate}
For these types of questions, our approach gives conservative point estimates and one-sided confidence intervals, both of which lower bound the true value. Under reasonable experiment designs, our point estimates will converge to a lower bound on the true value, while the one-sided interval will cover with probability converging to at least the desired level. These consistency and coverage properties are agnostic, in that they require no assumptions or restrictions on the nature of the interference, and rely only on the known experiment design. As a result, our estimates will remain valid lower bounds (but may become less tight) if an observed social network is only a crude proxy for the actual social mechanisms. 

We will demonstrate our approach using an experiment of \cite{cai2015social}, which studied whether crop insurance decisions of farmers could be affected by information received by their peers. Such experiments, in which the primary goal is to demonstrate the presence of interference, may be considered part of the collection of ``shoe leather'' efforts that are necessary to deeply understand a problem setting \citep{freedman2010statistical}.\footnote{``Historically, `shoe-leather epidemiology' is epitomized by intensive, door-to-door canvassing that wears out investigators' shoes'' \citep[p. xiii]{freedman2010statistical}.}  Without such efforts, treatment effects and policy recommendations may lead to questionable conclusions, such as in the abovementioned Uber examples where interference was mischaracterized. On the other hand, a more informed understanding of interference may lead to better assumptions and more credible policy recommendations.
For the experiment of \cite{cai2015social}, our results strengthen their findings, with confidence intervals that are both tighter and require fewer assumptions than those of the expected average treatment effect (EATE), an assumption-lean treatment effect whose study was motivated by interference \citep{savje2017average}.

\section{Related Work}

Our work is closely related to hypothesis testing in the presence of interference, which has been studied in \cite{aronow2012general,athey2017exact,pouget2019testing,basse2019randomization,han2022detecting,puelz2022graph,hoshino2023randomization}, for testing of null effects and hypotheses regarding interference structures, such as the null of no interference between units. In these settings, p-values generally cannot be inverted into confidence intervals for a treatment effect unless additional assumptions are added, and more generally the connection between testing and estimation under interference is not fully understood. Our work complements this literature, as we invert tests of interference specifications to construct one-sided intervals without additional assumptions -- however not for a treatment effect, but rather for the number of units who violate the specification. 

In doing so, we improve on prior work in this area \citep{choi2023estimating,choi2024new} in multiple respects. Most importantly, our confidence intervals are much tigher than those of previous approaches, due to a combination of improved computation, better estimator design, and a novel approach for variance estimation. Our target estimand is a proper parameter, whereas in prior approaches it was defined to be a random quantity whose value depended on the treatment assignment. Finally, our approach generalizes to answer a wide variety of questions regarding the characterization of interference.

The estimation of treatment effects in the presence of interference is an active area of research \citep{rosenbaum2007interference, savje2017average,li2019randomization,hu2021average, li2022random,leung2022rate,ogburn2022causal,harshaw2022design,zhao2022reconciling, park2023assumption, gao2023causal}. In many of these works, the target parameter describes the effect of a marginal departure from the experiment design. For such estimands, consistent estimates can generally be constructed in assumption-lean regimes, although additional assumptions are needed for coverage of interval estimates. \cite{harshaw2021optimized} focuses on the problem of variance estimation. For settings with interference, \cite{kandiros2024conflict} proposes a novel, more-tailored approach to experiment design. Estimation of total effects (as opposed to marginal effects) assuming low order interactions between unit treatments is considered in \cite{cortez2022staggered} under staggered rollout, and in \cite{cortez2023exploiting} for known exposure neighborhoods. \cite{farias2022markovian}, \cite{dhaouadi2023price}, and \cite{johari2024does} consider a particular form of interference (``Markovian'') arising in e-commerce settings. \cite{banerjee2024less} gives an empirical example of information sharing in which interference may be more subtle and complex than is captured by models of local exposure such as infection.

\section{Methodology (basic case)} \label{sec: basic}

In this section, we present a basic example, in which the estimand is the number of units who are affected by (i.e., their outcome depends on) the assignment of treatment. This quantity, which we denote by $\tau^\basic$, includes units who are affected by their own treatment as well as those who are affected by the treatment of others. 

Under certain experiment designs, $\tau^\basic$ may be used to characterize interference. In bipartite network designs, some units are not directly treated, and $\tau^\basic$ may be used to characterize whether these units are affected by the treatment of others \citep{zigler2020bipartite}. In experiments involving the assignment of units to groups, or the formation of network ties, $\tau^\basic$ characterizes whether units are affected by their randomly assigned neighbors or group members \citep{basse2024randomization,li2019randomization}. 

The remaining organization of the paper is as follows. Sections \ref{sec: tau any preliminaries} through \ref{sec: computation} respectively present the generative model, the estimand $\tau^\basic$, point estimation, usage of network information, H\'{a}jek normalization, inference, and computation. In Section \ref{sec: generalization}, the approach is then generalized to cover a wider range of target parameters, such as the questions presented in the introduction. Section \ref{sec: theory} considers asymptotic behavior, and Section \ref{sec: examples} gives a simulation study and data example to illustrate usage. All propositions and theorems are proven in the supplemental materials.

\subsection{Preliminaries and Oveview} \label{sec: tau any preliminaries}

We consider an experiment on $N$ units, with $X=(X_1,\ldots,X_N)$ denoting the binary-valued treatment of each unit, and $Y=(Y_1,\ldots,Y_N)$ denoting their binary outcomes. We allow for arbitrary interference, so that the outcome $Y_i$ of unit $i$ may potentially depend on all $N$ treatments, 
\begin{align}\label{eq: f}
Y_i = f_i(X_1,\ldots,X_N), \qquad i \in [N]
\end{align}
where the potential outcome mappings $\{f_i\}_{i=1}^N$ may be arbitrary and unknown.

Given an observed network, we will use ${\mathcal N}_i$ to denote the set consisting of unit $i$ and $i$'s ``close'' associates, and ${\mathcal N}_i^{(2)}$ to denote ``non-close'' associates of unit $i$, as defined by the researcher. For example, ${\mathcal N}_i$ might denote unit $i$ and its direct neighbors on the network, while ${\mathcal N}_i^{(2)}$ might denote units who are not in $\mathcal{N}_i$ but have overlapping sets of neighbors with unit $i$. We will let $W_i$ denote the thresholded number of treated close associates for unit $i$,
\begin{align} \label{eq: W}
 W_i = 1\{\textup{if at least $t_i$ units in ${\mathcal N}_i$, excluding unit $i$, are treated}\},
 \end{align}
and $W_i^{(2)}$ will denote the thresholded number of treated units in $\mathcal{N}_i^{(2)}$,
\begin{align} \label{eq: W2}
 W_i^{(2)} = 1\{ \textup{ if at least $t_i^{(2)}$ units in ${\mathcal N}_i^{(2)}$ are treated} \}
 \end{align}
for some choice of thresholds $t_i$ and $t_i^{(2)}$.

Given integer $n > 0$, we use $[n]$ to denote the set $\{1,\ldots, n\}$. Given a set $S \subseteq [N]$, we let $|S|$ denote its cardinality, and let $X_S$ denote the subvector of $X$ corresponding to the members of $S$.

\paragraph{Overview of Approach}

To estimate lower bounds on our target estimand $\tau^\basic$, we will take the following approach:

\begin{enumerate}
  \item Define the estimand $\tau^\basic$ in terms of an particular subset $\Iscr \subseteq [N]$, whose identity is unknown
  
  \item Propose idealized estimators $\hat{\tau}_1$ and $\hat{\tau}_2$ which will have good statistical properties, such as consistency and asymptotic normality, but require knowledge of $\Iscr$
  
  \item Show that $\Delta$, the difference in average outcomes between treated and control (which can be computed without knowledge of $\Iscr$) converges to a lower bound for $\max(\hat{\tau}_1,\hat{\tau}_2)$, so that if $\tau_1$ and $\tau_2$ are both consistent for $\tau^\basic$, then $\Delta$ is an asymptotic lower bound.
  
  \item Lower bound the confidence intervals induced by $\hat{\tau}_1$ and $\hat{\tau}_2$ and their variance estimates, by solving an integer program over the unknown subset $\Iscr$.
\end{enumerate}

\subsection{Estimand} \label{sec: tau any estimand}

Let $\tau^\basic$ denote the number of units who are affected by treatment, including both their own treatment or the treatment of others. We may write $\tau^\basic$ as the number of units whose outcome $f_i(X)$ is not constant over the possible treatment assignments $X \in \{0,1\}^N$.
\[ \tau^\basic = \sum_{i=1}^N 1\{ f_i(X) \textup{ is not constant in } X\} \]
Define $\Iscr \subset [N]$ as the subset of units who are unaffected by treatment and have constant outcome mappings, 
\[ \Iscr = \{i: f_i(X) \textup{ is constant in }X\}\]
so that $\tau^\basic = N - |\Iscr|$.

\subsection{Point estimation}
Let $H_i$ denote the indicator of whether unit $i$'s treatment and outcome have the same binary value,
\begin{align}\label{eq: H_i}
H_i = 1\{(X_i,Y_i) = (1,1) \textup{ or } (0,0)\},
\end{align}
and let $\hat{\tau}_1$ and $\hat{\tau}_2$ denote sampling-based estimators of $\tau^\basic = N - |\Iscr|$, in which the unknown cardinality of $\Iscr$ is unbiasedly estimated by a probability-weighted (i.e., Horvitz-Thompson) sample:
\begin{align}\label{eq: tau hat}
\hat{\tau}_1 = N - \sum_{i \in \Iscr} \frac{1\{H_i=1\}}{P(H_i=1)} \qquad \textup{ and } \qquad \hat{\tau}_2 = N - \sum_{i \in \Iscr} \frac{1\{H_i =0 \}}{P(H_i=0)}
\end{align}
Because $\hat{\tau}_1$ and $\hat{\tau}_2$ involve only units in $\Iscr$ whose outcomes are unaffected by treatment, they often will exhibit simple statistical behavior, even if strong interference exists between units who are not in $\Iscr$. For example, if treatment is assigned by independent Bernoulli randomization, then $\hat{\tau}_1$ and $\hat{\tau}_2$ are sums of independent variables.\footnote{If $i\in \Iscr$, then $Y_i$ equals a constant $c_i$ for all $X \in \{0,1\}$. Thus it holds for example that $P(H_i=1,H_j=1) = P(X_i=c_i, X_j=c_j) = P(X_i=c_i)P(X_j=c_j) = P(H_i=1)P(H_j=1)$, if $i,j \in \Iscr$ and $X_i \ind X_j$.}
 Similarly, if the dependencies between the unit treatments are bounded, then $\hat{\tau}_1$ and $\hat{\tau}_2$ are sums of variables whose dependencies will be similarly bounded. For this reason, under a variety of designs we may expect the values of $\hat{\tau}_1$ and $\hat{\tau}_2$, while unknown due to $\Iscr$ being unknown, to concentrate at their expectation (which equals $\tau^\basic$) and to be asymptotically normal.

Our motivation for $\hat{\tau}_1$ and $\hat{\tau}_2$ is the following: under mild conditions on the experiment design, the maximum of $\hat{\tau}_1$ and $\hat{\tau}_2$ is asymptotically lower bounded by the magnitude of the propensity-weighted difference in outcomes between treated and control, given by
\[ \Delta = \sum_{i=1}^N \left(\frac{X_i}{P(X_i=1)} - \frac{1-X_i}{P(X_i=0)}\right) Y_i,  \]
as stated by Proposition \ref{th: diff in means} below:

\begin{proposition}\label{th: diff in means}
Let the total weights of the treated and control converge to their expectations, so that
\begin{align} \label{eq: condition X}
\sum_{i=1}^N \frac{X_i}{P(X_i=1)} = N + O_P(N^{1/2}) \quad \textup{and} \quad \sum_{i=1}^N \frac{1-X_i}{P(X_i=0)} = N + O_P(N^{1/2})
\end{align}
Then it holds that
\begin{equation} \label{eq: diff in means}
\left| \sum_{i=1}^N \left(\frac{X_i}{P(X_i=1)} - \frac{1-X_i}{P(X_i=0)}\right) Y_i \right| \leq \max(\hat{\tau}_1, \hat{\tau}_2) + O_P(N^{1/2})
\end{equation}
\end{proposition}

As a corollary, it immediately follows that if $\hat{\tau}_1$ and $\hat{\tau}_2$ are consistent for $\tau^\basic$, then $\max(\hat{\tau}_1,\hat{\tau}_2)$ will be consistent as well, and $|\Delta|$ will be an asymptotic lower bound for $\tau^\basic$. 

\begin{corollary}\label{cor: diff in means}
Let the conditions of Theorem \ref{th: diff in means} hold. If $\hat{\tau}_1$ and $\hat{\tau}_2$ are within $O_P(N^{1/2})$ of $\tau^\basic$, then
\begin{equation} \label{eq: diff in means cor}
\left| \sum_{i=1}^N \left(\frac{X_i}{P(X_i=1)} - \frac{1-X_i}{P(X_i=0)}\right) Y_i \right| \leq \tau^\basic + O_P(N^{1/2})
\end{equation}
\end{corollary}

We note that $\Delta$, which involves units outside of $\Iscr$, need not concentrate in order for \eqref{eq: diff in means} or \eqref{eq: diff in means cor} to hold. In contrast to $\hat{\tau}_1$ and $\hat{\tau}_2$, which depend only on units in $\Iscr$, it may well be that $\Delta$ becomes less stable as the number of participants grows, but nonetheless is asymptotically a valid lower bound on the estimand.

\paragraph{Discussion}
To give intuition for Theorem \ref{th: diff in means}, consider an experiment in which $100$ units are randomly divided into equal-sized groups for treatment and control, with average outcomes of $0.4$ (20/50) for the treatment group and $0.8$ (40/50) for the control. How large can $|\Iscr|$ be? Units in $\Iscr$ have constant outcomes and are equally likely to be assigned to treatment or control. It follows that up to sampling variation, such units should be evenly distributed between treatment and control for each outcome level. Since there are 20 treated units and 40 control units with outcome level 1, it trivially holds that at most all 20 of the treated units with outcome level 1 could belong to $\Iscr$, and thus at most roughly 20 out of the 40 control units with outcome level 1 could belong to $\Iscr$ as well. (Otherwise, units in $\Iscr$ with outcome level 1 would not be evenly distributed between treatment and control.) Likewise, since only 10 control units have outcome level 0, at most all 10 could belong to $\Iscr$, implying that at most roughly 10 out of the 30 treated units with outcome level 0 could belong to $\Iscr$ as well. This yields an upper bound of $(20 + 10)\times 2 = 60$ units who are unaffected by treatment, or equivalently at least 40 units who are affected by treatment. This lower bound of 40 equals the difference in outcomes of 0.4 between treatment and control, and also equals the maximum possible value of $\hat{\tau}_1$ over all hypotheses for the unknown $\Iscr$. The limiting quantities that were used to construct this bound, namely the 20 treated units with outcome equal to 1, and the 10 control units with outcome equal to 0, correspond to the units $\{i: H_i=1\}$.

When might this lower bound on $\tau$ be tight? Can $\tau$ be upper bounded? To build intuition, we continue the example by presenting a generative model under which the lower bound is tight. In this model, the contrast variable $X$ appearing in \eqref{eq: diff in means} is a perfect predictor of the outcome for units not in $\mathcal{I}$. We then describe two models under which the lower bound is loose, and which suggest that non-vacuous upper bounds on $\tau$ cannot be constructed without additional assumptions.

\begin{example} \label{ex: nonnegative effects} (SUTVA with nonnegative effects) Of the 100 units in the experiment, suppose that units 1-60 are unaffected by treatment, and units 61-100 are affected by their own treatment, according to
\[ f_i(X) = \begin{cases} 1 & 1 \leq i \leq 40 \\ 0 & 41 \leq i \leq 60 \\ X_i & 61 \leq i \leq 100 \end{cases}
\]
Under this model, $40$ units are affected by treatment, and the expected outcomes for treated and control can be seen to be $0.8$ and $0.4$. Hence the diff-in-means will equal the fraction of units affected by treatment (up to sampling variation). 
\end{example}

\begin{example}(SUTVA with cancellation of effects) Suppose that units 1-20 are unaffected by treatment, while units 21-40 are affected negatively by their own treatment, and units 41-100 are affected positively, 
\[ f_i(X) = \begin{cases} 1 & 1 \leq i \leq 20 \\ 1-X_i & 21 \leq i \leq 40 \\ X_i & 41 \leq i \leq 100 \end{cases}
\]
While the expected outcomes match the previous example, equaling $0.8$ and $0.4$ for treated and control, the diff-in-means is a loose lower bound on the fraction of units affected by treatment, which is 0.8.
\end{example}

\begin{example} (Interference) Let the outcome models be given by
  \begin{align*} 
  f_i(X) & = \begin{cases} \alpha_{i1} X_i + \alpha_{i0} (1-X_i) & \textup{ if $\sum_i X_i$ is even} \\
  \beta_{i1} X_i + \beta_{i0} (1-X_i) & \textup{ if $\sum_i X_i$ is odd} \end{cases}
  \end{align*}
  where $(\alpha_{i1},\alpha_{i0})$ match the potential outcomes of the first example, equaling $(1,1)$ for units 1-40, (0,0) for units 41-60, and $(1,0)$ for units 61-100, while $(\beta_{i1},\beta_{i0})$ takes values $(1,0)$ and $(0,1)$ in some arbitrary proportion. If half of the units are randomly assigned to treatment, then $\sum_i X_i$ is even and the outcomes will be indistinguishable from those of the first example; however, unlike the first example, all 100 units are affected by treatment. In this example $\Delta$ does not concentrate at its expectation, but is nonetheless a valid lower bound for $\tau^\basic$.
\end{example}

\subsection{Using Network Information} Given an elicited social network, we may consider contrasts other than $\Delta$, based on a new variable $Z=(Z_1,\ldots,Z_N)$ which encodes the direct and indirect treatment of each unit. We will consider $Z \in \{1,0,-1\}^N$, where $1$ and $0$ respectively denote generalized versions of treatment and control, and $-1$ denotes that the unit will be excluded from comparison.

For example, to use $W$, the thresholded number of treated neighbors given by \eqref{eq: W}, we can define $Z$ according to
\begin{align}\label{eq: Z basic}
Z_i = \begin{cases} 1 & \textup{ if $X_i = W_i =1$} \\ 0 & \textup{ if $X_i=W_i=0$} \\ -1 & \textup{ if $X_i \neq W_i$} \end{cases} 
\end{align}
Under $Z$ so defined, unit $i$ receives ``effective treatment'' if $X_i=1$ and $W_i=1$; receives ``effective control'' if $X_i=0$ and $W_i=0$; and otherwise is excluded from comparison. 

Given $Z$, we can redefine $H = (H_1,\ldots,H_N)$ to use $Z$ instead of $X$,
\begin{align}\label{eq: S Z}
H_i = \begin{cases} 1 & \textup{ if }(Z_i,Y_i) = (1,1) \textup{ or } (0,0) \\
0 & \textup{ if } (Z_i,Y_i) = (1,0) \textup{ or } (0,1) \\
-1 & \textup{ if } Z_i = -1 \end{cases}
\end{align}
which induces new versions of $\hat{\tau}_1$ and $\hat{\tau}_2$ following \eqref{eq: tau hat} as before. Proposition \ref{th: diff in means} can then be adapted to show that the difference in outcomes between effective treatment and effective control, given by
\[ \Delta^{Z} = \sum_{i=1}^N \left( \frac{1\{Z_i=1\}}{P(Z_i=1)} - \frac{1\{Z_i=0\}}{P(Z_i=0)}\right)Y_i \]
is an asymptotic lower bound on the maximum of $\hat{\tau}_1$ and $\hat{\tau}_2$.

\begin{proposition} \label{th: diff in means Z}
Let $Z$ be given by \eqref{eq: Z basic} and let the total weights of effective treatment and control converge to their expectations,
\[ \sum_{i=1}^N \frac{1\{Z_i=1\}}{P(Z_i=1)} = N + O_P(N^{1/2}) \quad \textup{ and } \quad \sum_{i=1}^N \frac{1\{Z_i=0\}}{P(Z_i=0)} = N + O_P(N^{1/2})\]
and let $\hat{\tau}_1$ and $\hat{\tau}_2$ be given by \eqref{eq: tau hat} for $H$ given by \eqref{eq: S Z}. Then it holds that 
\begin{align} \label{eq: diff in means XW}
\left| \sum_{i=1}^N \left(\frac{1\{Z_i=1\}}{P(Z_i = 1)} - \frac{1\{Z_i=0\}}{P(Z_i=0)}\right)Y_i \right| \leq \max(\hat{\tau}_1, \hat{\tau}_2) + O_P(N^{1/2})
\end{align}
\end{proposition}
If $Z$ is a better predictor of the unit outcomes than $X$, then $\Delta^Z$ may be larger than the comparison of direct treatment and direct control given by $\Delta$. This results in a larger and tighter lower bound on $\tau^{\textup{basic}}$, assuming consistency of $\hat{\tau}_1$ and $\hat{\tau}_2$. 

\subsection{H\'{a}jek Normalization} \label{sec: tau any hajek}

To reduce the variance of our IPW-based estimators, we may use H\'{a}jek normalization \citep{aronow2017estimating,gao2023causal}. As a preliminary step, let $Y_{ki}$ denote transformed outcomes given by
\begin{align}\label{eq: Yk}
Y_{ki} = \begin{cases}  Y_i & \textup{ if } k = 1 \\ 1-Y_i & \textup{ if } k = 2,
\end{cases}
\qquad i \in [N],\ k = 1,2,
\end{align}
This allows $\hat{\tau}_1$ and $\hat{\tau}_2$ using $H$ given by \eqref{eq: S Z} to be written with a single formula
\begin{align} \label{eq: tau hat k}
\hat{\tau}_k = N - \sum_{i \in \Iscr} \frac{1\{Z_i = Y_{ki}\}}{P(Z_i = Y_{ki})} , \qquad k=1,2,
\end{align}
which will simplify our discussion of H\'{a}jek normalization for $\hat{\tau}_1$ and $\hat{\tau}_2$.

To construct a H\'{a}jek-normalized version of $\hat{\tau}_k$ as given by \eqref{eq: tau hat k}, we first observe that it equals
\begin{equation}\label{eq: tau hat pre hajek}
\hat{\tau}_k = N - \sum_{i \in \Iscr} \left(\frac{1\{Z_i =1, Y_{ki}=1\}}{P(Z_i=1)} + \frac{1\{Z_i=0, Y_{ki}=0\}}{P(Z_i=0)} \right) 
\end{equation}
where we have divided the event $\{Z_i = Y_{ki}\}$ into sub-events $\{Z_i = Y_{ki} = 1\}$ and $\{Z_i = Y_{ki} = 0\}$, and replaced the propensity score $P(Z_i = Y_{ki})$ by 
\[ P(Z_i = Y_{ki}) = Y_{ki}P(Z_i=1) + (1-Y_{ki})P(Z_i=0), \qquad \forall\ i \in \Iscr, \]
 which can be done because $Y_{ki}$ is constant in $X$ for $i \in \Iscr$. Eq. \eqref{eq: tau hat pre hajek} suggests a H\'{a}jek normalized version given by
\begin{equation}\label{eq: tau hat hajek}
\hat{\tau}_k^{\Haj} = N - \sum_{i \in \Iscr} \left[\frac{N}{\hat{N}_1} \frac{1\{Z_i=1\}}{P(Z_i=1)} \cdot Y_{ki}  + \frac{N}{\hat{N}_0}\frac{1\{Z_i=0\}}{P(Z_i=0)}  \cdot (1-Y_{ki}) \right],
\end{equation}
where $\hat{N}_0$ and $\hat{N}_1$ are estimates of $N$ given by 
\begin{align*}
\hat{N}_{0} & = \sum_{j=1}^N \frac{1\{Z_j=0\}}{P(Z_j=0)}, & \hat{N}_{1} & = \sum_{j=1}^N \frac{1\{Z_j=1\}}{P(Z_j=1)}.
\end{align*}
In Theorem \ref{th: hajek consistency bernoulli}, we show that $\max(\hat{\tau}_1^\Haj, \hat{\tau}_2^\Haj)$ is lower bounded by a H\'{a}jek-weighted comparison of units,
\[ \left| \sum_{i=1}^N \left( \frac{N}{\hat{N}_1} \frac{Z_i}{P(Z_i=1)} - \frac{N}{\hat{N}_0} \frac{(1-Z_i)}{P(Z_i=0)}\right) Y_i \right| \leq \max(\hat{\tau}_1^\Haj\, , \, \hat{\tau}_2^\Haj)\]
noting that the inequality holds non-asymptotically, in contrast to \eqref{eq: diff in means} and \eqref{eq: diff in means XW}.

\subsection{Inference and Variance Estimation} \label{sec: tau any inference}

If $\hat{\tau}_1^\Haj$ and $\hat{\tau}_2^\Haj$ are asymptotically normal, with consistent variance estimators denoted by $\hat{V}_1$ and $\hat{V}_2$, then by combining 1-sided confidence intervals it holds with probability converging to at least $1-\alpha$ that 
\begin{align} \label{eq: Phi CI}
\tau^\basic & \geq \max\left\{\hat{\tau}_1^\Haj - z_{1-\frac{\alpha}{2}} \sqrt{\hat{V}_1}, \, \hat{\tau}_2^\Haj - z_{1-\frac{\alpha}{2}} \sqrt{\hat{V}_2}\right\}. 
\end{align}
As the right hand side of \eqref{eq: Phi CI} requires knowledge of $\Iscr$, it cannot be computed. 

To construct a computable one-sided confidence interval for $\tau^{\textup{basic}}$, we will lower bound \eqref{eq: Phi CI} by minimizing over all possible hypotheses for the unknown $\Iscr$. Doing so results in the confidence statement that with probability at least $1-\alpha$, 
\begin{align} \label{eq: tau basic CI}
\tau^{\textup{basic}} \geq \max\left\{ \min_{\phi \in \{0,1\}^N} \hat{\tau}_1^\Haj(\phi) - z_{1-\frac{\alpha}{2}} \sqrt{\hat{V}_1(\phi)},\,  \min_{\phi \in \{0,1\}^N} \hat{\tau}_2^\Haj(\phi) - z_{1-\frac{\alpha}{2}} \sqrt{\hat{V}_2(\phi)}\,\right\},
\end{align}
where $\hat{\tau}_k^\Haj(\phi)$ and $\hat{V}_k(\phi)$ denote $\hat{\tau}_k^\Haj$ and $\hat{V}_k$ evaluated under the hypothesis that $\Iscr = \{i: \phi_i=1\}$ for $\phi \in \{0,1\}^N$. 

To complete our discussion of interval estimation, we next give an expression for $\hat{V}_k$, and then discuss computation.

\paragraph{Variance Estimation} 
The expression for $\hat{\tau}_k^\Haj$ given by \eqref{eq: tau hat hajek} is equivalent to
\begin{equation*}
\hat{\tau}_k^{\Haj} = N - \sum_{i\in \Iscr} \frac{N}{\hat{N}_{Y_{ki}}}\cdot \frac{1\{Z_i=Y_{ki}\}}{P(Z_i = Y_{ki})}
\end{equation*}
It follows that the variance of $\hat{\tau}_k^\Haj$ can be written as a sum of covariance terms,
\begin{align} \label{eq: var Phi}
\Var \hat{\tau}_k^\Haj = \sum_{i \in \Iscr} \sum_{j \in \Iscr} \Cov \left(\frac{N}{\hat{N}_{Y_{ki}}} \cdot \frac{1\{Z_i = Y_{ki}\}}{P(Z_i = Y_{ki})} \ , \ \frac{N}{\hat{N}_{Y_{kj}}} \cdot \frac{1\{Z_j = Y_{kj}\}}{P(Z_j = Y_{kj})} \right)
\end{align}
To construct the one-sided interval \eqref{eq: tau basic CI}, we propose $\hat{V}_k$ given by
\begin{align} \label{eq: V}
\hat{V}_k = \sum_{i \in \Iscr} \sum_{j \in \Iscr} \Cov \left(\frac{N}{\hat{N}_{Y_{ki}}} \frac{1\{Z_i = Y_{ki}\}}{P(Z_i = Y_{ki})} \ , \ \frac{N}{\hat{N}_{Y_{kj}}} \frac{1\{Z_j = Y_{kj}\}}{P(Z_j = Y_{kj})} \right) \cdot \frac{1\{Z_i = Y_{ki}, Z_j =Y_{kj}\}}{P(Z_i = Y_{ki}, Z_j= Y_{kj})}
\end{align}
Comparing \eqref{eq: var Phi} and \eqref{eq: V}, it can be seen that $\hat{V}_k$ is a Horvitz-Thompson estimator which samples the covariance terms appearing in $\operatorname{Var}\hat{\tau}_k^\Haj$. Hence it is unbiased provided that the sampling probabilities $P(Z_i = Y_{ki}, Z_j= Y_{kj})$ are non-zero  for all $i,j \in \Iscr$. Should this condition not hold due to dependencies between the entries of $Z$, our variance estimators will be conservatively biased, as will be established in Theorem \ref{th: hajek variance estimation bernoulli}. 

We use the estimated variance $\hat{V}_k$, rather than the actual variance given by \eqref{eq: var Phi}, in the 1-sided interval given by \eqref{eq: tau basic CI}. This is because \eqref{eq: tau basic CI} involves a minimization over all possible hypotheses for the unknown set $\Iscr$, and $\hat{V}_k$ can be less sensitive than the actual variance to this optimization, resulting in a less conservative interval. This is particularly the case when the contrast $\Delta^Z$ is large, as then either $\hat{V}_1$ or $\hat{V}_2$ involves only a small number of sampled covariance terms.

To evaluate $\hat{V}_k$ given by \eqref{eq: V}, we may use the following expressions for the covariance and propensity values which hold for all $i,j \in \Iscr$,
\begin{multline} \label{eq: Cij}
\Cov \left(\frac{N}{\hat{N}_{Y_{ki}}} \cdot \frac{1\{Z_i = Y_{ki}\}}{P(Z_i = Y_{ki})} \ , \ \frac{N}{\hat{N}_{Y_{kj}}} \cdot \frac{1\{Z_j = Y_{kj}\}}{P(Z_j = Y_{kj})} \right) = \\ 
\sum_{a=0}^1 \sum_{b=0}^1 1\{Y_{ki}=a, Y_{kj}=b\}\  \Cov \left(\frac{N}{\hat{N}_{a}} \cdot \frac{1\{Z_i = a\}}{P(Z_i = a)} \ , \ \frac{N}{\hat{N}_{b}} \cdot \frac{1\{Z_j = b\}}{P(Z_j = b)} \right)
\end{multline}
and
\begin{align} \label{eq: Pij}
P(Z_i = Y_{ki}, Z_j= Y_{kj}) =  \sum_{a=0}^1 \sum_{b=0}^1 1\{Y_{ki}=a, Y_{kj}=b\} \ P(Z_i = a, Z_j=b).  
\end{align}
The expressions \eqref{eq: Cij} and \eqref{eq: Pij} hold because $Y_{ki}$ and $Y_{kj}$ are constants when $i,j \in\Iscr$. 






\subsection{Computation} \label{sec: computation}

To construct the one-sided interval \eqref{eq: tau basic CI}, we must solve the optimization problem 
\begin{align}\label{eq: CI optimization}
\min_{\phi \in \{0,1\}^N} \hat{\tau}_k^\Haj(\phi) - z_{1 - \frac{\alpha}{2}}\sqrt{\hat{V}_k(\phi)},
\end{align}
for $k = 1,2$. We can write this problem as
\begin{align} \label{eq: quadprog}
\min_{\phi \in \{0,1\}^N} N - b^T \phi - z_{1-\frac{\alpha}{2}} \sqrt{\phi^T Q \phi}
\end{align}
where $\hat{\tau}_k^\Haj(\phi) = N - b^T \phi$ and $\hat{V}_k(\phi) = \phi^T Q \phi$ for $b \in \mathbb{R}^N$ and $Q \in \mathbb{R}^{N \times N}$ given by 
\begin{align} 
\label{eq: v quadprog} b_i & = \frac{N}{\hat{N}_{Y_{ki}}} \cdot \frac{1\{Z_i = Y_{ki}\}}{Y_{ki} \cdot P(Z_i=1) + (1-Y_{ki}) \cdot P(Z_i=0)} & i \in [N] \\
\label{eq: Q quadprog} Q_{ij} &= C_{ij} \cdot \frac{1\{Z_i=Y_{ki}, Z_j = Y_{kj}\}}{P_{ij}} & i,j \in [N]^2.
\end{align}
in which $C_{ij}$ and $P_{ij}$ appearing in \eqref{eq: Q quadprog} are given by the right hand sides of \eqref{eq: Cij} and \eqref{eq: Pij} respectively

To solve \eqref{eq: quadprog}, we let $q_i$ denote the $i$th row of $Q$, let $B = \sum_{ij} |Q_{ij}|$, and reformulate the problem as 
\begin{align*}
\max_{\phi, u} & \ b^T \phi + z_{1-\frac{\alpha}{2}}\sqrt{\sum_{i=1}^N u_i} \\
\textup{subject to} & \ u_i \leq q_i^T \phi + B (1-\phi_i)  & \forall\ i \in [N]\\
& \ u_i \geq q_i^T \phi - B (1-\phi_i)  & \forall\ i \in [N]\\
& \ u_i \leq B \phi_i  & \forall\ i \in [N]\\
& \ u_i \geq -B \phi_i  & \forall\ i \in [N]\\
& \ \phi \in \{0,1\}^N \\
& \ u \in \mathbb{R}^N
\end{align*}
It can be seen that $u_i=0$ if $\phi_i=0$, and $u_i = q_i^T\phi$ if $\phi_i=1$, enforcing that $\sum_i u_i = \phi^T Q \phi$. We solve this binary program using Gurobi, a commercial optimization solver. 
For the instances that we considered, involving 1000 or more units, Gurobi was generally able to solve the problem to acceptable tolerances. If stopped early, Gurobi returns an upper bound, retaining the coverage properties of the actual solution.

\section{Generalization to other estimands}  \label{sec: generalization}

We now generalize the approach of Section \ref{sec: basic} to cover estimands that more directly characterize interference, such as those described in the introduction. In Section \ref{sec: other estimands} we introduce the estimands $\tau^{\textup{indirect}}$, $\tau^{\textup{nonneighbors}}$, $\tau^{\textup{control}}$, and $\tau^{\textup{tr}}$, defining $\Iscr$ and Horvitz-Thompson estimators for each one. In Section \ref{sec: general template} we then give a general formulation that covers each of these examples, and define the H\'{a}jek-normalized estimators $\hat{\tau}_1^\Haj$ and $\hat{\tau}_2^\Haj$ and the H\'{a}jek-normalized contrast $\Delta^\Haj$. Section \ref{sec: general variance estimation} then presents variance estimators and inference.

\subsection{Other Estimands} \label{sec: other estimands}

\paragraph{The number of units affected by the treatment of others ($\tau^{\textup{indirect}}$)} \label{sec: theta}

Let $\Iscr \subseteq [N]$ identify the units who are not affected by the treatment of others
\begin{align} \label{eq: phi indirect}
\Iscr = \{i: \textup{$Y_i$ is affected only by $X_i$, so that $Y_i = f_i(X_i)$}\} 
\end{align}
and let $\tau^{\textup{indirect}} = N - |\Iscr|$, the number of units affected by the treatment of others. 

To estimate $\tau^{\textup{indirect}}$, we may use estimators $\hat{\tau}_1$ and $\hat{\tau}_2$ given by
\begin{align*} 
\hat{\tau}_k & = N - \sum_{i \in \Iscr} \frac{1\{W_i=Y_{ki}\}}{{P}(W_i=Y_{ki}|X_i)}. 
\end{align*}
This estimator differs from \eqref{eq: tau hat k} in that we condition on $X_i$ in the probability weight $P(W_i = Y_{ki}|X_i)$. This is done because the outcome $Y_{ki}$ is no longer constant for $i \in \Iscr$, but instead is a function of $X_i$. As a result, the probability $P(W_i = Y_{ki})$ is unknown if $f_i$ is unknown, even if $i \in \Iscr$. On the other hand, the conditional probability $P(W_i = Y_{ki}|X_i)$ may be evaluated without knowledge of $f_i$, according to
\[ P(W_i = Y_{ki} | X_i) = Y_{ki}\cdot P(W_i = 1|X_i) + (1-Y_{ki})\cdot P(W_i = 0|X_i), \qquad i \in \Iscr \]
as will be shown by Proposition \ref{th: Phi evaluation} in Section \ref{sec: general template}.


\paragraph{Long-range interference ($\tau^{\textup{nonneighbors}}$)} \label{sec: theta eta}

Let $\Iscr \subseteq [N]$ identify the units who are affected only by their own treatment and that of their close associates,
\begin{align*} 
\Iscr = \{i: \textup{$Y_i$ is affected only by $X_{{\mathcal N}_i}$, so that $Y_i = f_i(X_{{\mathcal N}_i})$}\} 
\end{align*}
and let $\tau^{\textup{nonneighbors}}$ equal $N - |\Iscr|$, the number of units affected by the treatment of others excluding close associates. To estimate $\tau^{\textup{nonneighbors}}$, we may use $\hat{\tau}_1$ and $\hat{\tau}_2$ given by
\begin{align*} 
\hat{\tau}_k & = N - \sum_{i \in \Iscr} \frac{1\{W_i^{(2)}=Y_{ki}\}}{{P}(W_i^{(2)}=Y_{ki}|X_{{\mathcal N}_i})} 
\end{align*}
where $W_i^{(2)}$ is given by \eqref{eq: W2}, and measures treatment of non-close associates.

\paragraph{The number of units whose outcome under direct control (treatment) is affected by the treatment of others ($\tau^{\textup{control}}$, $\tau^{\textup{tr}}$)} \label{sec: theta tr}

Let $\Iscr \subseteq [N]$ identify the units whose outcome, when receiving control themselves, does not depend on the treatment of others
\begin{align*} 
\Iscr = \{i: \textup{$Y_i$ is constant for all $X$ such that $X_i = 0$}\}
\end{align*}
in which case $\tau^{\textup{control}} = N - |\Iscr|$ equals the number of units whose outcome under direct control is affected by the treatment of others. For example, if a vaccine grants immunity, then the disease status of a unit when receiving the placebo (but not necessarily when receiving the vaccine) might depend on the vaccination status of others. 

To estimate $\tau^{\textup{control}}$, we consider $\hat{\tau}_1$ and $\hat{\tau}_2$ given by
\begin{align*} 
\hat{\tau}_k & = N - \sum_{i \in \Iscr} \frac{1\{X_i=0, W_i=Y_{ki}\}}{{P}(X_i=0, W_i=Y_{ki})}
\end{align*}
We may similarly define $\tau^{\textup{tr}}$, the number of units whose outcome under direct treatment is affected by the treatment of others, by defining $\Iscr$ to be
\begin{align*} 
\Iscr & = \{i: \textup{$Y_i$ is constant for all $X$ such that $X_i = 1$}\} 
\end{align*}
and $\hat{\tau}_1$ and $\hat{\tau}_2$ to be
\begin{align*} 
\hat{\tau}_k & = N - \sum_{i \in \Iscr} \frac{1\{X_i=1, W_i=Y_{ki}\}}{{P}(X_i=1, W_i=Y_{ki})} 
\end{align*}

\subsection{General Functional Form} \label{sec: general template}

We give a general functional form for our estimands and estimators. We may write our estimands as 
\begin{align} \label{eq: general estimand}
 \tau = N - |\Iscr|,
\end{align}
where $\Iscr \subseteq [N]$ contains the units whose outcome mapping satisfies an interference specification, which may be written as
\begin{align} \label{eq: phi general}
\Iscr = \{i: \textup{ $\exists\ \tilde{f}_i$ such that $f_i(X) = \tilde{f}_i(T_i(X))$ for all $X \in R_i$} \}
\end{align}
where $T_i \equiv T_i(X)$ equals a vector of summary statistics describing unit $i$'s exposure to treatment, and $R_i$ is a subset of $\{0,1\}^N$. For units in $\Iscr$, the value of $T_i$ dictates unit $i$'s outcome whenever $X \in R_i$. Both $T_i$ and $R_i$ are chosen by the researcher to define the specification of interest. For example, letting $T_i = X_i$ and $R_i = \{0,1\}^N$ induces $\Iscr$ given by \eqref{eq: phi indirect}. 

In addition to $T_i$ and $R_i$, the researcher also chooses $Z_i \equiv Z_i(X) \in \{1,0,-1\}$, which maps unit $i$'s treatment exposure to one of 3 groups (i.e., effect treatment, effective control, or not used). Table \ref{table: general} gives values for $\{T_i, R_i, Z_i\}_{i=1}^N$ corresponding to the estimands considered in this paper. 

To simplify notation, let $\bar{Z}_i$ for $i\in [N]$ equal $Z_i$ if $X \in R_i$, and equal -1 otherwise,
\begin{align}\label{eq: Z bar}
\bar{Z}_i = \begin{cases} Z_i & \textup{ if $X \in R_i$} \\ -1 & \textup{ if $X \notin R_i$,} \end{cases}
\end{align}
allowing the event $\{X \in R_i, Z_i = Y_{ki}\}$ to be written more concisely as $\{\bar{Z}_i = Y_{ki}\}$ under binary outcomes.

The Horvitz-Thompson estimators $\hat{\tau}_k$ for $k=1,2$ are given by
\begin{align}\label{eq: Phi general}
\hat{\tau}_k  = N - \sum_{i\in \Iscr} \frac{1\{\bar{Z}_i = Y_{ki}\}}{P(\bar{Z}_i  = Y_{ki}|T_i)}
\end{align}
and H\'{a}jek-normalized versions of $\hat{\tau}_k$ are given by
\begin{align} \label{eq: hajek general}
\hat{\tau}_k^\Haj = N - \sum_{i \in \Iscr} \frac{N}{\hat{N}_{Y_{ki}}} \cdot \frac{1\{\bar{Z}_i = Y_{ki}\}}{P(\bar{Z}_i = Y_{ki}|T_i)},
\end{align}
where $\hat{N}_1$ and $\hat{N}_0$ denote H\'{a}jek normalization factors
\begin{align} \label{eq: Nhat general}
\hat{N}_{a} &= \sum_{i=1}^N \frac{1\{\bar{Z}_i = a\}}{P(\bar{Z}_i = a|T_i)}, \qquad a\in \{0,1\}.
\end{align}
As will be shown in Theorem \ref{th: hajek consistency bernoulli}, the H\'{a}jek-weighted estimates will be lower bounded by the absolute value of a H\'{a}jek-weighted comparison of outcomes given by 
\begin{align} \label{eq: contrast general}
 \Delta^{\Haj} = \sum_{i=1}^N\left(\frac{N}{\hat{N}_1} \frac{1\{\bar{Z}_i = 1\}}{ P(\bar{Z}_i = 1| T_i)} - \frac{N}{\hat{N}_0} \frac{1\{\bar{Z}_i = 0\}}{P(\bar{Z}_i = 0|T_i)} \right) Y_i
\end{align}
For the estimands considered in this paper, Table \ref{table: tau hat} gives expressions for $\Delta^{\Haj}$.

\begin{table}[t!]
\centering
 \begin{tabular}{l|l|l|l} 
 Estimand & $T_i$ & $R_i$ & $Z_i$ \\ 
 \hline
 $\tau^{\textup{basic}}$            & 1       & $\{0,1\}^N$ & $X_i$ \\ 
 $\tau^{\textup{basic}}$ using \eqref{eq: Z basic}      & 1       & $\{0,1\}^N$ & $1\{X_i=W_i=1\} - 1\{X_i \neq W_i\}$ \\
 $\tau^{\textup{indirect}}$   & $X_i$     & $\{0,1\}^N$ & $W_i$ \\
 $\tau^{\textup{nonneighbors}}$ & $X_{\mathcal{N}_i}$  & $\{0,1\}^N$ & $W_i^{(2)}$ \\
 $\tau^{\textup{control}} $   & 1       & $\{X \in \{0,1\}^N: X_i=0\}$ & $W_i$ \\ 
 $\tau^{\textup{tr}}$     & 1      & $\{X \in \{0,1\}^N: X_i=1\}$ & $W_i$
 \end{tabular}
 \caption{Exposure mappings, subsets, and contrast mappings for various estimands. Note that letting $T_i = 1$ for all $X$ enforces $Y_i$ to equal a constant for all $X \in R_i$, if $i \in \Iscr$. }\label{table: general}
\end{table}

\begin{table}[t!]
\centering
 \begin{tabular}{l||l} 
 Estimand & H\'{a}jek-normalized contrast $\Delta^\Haj$ \\ 
 \hline
 $\tau^{\textup{basic}}$    
 & $\displaystyle{\sum_{i=1}^N \left( \frac{N}{\hat{N}_1}\frac{X_i}{P(X_i=1)} - \frac{N}{\hat{N}_0} \frac{1-X_i}{P(X_i=0)}\right) Y_i }$
          \\ &  \\
 $\tau^{\textup{basic}}$ (w/network) 
 & $\displaystyle{ \sum_{i=1}^N \left( \frac{N}{\hat{N}_1}\frac{1\{X_i=W_i=1\}}{P(X_i=W_i=1)} - \frac{N}{\hat{N}_0} \frac{1\{X_i=W_i=0\}}{P(X_i=W_i=0)} \right) Y_i }$ 
        \\  &  \\
 $\tau^{\textup{indirect}}$     
 & $\displaystyle{ \sum_{i=1}^N \left( \frac{N}{\hat{N}_1}\frac{1\{W_i=1\}}{P(W_i=1|X_i)} - \frac{N}{\hat{N}_0} \frac{1\{W_i=0\}}{P(W_i=0|X_i)}\right) Y_i }$ 
      \\ &  \\
 $\tau^{\textup{nonneighbors}}$ & 
 $\displaystyle{ \sum_{i=1}^N \left( \frac{N}{\hat{N}_1}\frac{1\{W_i^{(2)}=1\}}{P(W_i^{(2)}=1|X_{{\mathcal N}_i})} - \frac{N}{\hat{N}_0} \frac{1\{W_i^{(2)}=0\}}{P(W_i^{(2)}=0|X_{{\mathcal N}_i})}\right) Y_i }$
        \\  & \\
  $\tau^{\textup{control}} $
   & $\displaystyle{ \sum_{i=1}^N \left( \frac{N}{\hat{N}_1}\frac{1\{X_i=0, W_i=1\}}{P(X_i=0,W_i=1)} - \frac{N}{\hat{N}_0} \frac{1\{X_i=0, W_i=0\}}{P(X_i=0,W_i=0)}\right) Y_i }$ 
          \\ &   \\
 $\tau^{\textup{tr}}$   
  & $\displaystyle{ \sum_{i=1}^N \left( \frac{N}{\hat{N}_1}\frac{1\{X_i=1, W_i=1\}}{P(X_i=1,W_i=1)} - \frac{N}{\hat{N}_0} \frac{1\{X_i=1, W_i=0\}}{P(X_i=1,W_i=0)}\right) Y_i}$ 
              
 \end{tabular}
 \caption{H\'{a}jek weighted point estimate $\hat{\tau}$, expressed as contrasts given by \eqref{eq: contrast general}. The terms $\hat{N}_1$ and $\hat{N}_0$ are defined by \eqref{eq: Nhat general}. }\label{table: tau hat}
\end{table}

Analogous to the discussion of Section \ref{sec: basic}, Proposition \ref{th: tight} gives intuition for when our lower bounds may be tight, and suggests guidance for the choice of $\bar{Z}$. It states that if for all $i \notin \Iscr$ the contrast variable $\bar{Z}_i$ is a perfect predictor of the outcome $Y_i$ whenever $\bar{Z}_i \in \{0,1\}$,  then an IPW weighted contrast using $\bar{Z}_i$ will be unbiased for the estimand. 

\begin{proposition} \label{th: tight}
Let $\Delta^{\bar{Z}}$ equal
\[ \Delta^{\bar{Z}} = \sum_{i=1}^N \left(\frac{1\{\bar{Z}_i=1\}}{P(\bar{Z}_i = 1|T_i)}  - \frac{1\{\bar{Z}_i = 0\}}{P(\bar{Z}_i = 0|T_i)}\right)Y_i \]
and for all $i \notin \Iscr$ let $Y_i = \bar{Z}_i$ whenever $\bar{Z}_i \neq -1$. Then $\mathbb{E}[\Delta^{\bar{Z}}] = N - |\Iscr|.$
\end{proposition}

\subsection{Variance Estimation and Inference for General Case} \label{sec: general variance estimation}

Let $v_i$ denote 
\begin{align} \label{eq: v general}
v_i = \frac{N}{\hat{N}_{Y_{ki}}} \cdot \frac{1\{\bar{Z}_i = Y_{ki}\}}{P(\bar{Z}_i = Y_{ki}|T_i)}
\end{align}
so that $\hat{\tau}_k^\Haj = N - \sum_{i\in \Iscr} v_i$. It follows that the variance of $\hat{\tau}_k^\Haj$ can be written as a sum of covariances
\[ \Var \hat{\tau}_k^\Haj = \sum_{i \in \Iscr} \sum_{j \in \Iscr} \Cov(v_i, v_j) \]
Unlike the special case where $T_i$ is constant for all units, generally we cannot compute $\Cov(v_i, v_j)$ exactly without knowledge of the outcome models $f_i$ and $f_j$, even if $i,j \in \Iscr$. As a result, the covariance terms must be estimated. We thus propose variance estimator $\hat{V}_k$ given by
\begin{multline} \label{eq: V general}
\hat{V}_k = \sum_{i \in \Iscr}\sum_{j \in \Iscr}\Bigg\{ \mathbb{E}\left[v_i v_j|T_i, T_j\right] - 
\mathbb{E}\left[v_i|T_i\right]\cdot \mathbb{E}\left[v_j|T_j\right]  \cdot \\ \frac{P(T_i)\cdot P(T_j)}{P(T_i, T_j)} \Bigg\} \cdot 
\frac{1\{\bar{Z}_i = Y_{ki}, \bar{Z}_j = Y_{kj}\}}{P(\bar{Z}_i = Y_{ki}, \bar{Z}_j = Y_{kj}|T_i, T_j)}
\end{multline}
We will show in Theorem \ref{th: hajek variance estimation bernoulli} that \eqref{eq: V general} is conservatively biased, with the exception that if $T_i$ is a constant for each unit (as is the case for $\tau^{\textup{basic}}$, $\tau^{\textup{control}}$, and $\tau^{\textup{tr}}$), then the estimator simplifies to a weighted sample of covariance terms 
\begin{align*}
\hat{V}_k = \sum_{i \in \Iscr} \sum_{j \in \Iscr} \Cov(v_i, v_j) \cdot 
 \frac{1\{\bar{Z}_i = Y_{ki}, \bar{Z}_j = Y_{kj}\}}{P(\bar{Z}_i = Y_{ki}, \bar{Z}_j = Y_{kj})},
\end{align*}
and is unbiased if the sampling probabilities $P(\bar{Z}_i = Y_{ki}, \bar{Z}_j = Y_{kj})$ are non-zero, and conservative otherwise. 

If $\hat{\tau}_k^\Haj$ for $k=1,2$ are consistent and asymptotically normal, with variance consistently (or conservatively) estimated by $\hat{V}_k$, then we can construct an asymptotic $(1-\alpha)$ level  one-sided confidence interval for $\tau$, with lower boundary given by
\begin{align} \label{eq: CI general}
\max\left\{ \min_{\phi \in \{0,1\}^N} \hat{\tau}_1^\Haj(\phi) - z_{1-\frac{\alpha}{2}}\sqrt{\hat{V}_1(\phi)}\, , \, \min_{\phi \in \{0,1\}^N} \hat{\tau}_2^\Haj(\phi) - z_{1-\frac{\alpha}{2}} \sqrt{\hat{V}_2(\phi)}\, \right\},
\end{align}
where $\hat{\tau}_k^\Haj(\phi)$ and $\hat{V}_k(\phi)$ denote $\hat{\tau}_k^\Haj$ and $\hat{V}_k$, evaluated under the hypothesis that $\Iscr = \{i: \phi_i=1\}$.

To evaluate the terms $\hat{\tau}_k^\Haj(\phi)$ and $\hat{V}_k(\phi)$ appearing in \eqref{eq: CI general}, we may use Propositions \ref{th: Phi evaluation} and \ref{th: V evaluation}. Proposition \ref{th: Phi evaluation} shows that the sampling weights $P(\bar{Z}_i  = Y_{ki}|T_i)$ appearing in $\hat{\tau}_k^\Haj$ given by \eqref{eq: hajek general} do not require knowledge of the outcome mappings $\{f_i\}_{i=1}^N$. Proposition \ref{th: V evaluation} gives computable expressions for the terms appearing in $\hat{V}_k$ given by \eqref{eq: V general}. 
\begin{proposition}\label{th: Phi evaluation}
If $i \in \Iscr$ and $X \in R_i$, then it holds that
\[ P(\bar{Z}_i = Y_{ki}|T_i) = Y_{ki}\cdot P(\bar{Z}_i = 1|T_i) + (1-Y_{ki}) \cdot P(\bar{Z}_i=0|T_i). \]
\end{proposition}

\begin{proposition}\label{th: V evaluation}
If $i \in \Iscr$ and $\bar{Z}_i = Y_{ki}$, then it holds that
\begin{align*}
\mathbb{E}[v_i|T_i] & = \sum_{a=0}^1 1\{Y_{ki}=a\}\mathbb{E}\left[  \frac{1\{\bar{Z}_i = a\} N}{P(\bar{Z}_i=a|T_i) \hat{N}_a}  | T_i\right], 
\end{align*}
If $i,j \in \Iscr$ and $1\{\bar{Z}_i = Y_{ki}, \bar{Z}_j = Y_{kj}\} = 1$, then it holds that
\begin{align*}
\mathbb{E}[v_i v_j|T_i,T_j] =& \sum_{a,b=0}^1 1\{Y_{ki}=a, Y_{kj}=b\}\mathbb{E}\left[ \frac{1\{\bar{Z}_i = a\} N }{P(\bar{Z}_i=a|T_i) \hat{N}_a}  \cdot 
\frac{1\{\bar{Z}_j = b\} N}{P(\bar{Z}_j=b|T_j) \hat{N}_b}  | T_i, T_j\right] 
\end{align*}
and also that
\begin{align*}
P(\bar{Z}_i = Y_{ki}, \bar{Z}_j=Y_{kj}| T_i, T_j) =&  \sum_{a,b=0}^1 1\{Y_{ki}=a, Y_{kj}=b\}P(\bar{Z}_i = a, \bar{Z}_j=b | T_i, T_j)
\end{align*}
\end{proposition}

\section{Theory} \label{sec: theory}

In this section, we give consistency and coverage results for the general class of H\'{a}jek-normalized estimators described in Section \ref{sec: generalization}. Specifically, we show that $|\Delta^\Haj|$ asymptotically lower bounds the estimand; we show that the variance estimator $\hat{V}_k$ concentrates and is either unbiased or conservative; and we show that our intervals have the desired asymptotic coverage, albeit with a technical modification to cover settings approaching degeneracy, in which case normality may not hold.


To describe our requirements on the experiment design, for $i \in [N]$ let $\Gamma_i \subseteq [N]$ denote a subset of units whose treatment assignment is sufficient to determine the values of $T_i$ and $\bar{Z}_i$ for all $X \in \{0,1\}^N$, so that for all $X, X' \in \{0,1\}^N$ satisfying $X_{\Gamma_i} = X_{\Gamma_i}'$, it holds that
\[ T_i(X) = T_i(X'), \textup{ and }  \bar{Z}_i(X) = \bar{Z}_i(X')\]
 We will call $\Gamma_i$ the design neighborhood of unit $i$. We stress that requirements on $\{\Gamma_i\}_{i=1}^N$ restrict the experiment designs and estimands which may be explored, but places no assumptions on the outcome models $\{f_i\}_{i=1}^N$ that describe interference, which are unknown and  uncontrolled by the researcher.

Assumption \ref{as: design bernoulli} describes a simple experiment design that we study. It assumes a sequence of experiments in which the design neighborhoods are bounded in size and overlap, Bernoulli randomization with treatment probabilities bounded away from 0 and 1, and places a positivity condition on $P(\bar{Z}_i|T_i)$. 

\begin{assumption} \label{as: design bernoulli}
The treatment assignments $\{X_i\}_{i=1}^N$ are independent Bernoulli random variables, which satisfy for fixed constants $p > 0$ and $d > 0$ as $N \rightarrow \infty$ that
\begin{align}
 |\Gamma_i| & \leq d & \forall\ i \in [N] \label{eq: Gamma 1}\\
\sum_{j=1}^N 1\{i \in \Gamma_j\} & \leq d  & \forall\ i \in [N]  \label{eq: Gamma 2} \\
p \leq P(X_i = 1) & \leq 1-p & \forall\ i \in  [N] \label{eq: pmin X} \\
\label{eq: pmin}
 P(\bar{Z}_i = a \,|\, T_i=t) & \geq p  & \forall\ i \in [N], a \in \{0,1\}, t \in \mathcal{T}_i
\end{align}
where $\mathcal{T}_i$ denotes the support of $T_i$. 
\end{assumption}

The positivity condition \eqref{eq: pmin} requires $T_i$, $Z_i$, and $R_i$ to be chosen so that $Z_i$ and $R_i$ are not deterministic conditional on $T_i$. For example, if $T_i = X_{\mathcal{N}_i}$, then $Z_i$ and $R_i$ should depend at least partially on units outside of $\mathcal{N}_i$.

Theorem \ref{th: hajek consistency bernoulli} is a consistency result. It states that $\hat{\tau}^\Haj_k$ converges to $\tau$, and that the absolute value of $\Delta^\Haj$, the H\'{a}jek weighted contrast between units given by \eqref{eq: contrast general}, is asymptotically a lower bound for $\tau$. 

\begin{theorem} \label{th: hajek consistency bernoulli}
It holds that
\begin{align} 
 |\Delta^\Haj| & \leq \max(\hat{\tau}_1^\Haj, \hat{\tau}_2^\Haj)  \label{eq: hajek consistency contrast bernoulli} 
 \end{align}
If Assumption \ref{as: design bernoulli} holds, then 
\begin{align}
 \hat{\tau}_k^\Haj = \tau + O_P\left(N^{1/2}\right), \qquad k=1,2  \label{eq: hajek consistency Phi bernoulli}
\end{align}
in which case \eqref{eq: hajek consistency contrast bernoulli} implies $|\Delta^\Haj| \leq \tau + O_P\left(N^{1/2}\right)$

\end{theorem}

Theorem \ref{th: hajek variance estimation bernoulli} pertains to variance estimation. It states that $\hat{V}_k$ concentrates to its expectation which upper bounds the variance of $\hat{\tau}_k^\Haj$, and gives conditions under which the bound will be tight.

\begin{theorem} \label{th: hajek variance estimation bernoulli}
Let Assumption \ref{as: design bernoulli} hold. It holds that
\begin{align} \label{eq: hajek variance concentration bernoulli}
\hat{V}_k = \mathbb{E}[\hat{V}_k] + O_P\left(N^{1/2} \log N\right),\qquad k=1,2.
\end{align}
Additionally, it holds that 
\begin{align} \label{eq: hajek variance bias bernoulli}
\mathbb{E} \hat{V}_k \geq \Var \hat{\tau}_k^\Haj,
\end{align} 
with equality if all of the following conditions hold: 
\begin{enumerate}
  \item The exposure mappings $\{T_i\}_{i=1}^N$ are constant
  \item $P(\bar{Z}_i=a) > 0$ for all $i \in [N]$ and $a \in \{0,1\}$
  \item $P(\bar{Z}_i=a, \bar{Z}_j = b) > 0 $ for all $i,j \in [N]$ such that $i \neq j$, and $a,b \in \{0,1\}$.
\end{enumerate}
\end{theorem}

 Theorem \ref{th: hajek coverage bernoulli} gives coverage properties for the proposed confidence interval given by \eqref{eq: CI general}. It states that if the variance estimators $\hat{V}_k$ are lower bounded by a vanishing fraction of $N$, then the interval will have at least the desired asymptotic coverage level.

\begin{theorem} \label{th: hajek coverage bernoulli}
Let Assumption \ref{as: design bernoulli} hold. For any $\epsilon > 0$, let $\widetilde{V}_k(\phi)$ denote the thresholded variance estimate given by
\begin{align}\label{eq: V tilde bernoulli}
\widetilde{V}_k(\phi) = \max\left\{ \hat{V}_k(\phi)\ , \ \frac{N^{2/3+\epsilon}}{z_{1-\frac{\alpha}{2}}^2\cdot \alpha/2}  \right\} 
\end{align}
 and let $LB_{1-\alpha}$ denote
 \begin{align} \label{eq: LB}
LB_{1-\alpha} = \max\left\{ \min_{\phi \in \{0,1\}^N}  \hat{\tau}_1^\Haj(\phi) - z_{1-\frac{\alpha}{2}}\sqrt{\widetilde{V}_1(\phi)} \, , \, \min_{\phi \in \{0,1\}^N}  \hat{\tau}_2^\Haj(\phi) - z_{1-\frac{\alpha}{2}}\sqrt{\widetilde{V}_2(\phi)} \right\}
\end{align}
Then $\tau \geq LB_{1-\alpha}$ with probability converging to at least $1-\alpha$.
\end{theorem} 

In practice, we expect usage of $\tilde{V}_k(\phi)$ rather than $\hat{V}_k(\phi)$ to make little or no difference, as the two terms are equal except in cases involving interval half-widths of order $O(N^{1/3})$ or smaller. This is a vanishing fraction of the usual $O(\sqrt{N})$ rate that arises for estimation of population sizes under simple random sampling, and we do not expect the interval widths found by \eqref{eq: LB}, which are widened by the optimization over $\phi$, to be so small in practice.

\paragraph{Additional Experiment Designs} We have considered independent Bernoulli randomization primarily as a starting point for its simplicity. In the supplemental materials, we generalize the theorems to include two-stage experiment designs, in which the size of the design neighborhoods may grow with $N$. The result will allow contrasts (i.e., $\bar{Z}$) involving not only the treatment of each unit and a bounded number of close friends, but also the number of treated units in some broader division such as geographic region, or overlapping collections of neighboring geographic regions, in which the region sizes may grow at rate $O(N^{1-\delta})$ for any $\delta>0$. 



\section{Examples} \label{sec: examples}

\subsection{Simulation Study} \label{sec: simulations}

Table \ref{table: simulations} shows the results of a simulation study in which the outcomes depended on both direct and indirect treatment. 
For all estimands except $\tau^{\textup{nonneighbors}}$, outcomes were generated according to 
\begin{align} \label{eq: sim model}
 Y_i = 1\{ \textup{if } u_i \leq \alpha_{0i} + \alpha_{1i}\cdot X_i + \alpha_{2i} \cdot W_i\}
\end{align}
where the scalar parameters $u_i, \alpha_{0i}, \alpha_{1i}, \alpha_{2i}$ were randomly generated, with $u_i\sim \textup{Unif }[0,1]$, and $\alpha_{0i}, \alpha_{1i},\alpha_{2i}$ nonnegative. For simulations involving $\tau^{\textup{nonneighbors}}$, in \eqref{eq: sim model} we change $W_i$ to $\bar{W}_i$, the thresholded number of treated units in ${\mathcal N}_i \cup {\mathcal N}_i^{(2)}$ excluding unit $i$, so that units in ${\mathcal N}_i$ and ${\mathcal N}_i^{(2)}$ were equally important in determining the outcome of unit $i$. To compute the true value of the estimand, unit $i$ was considered to be affected by treatment if $Y_i$ varied as a function of $X$ under the generated parameter values, and was considered to be affected by treatment of others if $Y_i$ varied as a function of $X_{-i}$ for any value of $X_i$, and so forth.

Estimation of $\tau^{\textup{basic}}$ used the contrast variable $Z$ given by \eqref{eq: Z basic}. H\'{a}jek normalization was used on all estimators. For comparison, the treatment effect
\[ EATE = \frac{1}{N}\sum_{i=1}^N \Big( \mathbb{E}[Y_i|\bar{Z}_i = 1] - \mathbb{E}[Y_i|\bar{Z}_i = 0]\Big),\]
was also estimated, where $\bar{Z}$ followed Table \ref{table: general} for each estimand. For example, when estimating $\tau^{\textup{indirect}}$, we compared to the indirect effect given by 
\[\frac{1}{N}\sum_{i=1}^N \Big( \mathbb{E}[Y_i|W_i = 1] -  \mathbb{E}[Y_i|W_i = 0]\Big).\] 
For estimation and inference of these treatment effects, H\'{a}jek normalization and conservative variance estimates were used following \cite{gao2023causal}.\footnote{Variance estimates for EATE followed Eq. (11) (``adjusted HAC covariance estimator'') of \cite{gao2023causal} (arxiv version v3), using our $\bar{Z}_i$ for their treatment condition $T_i$. For the simulations, the correctly specified bandwidth $b_n=2$ was used.}

The results of the simulations are the following. For the estimands $\tau^{\textup{basic}}$, $\tau^{\textup{tr}}$, and $\tau^{\textup{control}}$, point estimates were unbiased (as might be expected given that the unit effects were nonnegative, consistent with the intuition of Example \ref{ex: nonnegative effects}), and the confidence intervals exhibited approximate calibration of coverage. For these estimands, the exposure mapping $T_i$ is constant, allowing for unbiased variance estimation as discussed in Theorem \ref{th: hajek variance estimation bernoulli}. For the remaining estimands, estimation was conservatively biased, and interval coverage rates were 100\%. Table \ref{table: simulations} also shows that the interval half-widths for our estimands were competitive, and often smaller than, those of comparable EATE treatment effects. Coverage of the EATE interval estimates was conservative, always exceeding 99\%. 

Our intuition for the smaller intervals of our approach is the following. Whereas treatment effect estimation uses $\Delta^\Haj$ as a point estimate and therefore must control its variability, we use $\Delta^\Haj$ rather as a lower bound for the point estimators $\hat{\tau}_1^\Haj$ and $\hat{\tau}_2^\Haj$. Thus, to control anticonservativeness we need to consider the variability of these point estimators, which we expect to be less than that of $\Delta^\Haj$ and also easier to estimate, as only units in $\Iscr$ are considered.

\begin{table}[t!]
\centering
 \begin{tabular}{llrrrrr|rrr}
          &   & Actual &  \multicolumn{2}{c}{$|\Delta^\Haj|/N$}   & Coverage & CI Half- &    \multicolumn{2}{c}{$\widehat{EATE}$} & EATE CI \\
Estimand & N & Value$/N$  & Avg & SD                              & of 95\% CI & width & Avg & SD & Half-width \\ 
\hline
$\tau^{\textup{basic}}$  
& 250 & 0.43  &  0.43   & 0.08   & \bf{0.94}   &  \bf{0.16}   & 0.43 & 0.08 & 0.26   \\
& 500 & 0.37  &   0.37  & 0.05   & \bf{0.96}   &   \bf{0.11}  &  0.36 & 0.05 & 0.21 \\
\hline
$\tau^{\textup{tr}}$ 
& 250 & 0.22  &   0.22 &    0.08 &   \bf{0.95}    &     \bf{0.15}  &  0.22 & 0.08 & 0.26  \\
 &  500 & 0.18  &  0.18  &   0.05 &   \bf{0.95}   &      \bf{0.11} &  0.18& 0.05 &  0.20 \\
\hline
$\tau^{\textup{control}}$ 
& 250 & 0.23  &   0.23  &   0.10 &   \bf{0.95}    &    \bf{0.16}  &  0.23 & 0.10 & 0.29  \\
&  500 & 0.21  &  0.21   &  0.06  &   \bf{0.96}   &     \bf{ 0.12} &  0.21 & 0.06 &  0.23 \\
\hline
$\tau^{\textup{indirect}}$ 
& 250& 0.37 &   0.22 &    0.06 &    1.0000   &    \bf{0.15}  &  0.22 & 0.06 & 0.20   \\
& 500 &     0.33 &   0.20  &   0.04  &   1.0000   &      \bf{0.10}  & 0.19 & 0.04 &  0.16 \\ 
\hline
$\tau^{\textup{nonneighbor}}$ 
& 250 & 0.37  &  0.22  &   0.07   &  1.0000    &     0.21  &   0.22 & 0.06 & 0.20    \\
& 500 &   0.33 &   0.20 &    0.04  &   1.0000    &     0.18  &  0.19 & 0.04 & 0.16 
\end{tabular}
 \caption{Results of Simulation Study. Bold values indicate approximate 95\% coverage, or that CI half-widths were smaller than those of comparable treatment effect. CI Half-width defined as average distance from point estimate to the interval boundary. For our estimands, $|\Delta^\Haj|/N$ was point estimate. 500 trials per simulation.}\label{table: simulations}
\end{table}

\subsection{Social Networks and the Decision to Insure}

We analyze an experiment that was originally described in \cite{cai2015social}, and re-analyzed in \cite{gao2023causal}.\footnote{data available at \texttt{https://www.openicpsr.org/openicpsr/project/113593/version/V1/view}} In this experiment, rural farmers were randomized to receive high or low information regarding a crop insurance product, with opportunity to purchase the product at the end of the information session. Before attending their assigned sessions, the farmers could communicate informally with others who had attended earlier sessions (sessions were administered in two rounds, with a three day gap between them), allowing for sharing of information through social contacts.  Social network information was elicited, with the farmers instructed to list 5 close friends with whom they specifically discussed rice production or financial issues. Farmers assigned to a second round low-information session were more likely to purchase insurance if more of their listed friends in the first round were assigned to a high-information session.

Table \ref{table: cai} in the supplement shows results from the original regression analysis of \cite{cai2015social}, which found that for farmers assigned to to low information sessions, each friend assigned to a first round high-information session caused an 8.9\% increase in insurance purchasing rates (95\% CI: [4.5\%, 13.2\%]). In their analysis, it was assumed that a unit's outcome depends only on the information received by themselves and their listed friends. To justify this assumption, a subjective argument was made that the three day gap between information sessions was long enough for information to reach direct friends, but too short for information to diffuse across an entire village. 

In our analysis, we make no assumptions on interference and estimate $\tau^{\textup{control}}$, the number of second round units whose outcome under direct control (i.e., assignment to a low info session) would be affected by the assignment of first round farmers into high and low information sessions. Interference is permitted between distant or even disconnected units; for example, a farmer might ask a question during an information session, affecting the information received by (and subsequently shared by) all farmers at the session. To construct our estimators, we require knowledge of the randomization design, which used stratification and is described in \cite{gao2023causal}. 

Our asymptotic lower bound equals 23\%, which may be interpreted as meaning that at least 23\% of the 880 second round farmers, if assigned to a low information session, would be affected by information given to the first round farmers (1-sided 95\% CI: at least 9\%). This point estimate of 23\% is equal to a H\'{a}jek weighted comparison of the second round units receiving low information, grouped according to their number of friends receiving high information in the previous round:
\begin{align}\label{eq: my cai specification}
 \frac{\Delta^\Haj}{N} = \frac{1}{N} \sum_{i=1}^N \left(\frac{1\{X_i=0, W_i=1\}N}{P(X_i=0,W_i=1)\hat{N}_1} - \frac{1\{X_i=0, W_i=0\}N}{P(X_i=0, W_i=0)\hat{N}_0}\right) Y_i
 \end{align}
Here $i \in [N]$ enumerates the $N = 880$ second round units with at least one first round friend; $X_i=0$ indicates that unit $i$ was assigned to a low information session; $W_i=1$ indicates that all of unit $i$'s first round friends received high information; $Y_i=1$ indicates that unit $i$ purchased insurance; and $\hat{N}_a = \sum_{i=1}^N 1\{X_i=0, W_i=a\}/P(X_i=0, W_i=a)$ for $a \in \{0,1\}$. Additional results can be found in Appendix \ref{sec: supplement examples} of the supplemental materials.

Table \ref{table: gao} in the supplement shows results from \cite{gao2023causal}, who estimate the spillover effects of having friends receive high information. They report weakly significant or non-significant spillover effects of the first round assignments on the second round units, assuming conditions on interference which are technical in nature. It may be difficult to judge whether these conditions are suitable for this application.
In contrast, no assumptions on interference are required for our estimates. For the purposes of demonstrating the presence of social influence, our estimand may be an appropriate target parameter, and has tighter, less questionable CIs. Thus it may serve as a useful complement to prior studies of this experiment.

\paragraph{Supplemental Materials} The supplement contains discussion of future work; further discussion and results for the simulations and data example; and proofs of all propositions and theorems, with generalization of Theorems \ref{th: hajek consistency bernoulli} - \ref{th: hajek coverage bernoulli} to include two-stage randomization.

 \bibliographystyle{apalike}
 {\footnotesize
 \bibliography{bibfile}
 }

\newpage

\bigskip

\begin{center}
{\large\bf Supplementary Material for ``Agnostic Characterization of Interference in Randomized Experiments''}
\end{center}






\appendix

\section{Directions for Future Work}

In principle, our approach can be extended to settings with categorical or continuous outcomes in various ways. The simplest approach is simply to binarize the outcomes, requiring a choice of threshold. Under binarization, our estimates will continue to lower bound the quantity of interest. For example, $\tau^\basic$ after outcome binarization denotes the number of units whose binarized outcome is affected by treatment, which is a subset of the units whose original outcome (i.e., pre-binarization) is affected by treatment. 

To avoid thresholding of continuous-valued outcomes, we conjecture that rank-based estimators reminiscent of \citeappendix{rosenbaum2007interference} may also be devised, or kernel-based density estimates of the outcomes conditional on receiving treatment or control. Under the latter approach, the total variation difference between the conditional densities takes the place of the diff-in-means contrast $\Delta$, although we are unsure whether the optimization approach of \eqref{eq: LB} can be successfully adapted for inference.

We expect that our approach can be adapted to most experiment designs that have been considered for estimation of treatment effects under the assumption of a correctly specified exposure model. This is because consistency and normality theorems that require a correctly specified exposure model (normally a strong assumption) can be applied to $\hat{\tau}_1$ and $\hat{\tau}_2$ without question, as these estimators only consider units in $\Iscr$, for which the exposure model condition holds by definition. 

Experiments involving group randomization or random edge formation may be of particular interest. As previously mentioned, for such experiments $\tau^\basic$ characterizes interference, as it denotes the number of units who are affected by their randomly assigned neighbors or group members. A possible approach might be similar to \citeappendix{li2019randomization} which studied peer effects, while removing assumptions regarding interference that were required by that work. Let $A \in \{0,1\}^{N \times N}$ denote a random adjacency matrix whose distribution is jointly exchangeable following \citeappendix{bickel2009nonparametric}, so that $A$ and $\Pi A \Pi^T$ are equally probable for any permutation matrix $\Pi \in \{0,1\}^{N \times N}$. For group randomization experiments, the adjacency matrix $A$ would represent the randomly assigned groups as a collection of cliques, while edge formation experiments could use various network models, such as latent space or stochastic blockmodel \citepappendix{bickel2009nonparametric,hoff2002latent}. Let $O=(O_1,\ldots,O_N) \in \{0,1\}^N$ denote binary pre-treatment covariates for the $N$ units, with covariate classes given by $\mathcal{O}_a = \{i:O_i = a\}$ for $a \in \{0,1\}$. Let the treatment variables $X = (X_1,\ldots,X_N)$ denote the thresholded number of neighbors whose observable level is 1, so that
\[ X_i = 1\left\{\sum_{j=1}^N A_{ij} O_j \geq t_{O_i}\right\}\]
where $t_0$ and $t_1$ denote thresholds for units in $\mathcal{O}_0$ and $\mathcal{O}_1$ respectively. We conjecture that conditional on the number of units in $\mathcal{O}_0$ and $\mathcal{O}_1$ who were treated (i.e., for whom $X_i=1$), inference will take simple form in which the conditional distribution of $X$ is given by sampling without replacement within $\mathcal{O}_k$ for $k=0,1$. We additionally conjecture that this property will facilitate adaptation of the necessary asymptotic results, given by Theorems \ref{th: hajek consistency bernoulli} - \ref{th: hajek coverage bernoulli}, to this setting.

\section{Additional Notes on Examples} \label{sec: supplement examples}

\paragraph{Simulations} \citeappendix{choi2023estimating} gives a proof-of-concept attempt to estimate $\tau^{\textup{indirect}}$, which was tested on the simulations described in Section \ref{sec: simulations}. The approach was found to be lacking in power compared to ours, with 1-sided intervals whose boundaries were considerably closer to zero (and included zero considerably more often), as summarized in Table \ref{table: choi 2023}.

\paragraph{Data Example}
For the estimation of $\tau^{\textup{control}}$ in the data analysis, $W_i$ was defined to equal 1 if all of unit $i$'s first round friends were assigned to high info sessions. This choice of threshold for $W_i$ was inspired by \citeappendix{choi2024new}, who found an effect threshold at 2 treated (i.e., first round high info) friends. Thresholding $W_i$ at 2 treated friends gives similar results for $\tau^{\textup{control}}$ (point estimate: 22\%, 1-sided 95\% CI: at least 8.8\%), but restricts analysis to units with at least 2 first round friends. Thresholds exceeding 2 would remove most of the units from the analysis. These results are summarized in Table \ref{table: my results}, which shows estimates of $\tau^{\textup{control}}$ for varying choices of threshold and under Bonferoni correction (which did not cause the intervals to become weakly significant or non-significant).

\begin{table}[p]
\centering
\begin{tabular}{l|ccc|ccc}
\hline
& \multicolumn{3}{c|}{\cite{choi2023estimating}} & \multicolumn{3}{c}{Our approach}\\
$N$ &  P(Significant) & Avg. $N^{-1} LB_{1-\alpha}$ & bias & P(Significant) & Avg. $N^{-1}LB_{1-\alpha}$ & bias \\
\hline
$N=250$ &  5\% & 0.0005  &-0.26  & 80\% & 0.076 & -0.15 \\
$N=500$ &  20\% & 0.0027  & -0.23  & 97\% & 0.093 & -0.13
\end{tabular}
\caption{Comparison of \cite{choi2023estimating} and our approach for estimating $\tau^{\textup{indirect}}$ on simulated data.  P(Significant) denotes the fraction of trials for which the 1-sided interval excluded zero. Avg. $N^{-1}LB_{1-\alpha}$ denotes the average value of the lower boundary (as a fracton of $N$) of the 1-sided interval. Bias reports the average difference (as a fraction of $N$) between the point estimator and $\tau^{\textup{indirect}}$.} \label{table: choi 2023}
\end{table}

\begin{table}[p]
\centering
 \begin{tabular}{l|rrr}
 \hline
 Threshold $t_i$  & Point Estimate &  1-sided 95\% CI & Bonferoni corrected \\
\hline
$\geq 1$ treated friend & 8.8\% & at least 0\% & at least 0\% \\
$\geq 2$ treated friends & 22\% & at least 8.8\% & at least 5.7\% \\
All 1st round friends treated & 23\% & at least 9.9\% & at least 7.0\%  
 \end{tabular}
 \caption{Estimates of $\tau^{\textup{control}}$ for data example of \cite{cai2015social}, under varying choices of threshold $t_i$. Bonferoni correction is for 3 tests ($\approx 98.33\%$ confidence level). } \label{table: my results}
\end{table}


Our result is significantly stronger than the proof-of-concept attempt of \citeappendix{choi2024new} to estimate the number of second round units who were affected by the first round treatments. Their analysis did not produce a point estimate, and gave a much more conservative value of 6 units (or 2\% of a matched subset of units) for the boundary of the 1-sided confidence interval, which is smaller than ours by a factor of 13. The approach of \citeappendix{choi2023estimating} also performs poorly, producing a conservative point estimate of at least 6 affected units, which was not statistically significant (95\% CI: at least zero units). 

We may also compare to an EATE-type treatment effect that considers the relative effects of receiving $(X_i,W_i)$ equal to $(0,1)$ versus $(0,0)$:
\begin{align}\label{eq: my EATE}
 \textup{treatment effect} = \frac{1}{N} \sum_{i=1}^N \left( \mathbb{E}[Y_i|X_i=0, W_i=1] - \mathbb{E}[Y_i|X_i=0, W_i=0]\right),
\end{align}
where the expectation is taken over the randomization of treatment under the experiment design. For this target parameter, the estimation method of \citeappendix{gao2023causal} gives a H\'{a}jek-normalized point estimate of 24\%, with a 95\% CI ranging in width from $[1.4\%, 43\%]$ to $[5.3\%, 40\%]$ depending on the choice of kernel matrix $K_n$, as summarized in Table \ref{table: my EATE}. 

The role of the kernel matrix $K_n$ is to choose which pairs of units should be used by the variance estimator of Eq. (11) in \citeappendix{gao2023causal}. Given an observed network, they propose a bandwidth rule to balance between bias (from excluding important pairs of units whose residuals are strongly correlated) and variance (from including too many unit pairs). In this rule, pairs of units are used for variance estimation if the length of the shortest path between them is less than or equal to a parameter $b_n$. However, it seems difficult to discern when one choice of $b_n$ will be more accurate than another, or if any of the tested choices will necessarily be accurate for the fixed size of the data at hand. As a result, the correctness of our intervals for \eqref{eq: my EATE} may be questionable.

Table \ref{table: gao} shows selected results from the analysis of \citeappendix{gao2023causal}. It shows point estimates and estimated standard errors for the effect on insurance takeup of having at least one first round friend assigned to a high information session. Under all choices of $b_n$ exceeding zero, the effect is not significant at 95\% confidence level.\footnote{\citeappendix{gao2023causal} also report weak significance using uncorrected variance estimates (which may be anticonservative), in Tables S4 and S5 of their paper.} Unlike our analysis, \citeappendix{gao2023causal} threshold the number of treated friends at 1, and includes second round units who were assigned to high info sessions. Changing both of these choices yields the above-mentioned effect of \eqref{eq: my EATE}.

Table \ref{table: cai} shows selected regression results from the original analysis of \citeappendix{cai2015social}, which we include for reference. 
\paragraph{Software} We used the commerical solver Gurobi, through MATLAB after installing the optimization toolbox YALMIP \citepappendix{lofberg2004yalmip}.

\begin{table}[p]
\centering
\begin{tabular}{l|c}
\hline
Choice of $K_n$ (i.e., criteria for $K_n(i,j)=1$ to hold) & 95\% CI \\
\hline
$b_n=2$ (Units $i$ and $j$ have path distance $\leq 2$):  & $[1.7\% \,, 43\%]$ \\
$b_n=4$ (Units $i$ and $j$ have path distance $\leq 4$): & $[2.8\% \,, 42\%]$ \\
Units $i$ and $j$ in the same village, or have friends & \\
\hskip1cm in the same village: & $[5.3\% \,, 40\%]$ \\
Units $i$ and $j$ in same strata, or have friends  &   \\
\hskip1cm in the same strata: & $[1.4\% \,, 43\%]$ 
\end{tabular}
\caption{95\% CI for EATE estimate of \eqref{eq: my EATE}, using Eq. (10) of \cite{gao2023causal} under varying choices of kernel matrix $K_n$. The last choice of $K_n$ (same strata/friends in same strata) captures dependencies that arise when treatments are assigned by stratified random sampling without replacement, or equivalently when inference is conditional on the stratified treatment rates.}
\label{table: my EATE}
\end{table}


\begin{table}[p]
\centering
\begin{tabular}{l|lll}
\hline
& Unadj & Add & Sat \\
\hline
Point Estimate: & 0.056 & 0.058 & 0.055 \\
std. err, $b_n=0$: & 0.027 & 0.027 & 0.027 \\
std. err, $b_n=2$: & 0.033 & 0.033 & 0.033  \\
std. err, $b_n=3$: & 0.031 & 0.031 & 0.031 \\
std. err, $b_n=4$: & 0.031 & 0.031 & 0.031 \\
std. err, $b_n=5$: & 0.032 & 0.032 & 0.032 \\
\end{tabular}
\caption{Estimated spillover effects reported by Table S5 of \cite{gao2023causal}, for ``two-dimensional exposure mapping'' model and ``WLS$^+$ SE'' standard errors. Regression specifications, where $z_i$ and $x_i$ respectively denote treatment and covariate variables, are as follows. Unaj: $Y_i \sim 1 + z_i$. Add: $Y_i \sim 1 + z_i + x_i$. Sat: $1 + z_i + (x_i - \bar{x}) + z_i(x_i - \bar{x})$.  } 
\label{table: gao}
\end{table}


\begin{table}[t]
\centering
 \begin{tabular}{l|ll}
 \hline
  Regression Variable & (2) & (4) \\
\hline
Invited to high info session: & 0.0298  & 0.0809   \\
&  (0.0332)  & (0.0397) \\
Fraction of friends invited to 1st  round high info session (NET): & 0.291***   & 0.444*** \\
& (0.0820)& (0.109) \\
(NET)*(invited to high info session): &  & -0.329** \\
 & & (0.161)
 \end{tabular}
 \caption{Regression results reported by Table 2 of \cite{cai2015social} for specifications (2) and (4). For the effect of treating a single friend, \cite{cai2015social} divided the NET coefficient by 5, the number of elicited friends. } \label{table: cai} 
\end{table}

\pagebreak

\section{Asymptotic Results} \label{sec: theory appendix}

The organization of this section is as follows. In Section \ref{sec: generalized assumption}, we present Assumption \ref{as: design}, a more general version of Assumption \ref{as: design bernoulli} which allows other experiment designs to be considered. We also present Theorems \ref{th: bernoulli design} and \ref{th: two stage design} which respectively describe conditions under which  Bernoulli randomization designs and two-stage randomization designs will satisfy the assumption. In Section \ref{sec: generalized theorems}, we present Theorems \ref{th: hajek consistency} -- \ref{th: hajek coverage}, which are generalized versions of Theorems \ref{th: hajek consistency bernoulli} -- \ref{th: hajek coverage bernoulli}, and respectively give results on consistency, variance estimation, and interval coverage under Assumption \ref{as: design}. Proofs of all propositions and theorems are given in Section \ref{sec: proofs}.

\subsection{Experiment Designs} \label{sec: generalized assumption}

 Assumption \ref{as: design} describes a class of experiment designs, which includes independent Bernoulli randomization and two stage randomization designs. It assumes that units can be divided into $M$ groups $G_1,\ldots,G_M$, such that the subvectors $X_{G_1}, \ldots, X_{G_M}$ of $X$ are independent random variables. 

\begin{assumption} \label{as: design}
The units may be partitioned into $M$ non-overlapping groups $G_1,\ldots,G_M \subset [N]$ such that the subvectors $X_{G_1},\ldots,X_{G_M}$ are independent random variables, and the following hold:
\begin{enumerate}
  \item For fixed constants $C >0$ and $0 < \delta \leq 1$ it holds as $N \rightarrow \infty$ that
\begin{align} \label{eq: group size}
 |G_m| \leq C \frac{N}{M} \qquad \textup{ and } \qquad |G_m| \leq N^{1-\delta} \qquad \forall m \in [M],
\end{align}
\item There exists group neighborhoods $\eta_1,\ldots,\eta_M \subset [M]$ satisfying for fixed $d > 0$ that 
\begin{align}
\label{eq: eta size} \max_m |\eta_m| & \leq d \\
\label{eq: eta overlap} \max_{\ell} \sum_{m=1}^M 1\{\ell \in \eta_{m}\} & \leq d
\end{align}
such that the design neighborhoods $\{\Gamma_i\}_{i=1}^N$ satisfy
\begin{align} \label{eq: design neighborhood}
\Gamma_i = \bigcup_{m \in \eta_{g(i)}} G_m, 
\end{align}
where $g(i)$ denotes the group of unit $i$ (that is, $g(i) = m$ if $i \in G_m$).
\item For a fixed constant $\rho > 0$ it holds for all $i,j \in [N]$ and $a,b \in \{0,1\}$ that
\begin{align}
\label{eq: Z support} P(\bar{Z}_i = a|T_i = t_i) &\geq \rho  & \forall t_i \in \mathcal{T}_i \\
\label{eq: T support} \frac{P(T_i = t_i)P(T_j = t_j)}{P(T_i = t_i, T_j = t_j)} & \leq \frac{1}{\rho} & \forall (t_i,t_j) \in \mathcal{T}_{ij} \\
\label{eq: ZZ support} P(\bar{Z_i}=a, \bar{Z}_j=b| T_i = t_i, T_j=t_j) & \in \{0\} \cup [\rho, 1] & \forall (t_i,t_j) \in \mathcal{T}_{ij}
\end{align}
where $\mathcal{T}_i$ denotes the support of $t_i$, and $\mathcal{T}_{ij}$ denotes the support of $(t_i,t_j)$
\end{enumerate}
\end{assumption}

The conditions of Assumption \ref{as: design} may be interpreted as follows. Condition \eqref{eq: group size} requires the size of the largest group to be within a constant factor of the average group size, and allows the group sizes to grow sublinearly as $N \rightarrow \infty$. Condition \eqref{eq: design neighborhood} restricts the design neighbornoods $\{\Gamma_i\}_{i=1}^N$ to equal equal each unit's own and neighboring groups, where the group neighborhoods $\{\eta_m\}_{m=1}^M$ are required by \eqref{eq: eta size}-\eqref{eq: eta overlap} to have bounded size and overlap. Conditions \eqref{eq: Z support}-\eqref{eq: ZZ support} are support conditions, which bound the magnitude of the inverse probability weights used by $\hat{\tau}_k^\Haj$ and $\hat{V}_k$.


\paragraph{Example: Bernoulli Design}

Theorem \ref{th: bernoulli design} states that the setting of Assumption \ref{as: design bernoulli}, in which the treatments $\{X_i\}_{i=1}^N$ are independent Bernoulli random variables, can be modeled under Assumption \ref{as: design} by assigning each unit to its own group, so that $M = N$.

\begin{theorem} \label{th: bernoulli design}
Let Assumption \ref{as: design bernoulli} hold. Then Assumption \ref{as: design} holds with $M = N$.
\end{theorem}

\paragraph{Example: Two-stage Design} Theorem \ref{th: two stage design} gives conditions under which $|\Gamma_i|$ may grow at rate $N^{1-\delta}$ for any $\delta > 0$, in which the treatments are randomized by a two-stage design. In this setting, the units are divided into groups, and treatment is randomized first at the group level, and then at the unit level conditional on group level treatment. The contrast variable $\bar{Z}_i$ may consider each unit's direct treatment, the treatment of their close associates $\mathcal{N}_i$, and the number of treated units in their own and neighboring groups, where the group sizes may grow with $N$ at any sublinear polynomial rate.  
 
\begin{theorem}\label{th: two stage design}
Let $g(i)$ denote the group membership of unit $i$, so that $g(i) = m$ if $i \in G_m$. Assume the following hold for fixed $C >0$, $0 < \delta \leq 1$, $d>0$, and $p > 0$ as $N \rightarrow \infty$
\begin{enumerate}
\item[A1.] The groups $G_1,\ldots,G_M$ form a partition of $[N]$ satisfying \eqref{eq: group size}, with group neighborhoods $\{\eta_m\}_{m=1}^M$ which satisfy the size and overlap conditions \eqref{eq: eta size} and \eqref{eq: eta overlap}.

\item[A2.] Each treatment subvector $X_{G_m}$ is independently generated by sampling $\theta_m$ units without replacement from $G_m$, where $\theta_1,\ldots,\theta_M$ are independent random variables with support satisfying
\begin{align}\label{eq: theta}
\operatorname{supp}(\theta_m) \subseteq \{0\} \cup \{\theta_m^{\textup{low}}, \ldots, \theta_m^{\textup{high}}\} \cup \{|G_m|\}
\end{align}
for $\theta_m^{\textup{low}} \geq p |G_m|$ and $\theta_m^{\textup{high}} \leq (1-p)|G_m|$, and which additionally satisfy
\begin{align}\label{eq: theta p}
P(\theta_m = t) \geq p \qquad \forall t \in \operatorname{supp}(\theta_m)
\end{align}

\item[A3.] For $i \in [N]$, $\bar{Z}_i$ satisfies
\begin{align} \label{eq: two stage Z}
\bar{Z}_i &= \zeta_i(X_{\mathcal{N}_i}, \theta_{\eta_{g(i)}}), & \forall i \in [N], 
\end{align}
for some set of mappings $\{\zeta_i\}_{i=1}^N$ and unit neighborhoods $\{\mathcal{N}_i\}_{i=1}^N$. Each neighborhood $\mathcal{N}_i$ is a subset of unit $i$'s own and neighboring groups so that
\begin{align}
\label{eq: two stage N def} \mathcal{N}_i \subset \bigcup_{m \in \eta_{g(i)}} G_m 
\end{align}
holds, and also satisfies $i \in \mathcal{N}_i$ and the size and overlap conditions
\begin{align}
\label{eq: two stage N size and overlap} \max_i |\mathcal{N}_i| \leq d \qquad \textup{ and } \qquad \max_i \sum_{j=1}^N 1\{i \in \mathcal{N}_j\} \leq d
\end{align}

\item[A4.] Either $T_i=1$ for all $i \in [N]$, or $T_i = X_i$ for all $i \in [N]$, or $T_i = X_{\mathcal{N}_i}$ for all $i \in [N]$
\item[A5.] It holds that
\begin{align} \label{eq: two stage Z support}
 P(\bar{Z}_i = a \,|\, T_i=t) & \geq p  & \forall\ i \in [N], a \in \{0,1\}, t \in \mathcal{T}_i
\end{align}
where $\mathcal{T}_i$ is the support of $T_i$.
\end{enumerate}
Then Assumption \ref{as: design} holds. Additionally, Assumption \ref{as: design} continues to hold if Conditions A3 and A4 are replaced by Condition A6, given below:
\begin{enumerate}
  \item[A6.] For $i \in [N]$, $\bar{Z}_i$ and $T_i$ satisfy
  \begin{align}\label{eq: condition A6}
  \bar{Z}_i & = \zeta_i(\theta_{\eta_{g(i)}}) \qquad \textup{ and }  & T_i & = X_{G_{g(i)}}, & \qquad \forall i \in [N],
  \end{align}
  for some set of mappings $\{\zeta_i\}_{i=1}^N$.
  
\end{enumerate}
\end{theorem}

The conditions Theorem \ref{th: two stage design} may be interpreted as follows. The conditions \eqref{eq: theta} and \eqref{eq: theta p} disallows treatment of a vanishing fraction of units, and requires the group level treatment probabilities to be bounded away from zero. The condition \eqref{eq: two stage Z} requires $\bar{Z}_i$ to depend on each unit's own treatment and the treatment of a bounded number of close associates, and on the number of treated units in their own and neighboring groups. \eqref{eq: two stage N def} restricts $\mathcal{N}_i$ to come from unit $i$'s own and neighboring groups, with size and overlap conditions given by \eqref{eq: two stage N size and overlap}. \eqref{eq: two stage Z support} places a support condition on $\bar{Z}_i$. \eqref{eq: condition A6}, which replaces conditions A3 and A4, allows the estimand to count the number of units who are affected by treatments outside of their own group, whose size may increase with $N$.



\subsection{Main Theorems} \label{sec: generalized theorems}

Theorem \ref{th: hajek consistency} is a consistency result. It states that $\hat{\tau}^\Haj_k$ converges to $\tau$, and that the absolute value of $\Delta^\Haj$, the H\'{a}jek weighted contrast between units given by \eqref{eq: contrast general}, is asymptotically a lower bound for $\tau$. Table \ref{table: tau hat} gives examples of \eqref{eq: contrast general} for the estimands considered in the paper.

\begin{theorem} \label{th: hajek consistency}
It holds that
\begin{align} 
 |\Delta^\Haj| & \leq \max(\hat{\tau}_1^\Haj, \hat{\tau}_2^\Haj)  \label{eq: hajek consistency contrast} 
 \end{align}
If Assumption \ref{as: design} holds, then 
\begin{align}
 \hat{\tau}_k^\Haj = \tau + O_P\left(\frac{N}{\sqrt{M}}\right), \qquad k=1,2  \label{eq: hajek consistency Phi}
\end{align}
in which case \eqref{eq: hajek consistency contrast} implies $|\Delta^\Haj| \leq \tau + O_P\left(\frac{N}{\sqrt{M}}\right)$

\end{theorem}

Theorem \ref{th: hajek variance estimation} pertains to variance estimation. It states that $\hat{V}_k$ concentrates to its expectation which upper bounds the variance of $\hat{\tau}_k^\Haj$, and gives conditions under which the bound will be tight.

\begin{theorem} \label{th: hajek variance estimation}
Let Assumption \ref{as: design} hold. It holds that
\begin{align} \label{eq: hajek variance concentration}
\hat{V}_k = \mathbb{E}[\hat{V}_k] + O_P\left(\frac{N^2 \log M}{M^{3/2}}\right),\qquad k=1,2.
\end{align}
Additionally, it holds that 
\begin{align} \label{eq: hajek variance bias}
\mathbb{E} \hat{V}_k \geq \Var \hat{\tau}_k^\Haj,
\end{align} 
with equality if all of the following conditions hold: 
\begin{enumerate}
  \item The exposure mappings $\{T_i\}_{i=1}^N$ are constant
  \item $P(\bar{Z}_i=a) > 0$ for all $i \in [N]$ and $a \in \{0,1\}$
  \item $P(\bar{Z}_i=a, \bar{Z}_j = b) > 0 $ for all $i,j \in [N]$ such that $i \neq j$, and $a,b \in \{0,1\}$.
\end{enumerate}
\end{theorem}

 Theorem \ref{th: hajek coverage} gives coverage properties for the proposed confidence interval given by \eqref{eq: CI general}. It states that if the variance estimators $\hat{V}_k$ are lower bounded by a vanishing fraction of $\frac{N^2}{M}$, then the interval will have at least the desired asymptotic coverage level.

\begin{theorem} \label{th: hajek coverage}
Let Assumption \ref{as: design} hold. For any $\epsilon > 0$, let $\widetilde{V}_k(\phi)$ denote the thresholded variance estimate given by
\begin{align}\label{eq: V tilde}
\widetilde{V}_k(\phi) = \max\left\{ \hat{V}_k(\phi)\ , \ \frac{1}{z_{1-\frac{\alpha}{2}}^2\cdot \alpha/2} \cdot \frac{N^2}{M^{4/3-\epsilon}} \right\} 
\end{align}
 and let $LB_{1-\alpha}$ denote
 \begin{align}
LB_{1-\alpha} = \max\left\{ \min_{\phi \in \{0,1\}^N}  \hat{\tau}_1^\Haj(\phi) - z_{1-\frac{\alpha}{2}}\sqrt{\widetilde{V}_1(\phi)} \, , \, \min_{\phi \in \{0,1\}^N}  \hat{\tau}_2^\Haj(\phi) - z_{1-\frac{\alpha}{2}}\sqrt{\widetilde{V}_2(\phi)} \right\}
\end{align}
Then $\tau \geq LB_{1-\alpha}$ with probability converging to at least $1-\alpha$.
\end{theorem}

\paragraph{Remark} Theorems \ref{th: hajek consistency bernoulli} - \ref{th: hajek coverage bernoulli} can be seen to follow directly from Theorems \ref{th: hajek consistency} - \ref{th: hajek coverage}, by substituting $M = N$ and applying Theorem \ref{th: bernoulli design}.

\paragraph{Remark} The rates appearing in Theorems \ref{th: hajek consistency} - \ref{th: hajek coverage} may be more familiar or intuitive after normalizing the estimand by $N$, so that it equals $1 - |\Iscr|/N$, the fraction of units not in $\Iscr$, and then normalizing $\hat{\tau}_k^\Haj$ and $\hat{V}_k$ accordingly by $N$ and $N^2$ accordingly. For this choice of normalization, Theorem \ref{th: hajek consistency} implies $O_P(1/\sqrt{M})$ estimation error, as might be expected for an experiment whose randomization was constructed from $M$ independent draws. Similarly, Theorem \ref{th: hajek variance estimation} implies $o_P(1/M)$ error for the variance estimator after normalization by $N$, with Theorem \ref{th: hajek coverage} lower bounding the estimated variance by $o_P(1/M)$ as well.


\section{Proofs} \label{sec: proofs}

Propositions \ref{th: diff in means} - \ref{th: V evaluation} are proven in Section \ref{sec: proofs 1}. Theorems \ref{th: bernoulli design} and \ref{th: two stage design} are proven in Section \ref{sec: proofs 2}. Proofs of Theorems \ref{th: hajek consistency}, \ref{th: hajek variance estimation}, and \ref{th: hajek coverage} (which are generalized versions of Theorems \ref{th: hajek consistency bernoulli}, \ref{th: hajek variance estimation bernoulli}, and \ref{th: hajek coverage bernoulli}), including auxilliary lemmas, are respectively contained in Sections \ref{sec: proofs 3}, \ref{sec: proofs 4}, and \ref{sec: proofs 5}.

\subsection{Proof of Propositions \ref{th: diff in means} - \ref{th: V evaluation}} \label{sec: proofs 1}

\begin{proof}[Proof of Proposition \ref{th: diff in means}]
It can be seen that
\begin{align}
\label{eq: proof diff in means 1} \hat{\tau}_1 & = N - \sum_{i \in \Iscr} \left(\frac{X_i Y_i}{P(X_i=1)} + \frac{(1-X_i)(1-Y_i)}{P(X_i=0)}\right) \\
\label{eq: proof diff in means 2} & \geq N -  \sum_{i=1}^N \left(\frac{X_i Y_i}{P(X_i=1)} + \frac{(1-X_i)(1-Y_i)}{P(X_i=0)}\right) \\
\label{eq: proof diff in means 3} & =  N - \sum_{i=1}^N \left(\frac{X_i Y_i}{P(X_i=1)} - \frac{(1-X_i)Y_i}{P(X_i=0)} + \frac{1-X_i}{P(X_i=0)}\right) \\
\label{eq: proof diff in means 4} & = \sum_{i=1}^N \left(\frac{1-X_i}{P(X_i=0)} - \frac{X_i}{P(X_i=1)}\right)Y_i + O_P(N^{1/2})
\end{align}
where \eqref{eq: proof diff in means 1} uses the definition of $\hat{\tau}_1$ given in \eqref{eq: tau hat} and that 
\[ P(H_i=1) = Y_i P(X_i=1) + (1-Y_i) P(X_i=0) \qquad \forall\ i \in \Iscr \]
as $Y_i$ is constant if $i \in \Iscr$; \eqref{eq: proof diff in means 2} holds as the terms in the summand are nonnegative; \eqref{eq: proof diff in means 3} is simple rearrangement of terms, and \eqref{eq: proof diff in means 4} uses \eqref{eq: condition X} assumed by the lemma. 

By essentially identical arguments, it holds that
\begin{align} 
\nonumber \hat{\tau}_2 & = N - \sum_{i\in \Iscr} \left(\frac{X_i(1-Y_i)}{P(X_i=1)} + \frac{(1-X_i)Y_i}{P(X_i=0)}
\right)\\
 \nonumber & \geq N - \sum_{i=1}^N \left(\frac{X_i(1-Y_i)}{P(X_i=1)} + \frac{(1-X_i)Y_i}{P(X_i=0)}
\right)\\
 \nonumber & \geq N - \sum_{i=1}^N \left(-\frac{X_i Y_i}{P(X_i=1)} + \frac{(1-X_i)Y_i}{P(X_i=0)} + \frac{X_i}{P(X_i=1)} \right)\\
& = \sum_{i=1}^N \left(\frac{X_i}{P(X_i=1)} - \frac{1-X_i}{P(X_i=0)}\right)Y_i + O_P(N^{1/2}) \label{eq: proof diff in means 5}
\end{align}
Combining \eqref{eq: proof diff in means 4} and \eqref{eq: proof diff in means 5} yields
\[ \max(\hat{\tau}_1,\hat{\tau}_2) \geq \left| \sum_{i=1}^N \left(\frac{X_i}{P(X_i=1)} -  \frac{1-X_i}{P(X_i=0)}\right) Y_i\right|,\]
proving the result.
\end{proof}

\begin{proof}[Proof of Proposition \ref{th: diff in means Z}]
Substitute $Z_i$ for $X_i$ in the proof of Theorem \ref{th: diff in means}. 
\end{proof}

\begin{proof}[Proof of Proposition \ref{th: tight}]
By \eqref{eq: phi general} and the conditions of the proposition, it holds for $a\in \{0,1\}$ that 
\begin{align}\label{eq: tightness}
1\{\bar{Z}_i = a\} Y_i = \begin{cases} 1\{\bar{Z}_i = a\}\tilde{f}_i(T_i) & \textup{ if } i \in \Iscr \\ 
1\{\bar{Z}_i = a\} \bar{Z}_i & \textup{ if } i \notin \Iscr \end{cases} 
\end{align}
where we have used that $\bar{Z}_i \in \{0,1\}$ implies $X \in R_i$ by \eqref{eq: Z bar}. Plugging \eqref{eq: tightness} into $\Delta^{\bar{Z}}$ yields
\begin{align*}
\Delta^{\bar{Z}} & = \sum_{i \in \Iscr} \left(\frac{1\{\bar{Z}_i=1\}}{P(\bar{Z}_i = 1|T_i)}  - \frac{1\{\bar{Z}_i = 0\}}{P(\bar{Z}_i = 0|T_i)}\right)\tilde{f}_i(T_i) + \sum_{i \notin\Iscr} \left(\frac{1\{\bar{Z}_i=1\}}{P(\bar{Z}_i = 1|T_i)}  - \frac{1\{\bar{Z}_i = 0\}}{P(\bar{Z}_i = 0|T_i)}\right)\bar{Z}_i 
\end{align*}
It can be seen that the first summation (which is over $i \in \Iscr$) has expectation zero, while the remaining summation has expectation $|\{i: i \notin \Iscr\}| = N - |\Iscr|$, proving the proposition.
\end{proof}

\begin{proof}[Proof of Proposition \ref{th: Phi evaluation}]
Let $\tilde{f}_{ki}$ for $i\in \Iscr$ denote the mapping
\begin{align} \label{eq: f_ki}
\tilde{f}_{ki}(T_i) = \begin{cases} \tilde{f}_i(T_i) & \textup{ if $k = 1$} \\ 1 - \tilde{f}_i(T_i) & \textup{ if $k = 2$} \end{cases}
\end{align}
As $Y_{ki} = \tilde{f}_{ki}(T_i)$ whenever $X \in R_i$ and $i \in \Iscr$, it follows that
\begin{align}
\label{eq: phi evaluation 1} P(X \in R_i, Z_i = Y_{ki}|T_i) & = P(X \in R_i, Z_i = \tilde{f}_{ki}(T_i)|T_i) \\
\nonumber & = \sum_{a=0}^1 1\{\tilde{f}_{ki}(T_i)=a\} \ P(X \in R_i, Z_i = a | T_i, X \in R_i) \\
\nonumber  & = \tilde{f}_{ki}(T_i) \cdot P(X \in R_i, Z_i = 1|T_i) + (1-\tilde{f}_{ki}(T_i)) \cdot P(X \in R_i, Z_i=0|T_i) \\
\label{eq: phi evaluation 4} & = Y_{ki} \cdot P(X \in R_i, Z_i = 1|T_i) + (1-Y_{ki}) \cdot P(X \in R_i, Z_i=0|T_i) 
\end{align}
where \eqref{eq: phi evaluation 1} holds because $i \in \Iscr$, and \eqref{eq: phi evaluation 4} holds because $X \in R_i$ and $i \in \Iscr$, so that we $Y_{ki} = \tilde{f}_{ki}(T_i)$
\end{proof}

\begin{proof}[Proof of Proposition \ref{th: V evaluation}]

Let $\tilde{f}_{ki}$ be given by \eqref{eq: f_ki}. Since $i \in \Iscr$ and $\bar{Z}_i = Y_{ki}$ (which implies $X \in R_i$), it holds that $Y_{ki} = \tilde{f}_{ki}(T_i)$, and hence that
\begin{align}
\nonumber \mathbb{E}[v_i | T_i] & = \mathbb{E}\left[ \frac{N}{\hat{N}_{Y_{ki}}} \cdot \frac{1\{\bar{Z}_i = Y_{ki}\}}{P(\bar{Z}_i = Y_{ki}|T_i)} | T_i\right] \\
\label{eq: V eval 1} & = P(X \in R_i|T_i) \cdot \mathbb{E}\left[ \frac{N}{\hat{N}_{Y_{ki}}} \cdot \frac{1\{\bar{Z}_i = Y_{ki}\}}{P(\bar{Z}_i = Y_{ki}|T_i)} | T_i, X \in R_i\right] \\
\label{eq: V eval 2} & = P(X \in R_i|T_i) \cdot \sum_{a = 0,1} 1\{Y_{ki} = a\} \mathbb{E}\left[ \frac{N}{\hat{N}_a} \cdot \frac{1\{\bar{Z}_i = a\}}{P(\bar{Z}_i = a|T_i)} | T_i, X \in R_i\right] \\
\label{eq: V eval 3} & = \sum_{a = 0,1} 1\{Y_{ki} = a\} \cdot P(X \in R_i|T_i) \mathbb{E}\left[ \frac{N}{\hat{N}_a} \cdot \frac{1\{\bar{Z}_i = a\}}{P(\bar{Z}_i = a|T_i)} | T_i, X \in R_i\right] \\
\label{eq: V eval 4} & = \sum_{a = 0,1} 1\{Y_{ki} = a\} \mathbb{E}\left[ \frac{N}{\hat{N}_a} \cdot \frac{1\{\bar{Z}_i = a\}}{P(\bar{Z}_i = a|T_i)} | T_i \right] 
\end{align}
where \eqref{eq: V eval 1} holds since $1\{\bar{Z}_i = Y_{ki}\} = 0$ if $X \notin R_i$; \eqref{eq: V eval 2} holds since the $Y_{ki}$ term in the conditional expectation equals $\tilde{f}_{ki}(T_i)$, and hence may be treated as a constant equal to the observed value of $Y_{ki}$; \eqref{eq: V eval 3} is a simple rearrangement of terms; and \eqref{eq: V eval 4} uses that $1\{\bar{Z}_i = a\} = 0$ for $a \in \{0,1\}$ if $X \notin R_i$.

By similar reasoning, if $i,j \in \Iscr$ and $1\{\bar{Z}_i = Y_{ki}, \bar{Z}_j = Y_{kj}\} = 1$ (which implies that $X \in R_i \cap R_j$), it holds that $Y_{ki} = \tilde{f}_{ki}(T_i)$ and $Y_{kj} = \tilde{f}_{kj}(T_j)$, and hence that
\begin{align*}
\mathbb{E}[v_i v_i| T_i, T_j]  & = \mathbb{E}\left[ \frac{1\{\bar{Z}_i = Y_{ki}\}N }{P(\bar{Z}_i = Y_{ki}|T_i) \hat{N}_{Y_{ki}}} \cdot \frac{1\{\bar{Z}_i = Y_{ki}\}N }{P(\bar{Z}_i = Y_{ki}|T_i) \hat{N}_{Y_{ki}}} | T_i, T_j \right]  \\
& = P(X \in R_i \cap R_j)\mathbb{E}\left[ \frac{1\{\bar{Z}_i = Y_{ki}\}N }{P(\bar{Z}_i = Y_{ki}|T_i) \hat{N}_{Y_{ki}}} \cdot \frac{1\{\bar{Z}_i = Y_{ki}\}N }{P(\bar{Z}_i = Y_{ki}|T_i) \hat{N}_{Y_{ki}}} | T_i, T_j, X \in R_i \cap R_j \right]  \\
& = P(X \in R_i \cap R_j) \sum_{a=0}^1 \sum_{b=0}^1 1\{Y_{ki}=a, Y_{kj} = b\} \cdot \\
& \hskip1cm \mathbb{E}\left[ \frac{1\{\bar{Z}_i = a\}N }{P(\bar{Z}_i = a|T_i) \hat{N}_{a}} \cdot \frac{1\{\bar{Z}_i = b\}N }{P(\bar{Z}_i = b|T_i) \hat{N}_{b}} | T_i, T_j, X \in R_i \cap R_j \right]  \\
& = \sum_{a=0}^1 \sum_{b=0}^1 1\{Y_{ki}=a, Y_{kj} = b\} \cdot \\
& \hskip1cm P(X \in R_i \cap R_j) \mathbb{E}\left[ \frac{1\{\bar{Z}_i = a\}N }{P(\bar{Z}_i = a|T_i) \hat{N}_{a}} \cdot \frac{1\{\bar{Z}_i = b\}N }{P(\bar{Z}_i = b|T_i) \hat{N}_{b}} | T_i, T_j, X \in R_i \cap R_j \right]  \\
& = \sum_{a=0}^1 \sum_{b=0}^1 1\{Y_{ki}=a, Y_{kj} = b\} \mathbb{E}\left[ \frac{1\{\bar{Z}_i = a\}N }{P(\bar{Z}_i = a|T_i) \hat{N}_{a}} \cdot \frac{1\{\bar{Z}_i = b\}N }{P(\bar{Z}_i = b|T_i) \hat{N}_{b}} | T_i, T_j \right]  
\end{align*}
and
\begin{align*}
P(\bar{Z}_i = Y_{ki}, \bar{Z}_j = Y_{kj} | T_i, T_j) & = P(X \in R_i \cap R_j|T_i, T_j) P(\bar{Z}_i = Y_{ki}, \bar{Z}_j = Y_{kj} | T_i, T_j, X \in R_j \cap R_j) \\
& = P(X \in R_i \cap R_j|T_i, T_j)  \cdot \sum_{a=0}^1 \sum_{b=0}^1 1\{Y_{ki}=a, Y_{kj}=b\} \cdot \\
& \hskip1cm P(\bar{Z}_i = a, \bar{Z}_j = b | T_i, T_j, X \in R_j \cap R_j) \\
& = \sum_{a=0}^1 \sum_{b=0}^1 1\{Y_{ki}=a, Y_{kj}=b\} \cdot \\
& \hskip1cm P(X \in R_i \cap R_j|T_i, T_j) P(\bar{Z}_i = a, \bar{Z}_j = b | T_i, T_j, X \in R_j \cap R_j) \\
& = \sum_{a=0}^1 \sum_{b=0}^1 1\{Y_{ki}=a, Y_{kj}=b\} \cdot P(\bar{Z}_i = a, \bar{Z}_j = b | T_i, T_j) \\
\end{align*}
\end{proof}

\subsection{Proof of Theorems \ref{th: bernoulli design} and \ref{th: two stage design}} \label{sec: proofs 2}

\begin{proof}[Proof of Theorem \ref{th: bernoulli design}]
Given design neighborhoods $\Gamma_i$ satisfying \eqref{eq: Gamma 1} and \eqref{eq: Gamma 2}, for $i \in [N]$ let $G_i = \{i\}$ and $\eta_i = \Gamma_i$. The conditions \eqref{eq: group size} - \eqref{eq: ZZ support} of Assumption \ref{as: design} can be verified as follows:
\begin{enumerate}
  \item \eqref{eq: group size} and \eqref{eq: design neighborhood} follow from the definitions of $\{G_i\}$ and $\{\eta_i\}$
  \item \eqref{eq: eta size} and \eqref{eq: eta overlap} follow because $\eta_i = \Gamma_i$, which satisfies \eqref{eq: Gamma 1} and \eqref{eq: Gamma 2}.
  \item \eqref{eq: Z support} holds because it is assumed as \eqref{eq: pmin}
  \item To show \eqref{eq: T support} and \eqref{eq: ZZ support}, we observe that $|\Gamma_i \cup \Gamma_j| \leq 2d$, and hence by \eqref{eq: pmin X} and independence of the treatment assignments it holds that
  \[ P(X_{\Gamma_i \cup \Gamma_j} = x) \geq \min(p, 1-p)^{2d} \qquad \forall x \in \{0,1\}^{|\Gamma_i \cup \Gamma_j|}\]
  As a result, for any $(t_i, t_j) \in \operatorname{supp}(T_i, T_j)$ it holds that
  \begin{align} \label{eq: bernoulli design 1}
    P(T_i = t_i, T_j = t_j) \geq \min(p, 1-p)^{2d} 
  \end{align}
  where we have used that $(T_i, T_j)$ depend on $X$ only through $X_{\Gamma_i \cup \Gamma_j}$. Similarly, for any $(a,b,t_i,t_j) \in \operatorname{supp}(\bar{Z}_i, \bar{Z}_j, T_i, T_j)$ it also holds that
  \begin{align}\label{eq: bernoulli design 2}
 P(\bar{Z}_i = a, \bar{Z}_j = b, T_i = t_i, T_j = t_j) \geq \min(p, 1-p)^{2d}, 
 \end{align}
where we have used that $(\bar{Z}_i, \bar{Z}_j, T_i, T_j)$ depends on $X$ only through $X_{\Gamma_i \cup \Gamma_j}$. By \eqref{eq: bernoulli design 1} and \eqref{eq: bernoulli design 2}, it follows that \eqref{eq: T support} and \eqref{eq: ZZ support} hold for $\rho = \min(p, 1-p)^{2d}$.
\end{enumerate}

\end{proof}

We prove Theorem \ref{th: two stage design} in two parts. First we prove the result under the original conditions A1 - A5. We then show that the result continues to hold when Conditions A3 and A4 are replaced by Condition A6, as stated by the theorem.

\begin{proof}[Proof of Theorem \ref{th: two stage design}, using Conditions A1 - A5]
The conditions \eqref{eq: group size} - \eqref{eq: ZZ support} of Assumption \ref{as: design} can be verified as follows:
\begin{enumerate}
  \item \eqref{eq: group size}, \eqref{eq: eta size}, \eqref{eq: eta overlap}, and \eqref{eq: Z support} hold because they are explicitly assumed by the conditions of the theorem.
  \item \eqref{eq: design neighborhood} holds because $\bar{Z}_i$ and $T_i$ depend only on $X_i$, $X_{\mathcal{N}_i}$, and $\theta_{\eta_{g(i)}}$, all of which  depend only on the treatment of unit $i$'s own and neighboring groups.
  \item To show \eqref{eq: T support} and \eqref{eq: ZZ support}, we will show that for $\rho$ given by
    \[ \rho = p^{2d} \left(\frac{p}{1+(1-p)2d}\right)^{2d}  \]
    it holds that
    \begin{align}
    \label{eq: sampling p bound T} P(T_i = t_i, T_j = t_j) & \geq  \rho \qquad \forall (t_i,t_j) \in \operatorname{supp}(T_i,T_j)
    \end{align}
    which implies \eqref{eq: T support}, and it also holds that 
    \begin{align}
    \label{eq: sampling p bound ZZ} P(\bar{Z}_i=a, \bar{Z}_j=b, T_i=t_i, T_j=t_j) & \geq \rho \qquad \forall (a,b,t_i, t_j) \in \operatorname{supp}(\bar{Z}_i, \bar{Z}_j, T_i, T_j)
    \end{align}
    which implies \eqref{eq: ZZ support}.
  \end{enumerate}
  To complete the proof, we now show \eqref{eq: sampling p bound ZZ}, which also implies \eqref{eq: sampling p bound T}.  Let the sets $S \subset [N]$ and $Q \subset [M]$ be given by 
  \[ S = \mathcal{N}_i \cup \mathcal{N}_j \qquad \textup{and} \qquad Q = \eta_{g(i)} \cup \eta_{g(j)} \]
  Because $(\bar{Z}_i, \bar{Z}_j, T_i, T_j)$ depend on $X$ and $\theta$ only through $X_S$ and $\theta_Q$, to show \eqref{eq: sampling p bound ZZ} it suffices to show that 
  \begin{align} \label{eq: SQ support}
P(X_S = x, \theta_Q = \xi) & \geq \rho \qquad \forall (x,\xi) \in \operatorname{supp}(X_S, \theta_Q)
  \end{align}
  To show \eqref{eq: SQ support}, let $X_{S,i}$ denote the $i$th entry of $X_S$, let $X_{S,1:i-1}$ denote the first $(i-1)$ entries of $X_S$, let $\bar{p} = \min(p,1-p)$, and let $(a_i,b_i,k_i,\tilde{k}_i)$ denote terms to be defined shortly. It then holds for any $(x,\xi) \in \operatorname{supp}(X_S, \theta_Q)$ that 
  \begin{align}
  \label{eq: two stage 1} P(X_S = x, \theta_Q = \xi) & \geq p^{2d} P(X_S = x | \theta_Q = \xi) \\
  \label{eq: two stage 2} & = p^{2d} \prod_{i=1}^{|S|} P(X_{S,i} = x_i|\theta_Q = \xi, X_{S,1:i-1}=x_{1:i-1}) \\
  \label{eq: two stage 3} & = p^{2d} \prod_{i=1}^{|S|} \frac{a_i - k_i}{b_i - \tilde{k}_i} \\
  \label{eq: two stage 4} & \geq p^{2d} \left(\frac{\bar{p}}{1+(1-\bar{p})2d}\right)^{2d}
  \end{align}
by the following steps:
\begin{enumerate}
  \item \eqref{eq: two stage 1} holds because $\theta_1,\ldots,\theta_M$ are indpendent and satisfy \eqref{eq: theta p}
  \item \eqref{eq: two stage 2} writes $P(X_S |\theta_Q)$ as a product of conditional probabilities
  \item \eqref{eq: two stage 3} holds because the entries of $X_S$ conditional on $\theta_Q$ are generated by sampling without replacement from each group, with $(a_i,b_i,k_i,\tilde{k}_i)$ defined as follows. Let $\{\ell_1,\ldots,\ell_{|S|}\} \subset [N]$ enumerate the units in $S$, so that $X_{S} = (X_{\ell_1}, \ldots, X_{\ell_{|S|}})$ and $X_{S,i} = X_{\ell_i}$. Then it holds that 
  \[ P(X_{S,i} = x_i|\theta_Q = \xi, X_{S,1:i-1}=x_{1:i-1}) = \frac{a_i - k_i}{b_i - \tilde{k}_i}\]
  for $(a_i,b_i,k_i, \tilde{k}_i)$ given by
\begin{align}
\nonumber a_i & = 
\begin{cases} \theta_{g(\ell_i)} & \textup{if } x_i=1 \\
|G_{g(\ell_i)}| - \theta_g(\ell_i) & \textup{if } x_i = 0,
\end{cases} & b_i &= |G_{g(\ell_i)}| \\
\label{eq: k tilde} k_i & = \sum_{j=1}^{i-1} 1\{x_j=x_i, g(\ell_j) = g(\ell_i)\}, & \tilde{k}_i &= \sum_{j=1}^{i-1} 1\{g(\ell_j) = g(\ell_i)\}
\end{align}
because the conditional probability of unit $\ell_i$ receiving treatment is equal to the probability of drawing a colored ball from an urn with $|G_{g(\ell_i)}|$ balls, of which $\theta_{g(\ell_i)}$ are colored, conditional on previously drawing without replacement a certain number of colored and non-colored balls corresponding to the treatment assignment of the units in $\{\ell_1,\ldots,\ell_{i-1}\}$ who are in the same group as $\ell_i$.
  \item To show \eqref{eq: two stage 4} holds we observe that 
  \begin{align}\label{eq: two stage 4 bound 1}
   \frac{a_i - k_i}{b_i - \tilde{k}_i} = 1 \qquad \textup{ if }\theta_{g(\ell_i)} = 0 \textup{ or } \theta_{g(\ell_i)}=|G_{g(\ell_i)}| 
   \end{align}
  which holds because the treatment of unit $\ell_i$ can only take one value when all units in $\ell_i$'s group receive the same treatment, and 
  \begin{align} \label{eq: two stage 4 bound 2}
   \frac{a_i - k_i}{b_i - \tilde{k}_i} \geq \frac{\bar{p}}{1 + (1-\bar{p})d} \qquad \textup{ if }\theta_{g(\ell_i)} \neq 0 \textup{ and } \theta_{g(\ell_i)}\neq |G_{g(\ell_i)}|
  \end{align}
  which holds by the following steps for $\theta_{g(\ell_i)} \notin \{0, |G_{g(\ell_i)}|\}$:
  \begin{align}
\label{eq: two stage 5} \frac{a_i - k_i}{b_i - \tilde{k}_i} & \geq \frac{a_i - k_i}{b_i - k_i} \\
\label{eq: two stage 6} & \geq \frac{a_i - k_i}{\frac{a_i}{\bar{p}} - k_i} \\
\label{eq: two stage 8} & = \frac{\bar{p}(a_i - k_i)}{(a_i - k_i) + (1-\bar{p})k_i} \\
\label{eq: two stage 9} & \geq \frac{\bar{p}}{1 + (1-\bar{p})2d}
\end{align}
where \eqref{eq: two stage 5} because $\tilde{k}_i \geq k_i$ by their definitions in \eqref{eq: k tilde}; \eqref{eq: two stage 6} holds because $b_i = |G_{g(\ell_i)}| \leq a_i/\bar{p}$ by \eqref{eq: theta}; \eqref{eq: two stage 8} follows from algebraic manipulation; and \eqref{eq: two stage 9} holds because $k_i \leq |S| \leq 2d$, and because the expression
\[\frac{\bar{p}(a_i - k_i)}{(a_i - k_i) + (1-\bar{p})2d}\]
is increasing in $(a_i - k_i)$, where $(a_i-k_i) \geq 1$ since $x$ is in the support of $X_S$ conditional on $\theta_Q$. Plugging into \eqref{eq: two stage 3} the bounds for $\frac{a_i-k_i}{b_i - \tilde{k}_i}$ given by \eqref{eq: two stage 4 bound 1} and \eqref{eq: two stage 4 bound 2} and bounding $|S| \leq 2d$ establishes \eqref{eq: two stage 4}.
\end{enumerate}

\end{proof}

\begin{proof}[Proof of Theorem \ref{th: two stage design}, using Condition A6]

To verify the conditions \eqref{eq: group size} - \eqref{eq: ZZ support} of Assumption \ref{as: design} when Conditions A3 and A4 are replaced by Condition A6, we observe that

\begin{enumerate}
  \item As before, \eqref{eq: group size}, \eqref{eq: eta size}, \eqref{eq: eta overlap}, and \eqref{eq: Z support} hold because they are explicitly assumed by the conditions of the theorem.
  \item \eqref{eq: design neighborhood} holds because $\bar{Z}_i$ and $T_i$ depend only on $\theta_{\eta_{g(i)}}$, the number of treated units in $i$'s own and neighboring groups.
  \item \eqref{eq: T support} holds because $T_i$ and $T_j$ are independent for $i\neq j$, under the conditions of the theorem.
  \item To show \eqref{eq: ZZ support}, we observe that $(\bar{Z}_i, \bar{Z}_j, T_i, T_j)$ depend on $X$ only through the group-level variables 
  $\{\theta_m: m \in \eta_{g(i)} \cup \eta_{g(j)}\}$. Since $\{\theta_m\}_{m=1}^M$ are independent variables which satisfy the support condition \eqref{eq: theta p}, and $|\eta_{g(i)} \cup \eta_{g(j)}| \leq 2d$ by \eqref{eq: eta size}, it follows for any $(a,b,t_i,t_j)$ in the support of $(\bar{Z}_i, \bar{Z}_j, T_i, T_j)$ that
  \begin{align}\label{eq: condition A6 proof}
  P(\bar{Z}_i=a, \bar{Z}_j=b, T_i = t_i, T_j=T_j) \geq p^{2d} 
  \end{align}
  which implies \eqref{eq: ZZ support} for $\rho = p^{2d}$.
\end{enumerate}

\end{proof}

\subsection{Proof of Theorem \ref{th: hajek consistency}} \label{sec: consistency} \label{sec: proofs 3}

\paragraph{Preliminaries}

The Azuma-Hoeffding (or McDiarmid's) inequality \citeappendix[Thm 6.2]{boucheron2013concentration} states that given independent variables $U = (U_1,\ldots,U_N)$, and a function $f$ satisfying the bounded difference property
\begin{equation} \label{eq: mcdiarmid condition}
 |f(U) - f(U')| \leq c_i \text{ if } U_j = U_j' \text{ for all } j \neq i, \qquad i \in [N],
\end{equation}
it holds that 
\begin{equation}\label{eq: mcdiarmid}
\mathbb{P}\left( | f(U) - \mathbb{E}f | > \epsilon \right) \leq 2\exp\left( - 2\epsilon^2/\sum_i c_i^2 \right).
\end{equation}

Lemma \ref{le: hajek mcdiarmid} applies McDiarmid's inequality to show concentration of the terms appearing in $\hat{\tau}_k^\Haj$ for $k=1,2$. 

\begin{lemma}\label{le: hajek mcdiarmid}
Let Assumption \ref{as: design} hold, and let $\phi^* \in \{0,1\}^N$ encode $\Iscr$ by
\begin{align}\label{eq: phi star}
 \phi_i^* = 1\{i \in \Iscr\}, \qquad i \in [N]
 \end{align}
For $i \in [N]$ and $\bar{Z}_i$ given by \eqref{eq: Z bar}, let $u_i$ denote the vector
\[ u_i = \left(\frac{1\{\bar{Z}_i=1\}}{P(\bar{Z}_i=1|T_i)}\ , \ \frac{1\{\bar{Z}_i=1\}}{P(\bar{Z}_i=1|T_i)}Y_i \phi_i^*\ , \ \frac{1\{\bar{Z}_i=0\}}{P(\bar{Z}_i=0|T_i)}\ , \  \frac{1\{\bar{Z}_i=0\}}{P(\bar{Z}_i=0|T_i)}(1-Y_i) \phi_i^* \right) \]
and let $\bar{u} = \frac{1}{N}\sum_{i=1}^N u_i$. Then $\| \bar{u} - \mathbb{E}\bar{u} \|_\infty = O_P\left(\frac{N}{\sqrt{M}}\right)$.
\end{lemma}

\begin{proof}
Let $U_m$ for $m \in [M]$ be given by
\[ U_m = \frac{1}{N} \sum_{i \in G_m} u_i \]
so that $\bar{u}$ satisfies $\bar{u} = \sum_{m=1}^M U_m$. As $|u_i| \leq 1/\rho$ by \eqref{eq: Z support} and $|G_m| \leq CN/M$ by \eqref{eq: group size}, it can be seen that $|U_m| \leq \frac{C}{M\rho}$. By \eqref{eq: eta overlap}, changing $X_{G_m}$ changes at most $d$ entries of $U_1,\ldots,U_M$, and hence changes $\bar{u}$ by at most $\frac{dC}{M \rho}$. Applying McDiarmid's inequality to the $\ell$th entry of $\bar{u} = \sum_{m=1}^M U_m$  yields that
\[ P( |\bar{u}_\ell - \mathbb{E}\bar{u}_\ell| \geq \epsilon ) \leq \exp\left( - \frac{2\epsilon^2}{M \left(\frac{dC}{M \rho}\right)^2}  \right) \]
and hence $\bar{u}_\ell - \mathbb{E}\bar{u}_\ell = O_P(M^{-1/2})$. Since $\bar{u}$ is 4-dimensional, it follows that $\| \bar{u} - \mathbb{E}\bar{u}\|_\infty = O_P(M^{-1/2})$ as well.
\end{proof}

\paragraph{Proof of Theorem \ref{th: hajek consistency}} The claims of Theorem \ref{th: hajek consistency}, namely \eqref{eq: hajek consistency contrast} and \eqref{eq: hajek consistency Phi}, are proven separately as follows:

\begin{proof}[Proof of Theorem \ref{th: hajek consistency}, Eq. \eqref{eq: hajek consistency contrast}]
It holds that
\begin{align} 
\nonumber \hat{\tau}_k^\Haj & =  N - \left(\sum_{i \in \Iscr} \frac{N}{\hat{N}_1} \cdot \frac{1\{\bar{Z}_i=1\}}{P(\bar{Z}_i=1|T_i)} \cdot Y_{ki}+ \sum_{i \in \Iscr} \frac{N}{\hat{N}_0} \cdot \frac{1\{\bar{Z}_i=0\}}{P(\bar{Z}_i=1|T_i)} \cdot (1-Y_{ki}) \right) \\
\label{eq: hajek contrast 2} & \geq N - \left(\sum_{i=1}^N \frac{N}{\hat{N}_1} \cdot \frac{1\{\bar{Z}_i=1\}}{P(\bar{Z}_i=1|T_i)} \cdot Y_{ki}+ \sum_{i=1}^N \frac{N}{\hat{N}_0} \cdot \frac{1\{\bar{Z}_i=0\}}{P(\bar{Z}_i=1|T_i)} \cdot (1-Y_{ki}) \right) \\
& = N - \left(\sum_{i=1}^N \frac{N}{\hat{N}_1} \cdot \frac{1\{\bar{Z}_i=1\}}{P(\bar{Z}_i=1|T_i)} \cdot Y_{ki} - \sum_{i=1}^N \frac{N}{\hat{N}_0} \cdot \frac{1\{\bar{Z}_i=0\}}{P(\bar{Z}_i=1|T_i)} \cdot Y_{ki} \right) - N\label{eq: N hajek} \\
 \label{eq: hajek constrast 3} & = \begin{cases} 
 - \Delta^{\Haj} & \textup{ if }k = 1 \\
 \Delta^{\Haj} & \textup{ if }k = 2
  \end{cases}
\end{align}
where \eqref{eq: hajek contrast 2} holds because the terms in the summand are non-negative; \eqref{eq: N hajek} uses 
\[ \sum_{i=1}^N \frac{N}{\hat{N}_a}\cdot \frac{1\{\bar{Z}_i=a\}}{P(\bar{Z}_i=a|T_i)} = N, \qquad \forall\ a \in \{0,1\},\]
and \eqref{eq: hajek constrast 3} uses the definition of $\Delta^\Haj$ given by \eqref{eq: contrast general}.
It follows that
\begin{align*}
\max(\hat{\tau}_1^\Haj, \hat{\tau}_2^\Haj) & \geq \max(\Delta^\Haj, -\Delta^\Haj)  \\ 
& = |\Delta^\Haj|
\end{align*}
proving \eqref{eq: hajek consistency contrast}.

\end{proof}

\begin{proof}[Proof of Theorem \ref{th: hajek consistency}, Eq. \eqref{eq: hajek consistency Phi}]

Proof of \eqref{eq: hajek consistency Phi} follows by Taylor approximation and is more or less standard. Given $k \in [2]$, let $a, b \in \mathbb{R}^2$ be given by
\begin{align}
a &= \left[ \begin{array}{l} \dfrac{1}{N} \displaystyle\sum_{i=1}^N \dfrac{1\{\bar{Z}_i=1\}}{P(\bar{Z}_i=1|T_i)} \cdot Y_{ki} \phi_i^* \\
\\
\dfrac{1}{N} \displaystyle\sum_{i=1}^N \dfrac{1\{\bar{Z}_i=0\}}{P(\bar{Z}_i=0|T_i)} \cdot (1-Y_{ki}) \phi_i^*
\end{array} \right] & \qquad \textup{ and } \qquad 
b & = \left[ \begin{array}{c} \dfrac{1}{N} \displaystyle\sum_{i=1}^N \dfrac{1\{\bar{Z}_i=1\}}{P(\bar{Z}_i=1|T_i)} \\ \\
\dfrac{1}{N} \displaystyle\sum_{i=1}^N \dfrac{1\{\bar{Z}_i=0\}}{P(\bar{Z}_i=0|T_i)} 
\end{array} \right]  \label{eq: linearization terms}
\end{align}
To evaluate $\mathbb{E}a$, we observe that for $z \in \{0,1\}$, $i \in \Iscr$, and $\tilde{f}_{ki}(T_i)$ given by \eqref{eq: f_ki}, it holds that
\begin{align}
\nonumber \mathbb{E}\left[\frac{1\{\bar{Z}_i = z\}}{P(\bar{Z}_i = z| T_i)} Y_{ki} \right] &= \mathbb{E}\left[\mathbb{E}\left\{\frac{1\{\bar{Z}_i = z\}}{P(\bar{Z}_i = z| T_i)} Y_{ki}  | T_i\right\}\right]  \\
\nonumber & = \mathbb{E}\left[ \mathbb{E}\left\{ Y_{ki} | T_i, \bar{Z}_i=z\right\} \right]\\
\label{eq: consistency linearization 2} & = \mathbb{E}\left[ \mathbb{E}\left\{ \tilde{f}_{ki}(T_i) | T_i, \bar{Z}_i=z\right\} \right]\\
\nonumber & = \mathbb{E}\left\{\mathbb{E}\left[ \tilde{f}_{ki}(T_i) | T_i\right]  \right\} \\
\nonumber & = \mathbb{E}[ \tilde{f}_{ki}(T_i)]
\end{align}
where \eqref{eq: consistency linearization 2} hold since $Y_{ki} = \tilde{f}_{ki}(T_i)$ if $i \in \Iscr$ and $X \in R_i$, which is implied by $\bar{Z}_i \in \{0,1\}$. It follows that $\mathbb{E}a$ and $\mathbb{E}b$ equal
\begin{align}
\mathbb{E}a & = \left[\begin{array}{l} \dfrac{1}{N} \displaystyle\sum_{i \in \Iscr} \mathbb{E}[\tilde{f}_{ki}(T_i)]  \\
\dfrac{1}{N} \displaystyle\sum_{i \in \Iscr} \mathbb{E}[1-\tilde{f}_{ki}(T_i)] 
 \end{array}\right],  & \mathbb{E}b = \left[\begin{array}{l} 1 \\ 1 \end{array}\right] \label{eq: linearization expected terms}
\end{align}
and that $\hat{\tau}_k^\Haj$ and $\tau$ are given by
\begin{align*}
\hat{\tau}_k^\Haj & = N\left(\frac{a_1}{b_1} + \frac{a_2}{b_2}\right), & \tau & = N \left(\frac{\mathbb{E}a_1}{\mathbb{E}b_1} + \frac{\mathbb{E}a_2}{\mathbb{E}b_2}\right)
\end{align*}
and the estimation error may be written as
\begin{align*}
\frac{1}{N}\left(\hat{\tau}_k^\Haj - \tau\right) & = \frac{\epsilon_1}{b_1} + \frac{\epsilon_2}{b_2}  
\end{align*}
where $\epsilon \in \mathbb{R}^2$ denotes
\begin{align}
\epsilon = \left[\begin{array}{l} a_1 - b_1\dfrac{\mathbb{E}a_1}{\mathbb{E}b_1} \\ \\ a_2 - b_2 \dfrac{\mathbb{E}a_2}{\mathbb{E}b_2} \end{array}\right] \label{eq: epsilon}
\end{align}
which can be seen to satisfy $\mathbb{E}\epsilon = 0$.

To show \eqref{eq: hajek consistency Phi}, we define the mapping $h$ by
\[ h(\epsilon, b)  = \frac{\epsilon_1}{b_1} + \frac{\epsilon_2}{b_2} \]
so that $N^{-1}(\hat{\tau}_k^\Haj - \tau) = h(\epsilon,b)$. Letting $\xi = (\epsilon, b)$, by mean value theorem it holds for some $\tilde{\xi}$ in the convex hull of $\xi$ and $\mathbb{E}\xi$ that
\begin{align} 
\nonumber h(\xi) - h(\mathbb{E}\xi) & = \nabla h(\tilde{\xi})^T(\xi - \mathbb{E}\xi)  \\
& = \nabla h(\mathbb{E}\xi)^T(\xi - \mathbb{E}\xi) + \left(\nabla h(\tilde{\xi}) - \nabla h(\mathbb{E}\xi)\right)^T(\xi - \mathbb{E}\xi) \label{eq: hajek linearization}
\end{align}
where the derivative $\nabla h(\xi)$ is given by
\begin{align} \label{eq: nabla h}
\nabla h(\xi) = \left[ \dfrac{1}{b_1}\ \ \ \dfrac{1}{b_2} \ \ -\frac{\epsilon_1}{b_1^2} \ \ -\frac{\epsilon_2}{b_2^2} \right]
\end{align}
As $\xi$ is an affine transformation of $\bar{u}$ defined in Lemma \ref{le: hajek mcdiarmid}, it follows from that Lemma that
\begin{align} \label{eq: xi consistency}
\| \xi - \mathbb{E}\xi \|_\infty = O_P(M^{-1/2})
\end{align}
from which it follows that 
\begin{align} \label{eq: derivative consistency}
\nabla h(\mathbb{E}\xi)^T(\xi - \mathbb{E}\xi) = O_P(M^{-1/2})
\end{align}
where we have used $\nabla h(\mathbb{E}\xi) = \left(1, 1, 0, 0\right)$.

As $\tilde{\xi}$ is in the convex hull of $\xi$ and $\mathbb{E}\xi$, \eqref{eq: xi consistency} implies that
$\| \tilde{\xi} - \mathbb{E}\xi\| = O_P(M^{-1/2})$. By continuity of $\nabla h$ at $\mathbb{E}\xi$, it follows that
$\| \nabla h(\tilde{\xi}) - \nabla h(\mathbb{E}\xi)\| = O_P(M^{-1/2})$ as well, and hence that
\begin{align} \label{eq: remainder consistency}
\left(\nabla h(\tilde{\xi}) - \nabla h(\mathbb{E}\xi)\right)^T(\xi - \mathbb{E}\xi) = O_P(M^{-1})
\end{align}
Substituting \eqref{eq: derivative consistency} and \eqref{eq: remainder consistency} into \eqref{eq: hajek linearization} proves that $N^{-1}(\hat{\tau}_k - \tau) = O_P(M^{-1/2})$, which is equivalent to the claim of \eqref{eq: hajek consistency Phi} made by the lemma.
\end{proof}

\subsection{Proof of Theorem \ref{th: hajek variance estimation}} \label{sec: variance consistency} \label{sec: proofs 4}

In our proof of Theorem \ref{th: hajek variance estimation}, we will use $C_{ij}(T_i, T_j)$ to abbreviate the covariance estimates appearing in $\hat{V}_k$. It is given by
\begin{align} \label{eq: Cij proofs}
C_{ij}(T_i, T_j) = \mathbb{E}[v_i v_j |T_i, T_j] - \mathbb{E}[v_i|T_i] \mathbb{E}[v_j|T_j] \frac{P(T_i)P(T_j)}{P(T_i, T_j)}
\end{align}
To prove the theorem, we will use Lemmas \ref{le: hajek Vij} - \ref{le: Gamma overlap}, whose proofs will be given following the proof of Theorem \ref{th: hajek variance estimation}. Lemma \ref{le: hajek Vij} bounds the values of $\{C_{ij}\}$.

\begin{lemma} \label{le: hajek Vij}
Let Assumption \ref{as: design} hold and let $C_{ij}$ be given by \eqref{eq: Cij proofs}. Let $\mathcal{J}_0$ and $\mathcal{J}_1$ denote the sets
\begin{align*}
\mathcal{J}_0 &= \{i,j \in \Iscr: |\Gamma_i \cap \Gamma_j| = 0\} \\
\mathcal{J}_1 &= \{i,j \in \Iscr: |\Gamma_i \cap \Gamma_j| \neq 0 \}
\end{align*}
It holds w.p. 1 that 
\begin{align}
\max_{i,j \in \mathcal{J}_0} C_{ij}(T_i, T_j) & \leq O\left(\frac{\log M}{M}\right) \label{eq: hajek Vij 1} \\
\max_{i,j \in \mathcal{J}_1} C_{ij}(T_i, T_j) & \leq O(1) \label{eq: hajek Vij 2}  
\end{align}
\end{lemma}

Lemmas \ref{le: V1 bias 1} and \ref{le: V1 bias 2} are used to prove \eqref{eq: hajek variance bias} of Theorem \ref{th: hajek variance estimation}. 

\begin{lemma} \label{le: V1 bias 1}
Let $C_{ij}$ be given by \eqref{eq: Cij proofs}. If $i,j \in \Iscr$ and $P(\bar{Z}_i = Y_{ki}, \bar{Z}_j = Y_{kj}|T_i, T_j) = 0$, then $C_{ij}(T_i, T_j) \leq 0$
\end{lemma}

\begin{lemma} \label{le: V1 bias 2}
It holds that $\mathbb{E}[C_{ij}(T_i, T_j)] \geq \Cov(v_i, v_j)$.
\end{lemma}

Lemmas  \ref{le: hajek V1 terms} and \ref{le: hajek V1} bound terms appearing in our variance estimators, and will be used to prove Lemma \ref{le: hajek Vij}.

\begin{lemma} \label{le: hajek V1 terms}
Let Assumption \ref{as: design} hold, and let $v_i$ for $i \in [N]$ denote
\[ v_i = \frac{N}{\hat{N}_{Y_{ki}}} \cdot \frac{1\{\bar{Z}_i = Y_{ki}\}}{P(\bar{Z}_i = Y_{ki}|T_i)}, \]
where $\hat{N}_a = \sum_{j=1}^N 1\{\bar{Z}_i = a\}/P(\bar{Z}_i=a|T_i)$ for $a \in \{0,1\}$. Then
\begin{align} \label{eq: hajek vi bound 1}
\mathbb{E}\left[ v_i |T_i\right] & = \mathbb{E}\left[\frac{N}{\hat{N}_{Y_{ki}}}|T_i, \bar{Z}_i = Y_{ki}\right].
\end{align}
If $i,j \in \Iscr$ and $|\Gamma_i \cap \Gamma_j| = 0$, then
\begin{align} \label{eq: hajek vivj bound 1}
\mathbb{E}\left[ v_i v_j|T_i, T_j\right] & = \mathbb{E}\left[ \frac{N^2}{\hat{N}_{Y_{ki}} \hat{N}_{Y_{kj}}} | T_i, T_j, \bar{Z}_i = Y_{ki}, \bar{Z}_j = Y_{kj}\right].
\end{align}
If $i,j \in \Iscr$ and $|\Gamma_i \cap \Gamma_j| \neq 0$, then
\begin{align} \label{eq: hajek vivj bound 2}
\mathbb{E}\left[ v_i, v_j|T_i T_j\right] & \leq \frac{1}{\rho^2} \mathbb{E}\left[ \frac{N^2}{\hat{N}_{Y_{ki}} \hat{N}_{Y_{kj}}} | T_i , T_j, \bar{Z}_i = Y_{ki}, \bar{Z}_j = Y_{kj} \right].
\end{align}
\end{lemma}

\begin{lemma} \label{le: hajek V1}
Let Assumption \ref{as: design} hold. Then 
\begin{align}
\max_{i \in \Iscr}  \mathbb{E}\left[ \frac{N}{\hat{N}_{Y_{ki}}}| T_i, \bar{Z}_i = Y_{ki}\right] & =  1 + O\left(\frac{\log M}{M}\right) \label{eq: hajek V1 1}\\
\max_{i,j \in \Iscr} \mathbb{E}\left[ \frac{N}{\hat{N}_{Y_{ki}}}| T_i, \bar{Z}_i = Y_{ki}, T_j, \bar{Z}_j = Y_{kj}\right] & = 1 + O\left(\frac{\log M}{M}\right) \label{eq: hajek V1 2} \\
\max_{i,j \in \Iscr} \Var\left[ \frac{N}{\hat{N}_{Y_{ki}}}| T_i, \bar{Z}_i = Y_{ki}, T_j, \bar{Z}_j = Y_{kj}\right] & = O\left(\frac{1}{M}\right) \label{eq: hajek V1 3}
\end{align}
\end{lemma}

Lemma \ref{le: Gamma overlap} bounds the extent to which the design neighborhoods $\{\Gamma_i\}_{i=1}^N$ may overlap, and the number of design neighborhoods to which a group $G_m$ may belong.
\begin{lemma}\label{le: Gamma overlap}
Let Assumption \ref{as: design} hold. It holds that
\begin{align}
\label{eq: d2 bound} \max_{i \in [N]} \sum_j 1\{|\Gamma_i \cup \Gamma_j| \neq 0\} & \leq \frac{d^2 CN}{M}  \\
\label{eq: d bound} \max_{m \in [M]} \sum_i 1\{G_m \subseteq \Gamma_i \} & \leq \frac{dCN}{M}
\end{align}

\end{lemma}

\paragraph{Proof of Theorem \ref{th: hajek variance estimation}} The claims of Theorem \ref{th: hajek variance estimation}, namely \eqref{eq: hajek variance concentration}, \eqref{eq: hajek variance bias}, and the conditions under which $\hat{V}_k$ is unbiased, are proven separately as follows.

\begin{proof}[Proof of Theorem \ref{th: hajek variance estimation}, Eq. \eqref{eq: hajek variance concentration}]
For $C_{ij}$ given by \eqref{eq: Cij proofs}, let $Q \in \mathbb{R}^{N \times N}$ denote the matrix given by
\begin{align}\label{eq: Q hajek variance estimation}
Q_{ij} = \begin{cases} 
C_{ij} \cdot \dfrac{1\{\bar{Z}_i = Y_{ki}, \bar{Z}_j = Y_{kj}\}}{P(\bar{Z}_i = Y_{ki}, \bar{Z}_j = Y_{kj}| T_i, T_j)} & \textup{ if $i,j \in \Iscr$} \\
0 & \textup{ else }
\end{cases} 
\end{align}
so that $\hat{V}_k = \sum_{ij} Q_{ij}$. 

To bound $|Q_{ij}|$, we may apply Lemma \ref{le: hajek Vij} and \eqref{eq: T support} to bound $|C_{ij}|$ and then \eqref{eq: ZZ support} to bound the term  $1/P(\bar{Z}_i = Y_{ki}, \bar{Z}_j = Y_{kj}|T_i, T_j)$ appearing in \eqref{eq: Q hajek variance estimation}, yielding for some $C_1 > 0$ that 
\begin{align}\label{eq: Mij bound}
|Q_{ij}| \leq \begin{cases} \displaystyle{\frac{C_1 \log M}{M}} & i,j \in \mathcal{J}_0 \\
\displaystyle{C_1} & i,j \in \mathcal{J}_1. \end{cases}
\end{align}
By \eqref{eq: d2 bound} it holds that 
\begin{align}\label{eq: Mij bound 2}
\max_{i \in [N]} |\{j: (i,j) \in \mathcal{J}_1 \}| \leq \frac{d^2 CN}{M} .
\end{align}

Changing $X_{G_m}$ changes at most $\frac{dCN}{M}$ rows and columns of $Q$, by \eqref{eq: d bound}. Combining \eqref{eq: Mij bound} and \eqref{eq: Mij bound 2} implies that each row or column of $Q$ has entries summing to at most $C_2\frac{N}{M}(\log M + 1)$ for some $C_2 > 0$.  It thus follows for some $C_3 > 0$ that for all $X, X'$ differing only in $X_{G_m}$ (for any $m \in [M]$) that
\[ |\hat{V}_k(X) - \hat{V}_k(X')| \leq C_3 \frac{N^2}{M^2}(\log M+1)\]
and it follows by McDiarmid's inequality that
\[ P(|\hat{V}_k - \mathbb{E}\hat{V}_k| \geq \epsilon) \leq 2 \exp \left(\frac{-2\epsilon^2}{M (C C_1)^2 N^4 M^{-4}(\log M + 1)^2}\right)\]
which implies that $|V - \mathbb{E}\hat{V}_k| = O_P\left(\frac{N^2\log M}{M^{3/2}}\right)$.
\end{proof}

\begin{proof}[Proof of Theorem \ref{th: hajek variance estimation}, Eq. \eqref{eq: hajek variance bias}]
Recall that $C_{ij}$ is given by \eqref{eq: Cij proofs}. To show that $\hat{V}_k$ is not anti-conservative, we observe that
\begin{align}
\mathbb{E}\hat{V}_k & = \sum_{i,j \in \Iscr}\ \mathbb{E}\left[ \mathbb{E}\left[ C_{ij}(T_i, T_j)\cdot \frac{1\{\bar{Z}_i = Y_{ki}, \bar{Z}_j = Y_{kj}\}}{P(\bar{Z}_i = Y_{ki}, \bar{Z}_j = Y_{kj}| T_i, T_j)}\ \bigg| T_i, T_j\right] \right] \nonumber \\
& = \sum_{i,j \in \Iscr}\ \mathbb{E}\left[ C_{ij}(T_i, T_j)\cdot \mathbb{E}\left[  \frac{1\{\bar{Z}_i = Y_{ki}, \bar{Z}_j = Y_{kj}\}}{P(\bar{Z}_i = Y_{ki}, \bar{Z}_j = Y_{kj}| T_i, T_j)}\ \bigg| T_i, T_j\right] \right] \nonumber \\
& = \sum_{i,j \in \Iscr}\ \sum_{t_i, t_j} P(T_i=t_i, T_j=t_j) \cdot C_{ij}(t_i, t_j)\cdot \nonumber \\
& \hskip1cm 1\{P(\bar{Z}_i = Y_{ki}, \bar{Z}_j = Y_{kj}|T_i = t_i, T_j=t_j) \neq 0\} \nonumber \\ 
& \geq \sum_{i,j \in \Iscr}\ \sum_{t_i, t_j} P(T_i=t_i, T_j=t_j) C_{ij}(t_i, t_j) \label{eq: th V1 bias eq 1}\\
& \geq \sum_{i,j \in \Iscr}\ \Cov(v_j, v_j) \label{eq: th V1 bias eq 2}
\end{align}
where \eqref{eq: th V1 bias eq 1} holds by Lemma \ref{le: V1 bias 1}, and \eqref{eq: th V1 bias eq 2} holds by Lemma \ref{le: V1 bias 2}. It can be seen that \eqref{eq: th V1 bias eq 2} equals the claim of \eqref{eq: hajek variance bias} made by Theorem \ref{th: hajek variance estimation}.
\end{proof}

\begin{proof}[Proof of Theorem \ref{th: hajek variance estimation}, unbiasedness claim]

To verify the conditions for unbiasedness given by Theorem \ref{th: hajek variance estimation}, we observe that if $T_i$ is the constant mapping for all units, then $\hat{V}_k$ simplifies to
\begin{align*}
\hat{V}_k = \sum_{i,j \in \Iscr} {\Cov}\left( v_i, v_j\right) \cdot 
 \frac{1\{\bar{Z}_i = Y_{ki}, \bar{Z}_j = Y_{kj}\}}{P\left( \bar{Z}_i = Y_{ki}, \bar{Z}_j = Y_{kj}\right)}
\end{align*}
where $Y_{ki}$ is constant for all $i \in [N]$ such that $i \in \Iscr$. As this is an importance sampled estimator of covariance terms appearing in $\Var \tau_k^{\Haj}$, it is unbiased provided that the sampling probabilities $P\left( \bar{Z}_i = Y_{ki}, \bar{Z}_j = Y_{kj}\right)$ (which equals $P(\bar{Z}_i = Y_{ki})$ for $i = j$) are non-zero for each $i,j \in \Iscr$.

\end{proof}

\paragraph{Proof of Lemmas \ref{le: hajek Vij} - \ref{le: hajek V1}}

\begin{proof}[Proof of Lemma \ref{le: hajek Vij}]

For $i,j$ satisfying $i,j \in \Iscr$, let $F_{\{i\}}$ and $F_{\{i,j\}}$ denote sigma-fields such that 
\begin{align*}
\mathbb{E}[ \ \cdot\ | F_{\{i\}}] & = \mathbb{E}[ \ \cdot\ | T_i, \bar{Z}_i = Y_{ki}] \\
\mathbb{E}[ \ \cdot\ | F_{\{i,j\}}] & = \mathbb{E}[ \ \cdot\ | T_i, T_j, \bar{Z}_i = Y_{ki}, \bar{Z}_j = Y_{kj}] 
\end{align*}
Recall $C_{ij}$ is given by \eqref{eq: Cij proofs}. To show \eqref{eq: hajek Vij 1}, we observe that for $i,j \in \mathcal{J}_0$ that
\begin{align}
C_{ij}(T_i, T_j) & = \mathbb{E}[v_i v_j|T_i, T_j] - \mathbb{E}[v_i|T_i]\mathbb{E}[v_j|T_j]\frac{P(T_i)P(T_j)}{P(T_i, T_j)}  \nonumber \\
&= \mathbb{E}\left[\frac{N^2}{\hat{N}_{Y_{ki}}\hat{N}_{Y_{kj}}} | F_{\{i,j\}} \right] - \mathbb{E}\left[\frac{N}{\hat{N}_{Y_{ki}}} | F_{\{i\}}\right]\mathbb{E}\left[\frac{N}{\hat{N}_{Y_{kj}}}|F_{\{j\}}\right] \label{eq: Vij use le 6} \\
& = \mathbb{E}\left[\frac{N^2}{\hat{N}_{Y_{ki}}\hat{N}_{Y_{kj}}} | F_{\{i,j\}} \right] - \mathbb{E}\left[\frac{N}{\hat{N}_{Y_{ki}}} | F_{\{i,j\}}\right]\mathbb{E}\left[\frac{N}{\hat{N}_{Y_{kj}}}|F_{\{i,j\}}\right]  \nonumber  \\
& \hskip1cm - \mathbb{E}\left[\frac{N}{\hat{N}_{Y_{ki}}} | F_{\{i\}}\right]\mathbb{E}\left[\frac{N}{\hat{N}_{Y_{kj}}}|F_{\{j\}}\right] + \mathbb{E}\left[\frac{N}{\hat{N}_{Y_{ki}}} | F_{\{i,j\}}\right]\mathbb{E}\left[\frac{N}{\hat{N}_{Y_{kj}}}|F_{\{i,j\}}\right] \nonumber  \\
& = \Cov\left(\frac{N}{\hat{N}_{Y_{ki}}}, \frac{N}{\hat{N}_{Y_{kj}}}| F_{\{i,j\}}\right) - \left(\mathbb{E}\left[\frac{N}{\hat{N}_{Y_{ki}}}|F_{\{i\}}\right] - \mathbb{E}\left[\frac{N}{\hat{N}_{Y_{ki}}}|F_{\{i,j\}}\right] \right) \mathbb{E}\left[\frac{N}{\hat{N}_{Y_{kj}}}|F_{\{j\}}\right]  \nonumber  \\
& \hskip1cm - \mathbb{E}\left[\frac{N}{\hat{N}_{Y_{ki}}}|F_{\{i,j\}}\right]\left(\mathbb{E}\left[\frac{N}{\hat{N}_{Y_{kj}}}|F_{\{j\}}\right] - \mathbb{E}\left[\frac{N}{\hat{N}_{Y_{kj}}}|F_{\{i,j\}}\right] \right) \nonumber  \\
& \leq \sqrt{\Var\left(\frac{N}{\hat{N}_{Y_{ki}}}| F_{\{i,j\}}\right)\Var\left(\frac{N}{\hat{N}_{Y_{kj}}}| F_{\{i,j\}}\right)}  \nonumber \\
& \hskip1cm - \left(\mathbb{E}\left[\frac{N}{\hat{N}_{Y_{ki}}}|F_{\{i\}}\right] - \mathbb{E}\left[\frac{N}{\hat{N}_{Y_{ki}}}|F_{\{i,j\}}\right] \right) \mathbb{E}\left[\frac{N}{\hat{N}_{Y_{kj}}}|F_{\{j\}}\right]  \nonumber  \\
& \hskip1cm - \mathbb{E}\left[\frac{N}{\hat{N}_{Y_{ki}}}|F_{\{i,j\}}\right]\left(\mathbb{E}\left[\frac{N}{\hat{N}_{Y_{kj}}}|F_{\{j\}}\right] - \mathbb{E}\left[\frac{N}{\hat{N}_{Y_{kj}}}|F_{\{i,j\}}\right] \right) \nonumber  \\
& \leq O\left(\frac{\log M}{M}\right) \label{eq: Vij use le 7}
\end{align}
where \eqref{eq: Vij use le 6} uses Lemma \ref{le: hajek V1 terms}, and \eqref{eq: Vij use le 7} uses Lemma \ref{le: hajek V1}.

To show \eqref{eq: hajek Vij 2}, we observe for $i,j \in \mathcal{J}_1$ that
\begin{align}
C_{ij}(T_i, T_j) & = \mathbb{E}[v_i v_j|T_i, T_j] - \mathbb{E}[v_i|T_i]\mathbb{E}[v_j|T_j]\frac{P(T_i)P(T_j)}{P(T_i, T_j)}  \nonumber \\
& \leq \frac{1}{\rho^2} \mathbb{E}\left[ \frac{N^2}{\hat{N}_{Y_{ki}}\hat{N}_{Y_{kj}}}| F_{\{i,j\}}\right] - \mathbb{E}\left[\frac{N}{\hat{N}_{Y_{ki}}} | F_{\{i\}}\right]\mathbb{E}\left[\frac{N}{\hat{N}_{Y_{kj}}}|F_{\{j\}}\right]\cdot \frac{1}{\rho} \label{eq: Vij 2 eq 1} \\
& \leq \frac{1}{\rho^2} \left(\mathbb{E}\left[ \frac{N^2}{(\hat{N}_{Y_{ki}})^2}| F_{\{i,j\}}\right] \cdot \mathbb{E}\left[ \frac{N^2}{(\hat{N}_{Y_{kj}})^2}| F_{\{i,j\}}\right]\right)^{1/2} + \nonumber \\
& \hskip1cm \mathbb{E}\left[\frac{N}{\hat{N}_{Y_{ki}}} | F_{\{i\}}\right]\mathbb{E}\left[\frac{N}{\hat{N}_{Y_{kj}}}|F_{\{j\}}\right]\cdot \frac{1}{\rho} \nonumber  \\
& = O\left(1\right) + O\left(\frac{\log M}{M}\right) \label{eq: Vij 2 eq 2}
\end{align}
where \eqref{eq: Vij 2 eq 1} uses Lemma \ref{le: hajek V1 terms} and \eqref{eq: T support}, and \eqref{eq: Vij 2 eq 2} uses Lemma \ref{le: hajek V1} and the identity
\begin{align*}
\mathbb{E}\left[ \frac{N^2}{(\hat{N}_{Y_{ki}})^2}| F_{\{i,j\}}\right] & = \Var\left[\frac{N}{\hat{N}_{Y_{ki}}}|F_{\{i,j\}}\right] + \left(\mathbb{E}\left[\frac{N}{\hat{N}_{Y_{ki}}}|F_{\{i,j\}}\right]\right)^2 \\
& = O\left(\frac{1}{M}\right) + O(1)
\end{align*}
which holds by Lemma \ref{le: hajek V1} as well.
\end{proof}




\begin{proof}[Proof of Lemma \ref{le: V1 bias 1}]
Recall that $v_i$ and $C_{ij}$ are respectively given by \eqref{eq: v general} and \eqref{eq: Cij proofs}. If $i,j \in \Iscr$ and $P(\bar{Z}_i = Y_{ki}, \bar{Z}_j = Y_{kj}| T_i, T_j) = 0$, then
\begin{align}
\nonumber \mathbb{E}[v_iv_j |T_i, T_j] & = \mathbb{E}\left[ \frac{1\{\bar{Z}_i = Y_{ki}, \bar{Z}_j = Y_{kj}\}}{P(\bar{Z}_i = Y_{ki}|T_i) P(\bar{Z}_j = Y_{kj}|T_j)} \cdot \frac{N^2}{\hat{N}_{Y_{ki}} \hat{N}_{Y_{kj}}} | T_i, T_j\right]  \\
\label{eq: V1 bias 1 step 2} & = 0
\end{align}
where 
\eqref{eq: V1 bias 1 step 2} uses $P(\bar{Z}_i = Y_{ki}, \bar{Z}_j = Y_{kj}| T_i, T_j) = 0$. It follows that
\begin{align}
\nonumber C_{ij}(T_i, T_j) & =  \mathbb{E}\left[ v_i v_j | T_i, T_j\right] - \mathbb{E}[v_i|T_i]\cdot \mathbb{E}[v_j|T_j] \cdot \frac{P(T_i)P(T_j)}{P(T_i, T_j)} \\
& = 0 - \mathbb{E}[v_i|T_i]\cdot \mathbb{E}[v_j|T_j] \cdot \frac{P(T_i )P(T_j )}{P(T_i , T_j )} \label{eq: le V1 bias 1 eq 1}\\
& \leq  0 \label{eq: le V1 bias 1 eq 2}
\end{align}
where \eqref{eq: le V1 bias 1 eq 1} holds by \eqref{eq: V1 bias 1 step 2}, and \eqref{eq: le V1 bias 1 eq 2} holds because $v_i \geq 0$ always, which implies that $\mathbb{E}[v_i|T_i] \geq 0$.
\end{proof}

\begin{proof}[Proof of Lemma \ref{le: V1 bias 2}]
Let $\mathcal{T}_i$ denote the support of $T_i(X)$.
\begin{align}
\nonumber \nonumber & \sum_{t_i \in \mathcal{T}_i} \sum_{t_j \in \mathcal{T}_j} P(T_i = t_i, T_j = t_j) C_{ij}(t_i, t_j) \\
\nonumber & = \sum_{t_i \in \mathcal{T}_i} \sum_{t_j \in \mathcal{T}_j} P(T_i = t_i, T_j = t_j) \bigg\{ \mathbb{E}\left[v_i v_j|T_i=t_i, T_j = t_j\right] - \\ 
\label{eq: le bias 2 eq 1} & \hskip1cm \mathbb{E}[v_i|T_i = t_i]\mathbb{E}[v_j|T_j=t_j]\cdot \frac{P(T_i=t_i)P(T_j = t_j)}{P(T_i = t_i, T_j = t_j)}\bigg\} \\
\nonumber & = \mathbb{E}[v_i v_j]  - \sum_{t_i \in \mathcal{T}_i, t_j \in \mathcal{T}_j} 1\{P(T_i = t_i, T_j = t_j) \neq 0\} \cdot \\
\label{eq: le bias 2 eq 2} & \hskip1cm  \mathbb{E}[v_i|T_i = t_j]\mathbb{E}[v_j|T_j = t_j] P(T_i = t_i)P(T_j = t_j) \\
\label{eq: le bias 2 eq 3} & \geq \mathbb{E}[v_i v_j] - \sum_{t_i \in \mathcal{T}_i, t_j \in \mathcal{T}_j} \mathbb{E}[v_i|T_i = t_i]\mathbb{E}[v_j|T_j = t_j] P(T_i = t_i)P(T_j = t_j) \\
\nonumber & = \mathbb{E}[v_iv_j] - \mathbb{E}[v_i] \cdot \mathbb{E}[v_j] = \Cov(v_i,v_j)
\end{align}
where \eqref{eq: le bias 2 eq 1} holds by definition of $C_{ij}$ in \eqref{eq: Cij proofs}, and \eqref{eq: le bias 2 eq 2} holds because 
\[ 
\sum_{t_i \in \mathcal{T}_i} \sum_{t_j \in \mathcal{T}_j} P(T_i = t_i, T_j = t_j) \mathbb{E}\left[v_i v_j|T_i=t_i, T_j = t_j\right] = \mathbb{E}[v_i v_j],
\]
and \eqref{eq: le bias 2 eq 3} holds because 
\[ \mathbb{E}[v_i|T_i = t_i]\mathbb{E}[v_j|T_j = t_j] P(T_i = t_i)P(T_j = t_j) \geq 0\]
since $v_i \geq 0$ always.
\end{proof}


\begin{proof}[Proof of Lemma \ref{le: hajek V1 terms}]
Recall that $v_i$ is given by \eqref{eq: v general}. To show \eqref{eq: hajek vi bound 1}, we observe that
\begin{align*}
\mathbb{E}\left[ v_i |T_i\right] & = \mathbb{E}\left[ \frac{1\{\bar{Z}_i = Y_{ki}\}N}{P(\bar{Z}_i = Y_{ki}|T_i)\hat{N}_{Y_{ki}}}| T_i\right] \\
& = P(\bar{Z}_i = Y_{ki}|T_i) \mathbb{E}\left[\frac{N}{P(\bar{Z}_i = Y_{ki}|T_i) \hat{N}_{Y_{ki}}} | T_i, \bar{Z}_i = Y_{ki}\right] \\
& = \mathbb{E}\left[\frac{N}{\hat{N}_{Y_{ki}}}|T_i, \bar{Z}_i = Y_{ki}\right].
\end{align*}
To show \eqref{eq: hajek vivj bound 1}, we observe that if $i,j \in \Iscr$ and $|\Gamma_i \cap \Gamma_j| = 0$ then
\begin{align*}
\mathbb{E}\left[ v_i v_j|T_i T_j\right] & = \mathbb{E}\left[ \frac{1\{\bar{Z}_i = Y_{ki}\}\cdot 1\{\bar{Z}_j = Y_{kj}\}\cdot N^2}{P(\bar{Z}_i = Y_{ki}|T_i) P(\bar{Z}_j = Y_{kj}|T_j) \hat{N}_{Y_{ki}}\hat{N}_{Y_{kj}}}| T_i, T_j\right] \\
& = P(\bar{Z}_i = Y_{ki}, \bar{Z}_j = Y_{kj}|T_i, T_j) \cdot \\
& \hskip1cm \mathbb{E}\left[\frac{N^2}{P(\bar{Z}_i = Y_{ki}|T_i) P(\bar{Z}_j = Y_{kj}|T_j) \hat{N}_{Y_{ki}} \hat{N}_{Y_{kj}}} | T_i, T_j, \bar{Z}_i = Y_{ki}, \bar{Z}_j = Y_{kj}\right] \\
& = \mathbb{E}\left[ \frac{N^2}{\hat{N}_{Y_{ki}} \hat{N}_{Y_{kj}}} | T_i, T_j, \bar{Z}_i = Y_{ki}, \bar{Z}_j = Y_{kj}\right]
\end{align*}
where we have used that the events $\{\bar{Z}_i = Y_{ki}\}$ and $\{\bar{Z}_j = Y_{kj}\}$ are independent conditioned on $T_i$ and $T_j$ when $|\Gamma_i \cap \Gamma_i| = 0$.

To show \eqref{eq: hajek vivj bound 2}, we observe that if $i,j \in \Iscr$ and $|\Gamma_i \cap \Gamma_j| \neq 0$, then
\begin{align*}
\mathbb{E}\left[ v_i v_j|T_i T_j\right] & = \mathbb{E}\left[ \frac{1\{\bar{Z}_i = Y_{ki}\}\cdot 1\{\bar{Z}_j = Y_{kj}\} \cdot N^2}{P(\bar{Z}_i = Y_{ki}|T_i) P(\bar{Z}_j = Y_{kj}|T_j) \hat{N}_{Y_{ki}}\hat{N}_{Y_{kj}}}| T_i, T_j\right] \\
& = P(\bar{Z}_i = Y_{ki}, \bar{Z}_j = Y_{kj}|T_i, T_j) \cdot \\
& \hskip1cm \mathbb{E}\left[\frac{N^2}{P(\bar{Z}_i = Y_{ki}|T_i) P(\bar{Z}_j = Y_{kj}|T_j) \hat{N}_{Y_{ki}} \hat{N}_{Y_{kj}}} | T_i, T_j, \bar{Z}_i = Y_{ki}, \bar{Z}_j = Y_{kj}\right] \\
& \leq \frac{1}{\rho^2} \mathbb{E}\left[ \frac{N^2}{\hat{N}_{Y_{ki}} \hat{N}_{Y_{kj}}} | T_i, T_j, \bar{Z}_i = Y_{ki}, \bar{Z}_j = Y_{kj}\right] 
\end{align*}
where the final inequality uses \eqref{eq: Z support}.
\end{proof}

\begin{proof}[Proof of Lemma \ref{le: hajek V1}]

Given a set $S \subseteq [N]$, let $\Gamma_S$ denote the set
\[ \Gamma_S = \bigcup_{i \in S} \Gamma_i \]
and let $Q$ and $\bar{S}$ denote the sets
\begin{align*}
Q & = \{j: |\Gamma_j \cap \Gamma_S| = 0\}, & \bar{S} & = [N]/Q
\end{align*}
and given $Y_{ki}, Q$, and $\bar{S}$, let $\hat{N}_Q$ and $\hat{N}_{\bar{S}}$ denote
\begin{align} \label{eq: hat NQ and hat NS}
\hat{N}_Q & = \sum_{i\in Q} \frac{1\{\bar{Z}_j = Y_{ki}\}}{P\left(\bar{Z}_j = Y_{ki}|T_j\right)} &
\hat{N}_{\bar{S}} & = \sum_{i\in \bar{S}} \frac{1\{\bar{Z}_j = Y_{ki}\}}{P\left(\bar{Z}_j = Y_{ki}|T_j\right)}
\end{align}
so that $\hat{N}_{Y_{ki}}$ satisfies
\[ \hat{N}_{Y_{ki}} = \hat{N}_Q + \hat{N}_{\bar{S}}.\]
We observe that $\hat{N}_Q$ is independent of $X_{\Gamma_S}$, and hence that
\[ \mathbb{E}\left[\hat{N}_{Q}\ \Big|\  X_{\Gamma_S}\right] = N_{Q} \]
where $N_{Q} = N - |\bar{S}|$.

From \eqref{eq: d bound}, it can be seen that changing $X_{G_m}$ for any $m \in [M]$ changes at most $\frac{dCN}{M}$ entries in $Q$, and hence changes 
$\hat{N}_Q$ by at most $\frac{dCN}{M \rho}$, where we have used \eqref{eq: Z support}. As a result, by McDiarmid's inequality it holds for $C_1 = dC/\rho$ that
\begin{align} \label{eq: B mcdiarmid}
P\left(\left| \hat{N}_Q - N_{Q} \right| \geq \epsilon\ \Big| \ X_{\Gamma_S} \right) \leq \exp\left(-\frac{\epsilon^2 M}{C_1^2 N^2}\right)
\end{align}
Thus for $\epsilon_B = \frac{C_1N}{\sqrt{M}}\sqrt{(2\delta^{-1}+1) \log M}$ and the indicator
\[ B = 1\left\{\left| \hat{N}_Q - N_{Q}\right| \geq \epsilon_B \right\},\]
it follows from \eqref{eq: B mcdiarmid} that 
\begin{align} 
\nonumber P(B = 1\ |\ X_{\Gamma_S}) & \leq \frac{1}{M^{\frac{2}{\delta}+1}} \\
& \leq \frac{1}{N^{2 + \delta}} \label{eq: B prob}
\end{align}
where the last inequality uses that $|G_m| \leq N^{1-\delta}$, which implies $M \geq N^\delta$.

For $i,j \in \Iscr$, we let $F_{\{i\}}$ and $F_{\{i,j\}}$ denote sigma-fields such that
\begin{align*}
\mathbb{E}[\ \cdot \ | F_{\{i\}}] & = \mathbb{E}[\ \cdot \ | T_i, \bar{Z}_i = Y_{ki}] \\
\mathbb{E}[\ \cdot \ | F_{\{i,j\}}] & = \mathbb{E}[\ \cdot \ | T_i, T_j, \bar{Z}_i = Y_{ki}, \bar{Z}_j = Y_{kj}] 
\end{align*}
as previously used in the proof of Lemma \ref{le: hajek Vij}. Given $S$ equal to either $\{i\}$ or $\{i,j\}$, we can induce $F_S$, $Q$, and $\bar{S}$, and observe that 
\begin{align}
\mathbb{E}\left[ \frac{N}{\hat{N}_{Y_{ki}}}| F_S\right] & = \mathbb{E}\left[ 1 + \frac{N - \hat{N}_{Y_{ki}}}{N} + \frac{(N - \hat{N}_{Y_{ki}})^2}{N \hat{N}_{Y_{ki}}}| F_S \right] \label{eq: my linearization} \\
\nonumber & \leq 1 + \left|\mathbb{E}\left[\frac{|\bar{S}| - \hat{N}_{\bar{S}}}{N}|F_S\right]\right| + \left|\mathbb{E}\left[\frac{N_{Q} - \hat{N}_Q}{N}|F_S\right]\right|  +  \\ 
 & \hskip1cm \left|\mathbb{E}\left[\frac{(N - \hat{N}_{Y_{ki}})^2}{N \hat{N}_{Y_{ki}}} | F_S\right] \right| \label{eq: decompose N - Nhat} \\
 & \leq 1 + \frac{|\bar{S}|}{N\rho} + 0 + \left|\mathbb{E}\left[\frac{(N - \hat{N}_{Y_{ki}})^2}{N \hat{N}_{Y_{ki}}} | F_S, B = 0\right]\right| +\nonumber  \\
& \hskip1cm P(B=1|F_S)\cdot \left| \mathbb{E}\left[\frac{(N - \hat{N}_{Y_{ki}})^2}{N \hat{N}_{Y_{ki}}} | F_S, B = 1\right] \right| \label{eq: bound NS} \\
& \leq 1 + \frac{|\bar{S}|}{N\rho} + \frac{(|\bar{S}|/\rho + \epsilon_B)^2}{N(N_{Q} - \epsilon_B)} + \frac{1}{N^{2 + \delta}}\cdot\left(\frac{N}{\rho}\right)^2\cdot \frac{1}{N} \label{eq: substitute S and B} \\
& = 1 + O\left(\frac{1}{M}\right) + O\left(\frac{\frac{N^2}{M^2} + \frac{N^2}{M}\log M}{N^2}\right) + O\left(\frac{1}{N^{1+\delta}}\right) \label{eq: bound S and B} \\
& = 1 + O\left(\frac{\log M}{M}\right) \label{eq: bound S and B 2}
\end{align}
where \eqref{eq: my linearization} holds by the identity $\frac{a}{b} = 1 + \frac{a-b}{a} + \frac{(a-b)^2}{ab}$; \eqref{eq: decompose N - Nhat} decomposes $N - \hat{N}_{Y_{ki}} = |\bar{S}| - \hat{N}_{\bar{S}} + N_Q - \hat{N}_Q$; 
\eqref{eq: bound NS} uses the bound
\[ \left|\mathbb{E}\left[\frac{|\bar{S}| - \hat{N}_{|\bar{S}|}}{N} | F_S \right] \right| \leq \frac{|\bar{S}|}{N \rho},\]
which holds because  $0 \leq \hat{N}_{\bar{S}} \leq \frac{|\bar{S}|}{\rho}$, and uses 
\[ \mathbb{E}\left[ \frac{N_Q - \hat{N}_Q}{N}|F_S\right] = 0, \]
which holds because $\mathbb{E}[\hat{N}_Q|F_S] = N_Q$; 
\eqref{eq: substitute S and B} uses $P(B=1|F_S) \leq 1/N^{2+\delta}$ which holds by \eqref{eq: B prob}, and also uses
\begin{align} \label{eq: upper bound squared frac}
\frac{(N - \hat{N}_{Y_{ki}})^2}{N \hat{N}_{Y_{ki}}} \leq \begin{cases} \displaystyle{\frac{(|\bar{S}|/\rho + \epsilon_B)^2}{N (N_Q - \epsilon_B)}} & \textup{ if } B = 0 \\ \\ \displaystyle{\left( \frac{N}{\rho}\right)^2  \cdot \frac{1}{N}} & \textup{ if } B = 1 \end{cases}
\end{align}
which will be proven after the main result; \eqref{eq: bound S and B} holds because $|\bar{S}| \leq \frac{2CN}{M}$ since $|S| \leq 2$, and also by definition of $\epsilon_B$, and $N_Q \geq N - \frac{2CN}{M}$.  Letting $S=\{i\}$ implies \eqref{eq: hajek V1 1}, while letting $S = \{i,j\}$ implies \eqref{eq: hajek V1 2}.

To show \eqref{eq: hajek V1 3}, we let $S= \{i,j\}$ (inducing $F_S$, $Q$, and $\bar{S}$), and observe that
\begin{align}
\Var\left(\frac{N}{\hat{N}_{Y_{ki}}}|F_S\right) & = \Var\left[1 + \frac{|\bar{S}| - \hat{N}_{\bar{S}}}{N} + \frac{N_{Q} - \hat{N}_Q}{N} + \frac{(N - \hat{N}_{Y_{ki}})^2}{N \hat{N}_{Y_{ki}}}| F_S\right] \nonumber \\
& = 5 \Var\left[\frac{|\bar{S}| - \hat{N}_{\bar{S}}}{N} | F_S\right] + 5 \Var \left[ \frac{N_{Q} - \hat{N}_Q}{N}  | F_S \right] + 5 \Var\left[ \frac{(N - \hat{N}_{Y_{ki}})^2}{N \hat{N}_{Y_{ki}}}| F_S\right] \label{eq: var bound}\\
& \leq \left(\frac{|\bar{S}|}{N \rho}\right)^2 + \frac{Cd}{\rho^2 M} + \mathbb{E}\left[\Var\left(\frac{(N-\hat{N}_{Y_{ki}})^2}{N \hat{N}_{Y_{ki}}} | F_S, B\right) | F_S \right] +  \nonumber \\
& \hskip1cm \Var\left( \mathbb{E}\left[ \frac{(N - \hat{N}_{Y_{ki}})^2}{N \hat{N}_{Y_{ki}}}| F_S, B\right] | F_S\right) \label{eq: var bound 1} \\
& \leq \left(\frac{|\bar{S}|}{N \rho}\right)^2 + \frac{Cd}{\rho^2 M}  + \left(\frac{(|\bar{S}|/\rho+ \epsilon_B)^2}{N(N_{Q} - \epsilon_B)}\right)^2 + \left(\frac{N^2}{\rho^2N}\right)^2\cdot\frac{1}{N^{2 + \delta}} \label{eq: var bound 2}  \\
& = O\left(\frac{1}{M^2}\right) + O\left(\frac{1}{M}\right) + O\left(\left[\frac{\log M}{M}\right]^2\right) + O\left(\frac{1}{M}\right) \label{eq: var bound 3} \\
& = O\left(\frac{1}{M}\right) \label{eq: var bound 4}
\end{align}
where \eqref{eq: var bound} follows by the identity
\begin{align*}
\Var \sum_{i=1}^3 a_i & = \sum_{i=1}^3 \sum_{j=1}^3 \Cov(a_i, a_j)  \\
 & \leq \sum_{i=1}^3 \Var a_i + \sum_{i=1}^3 \sum_{j: j\neq i} (\Var a_i + \Var a_j),
\end{align*}
\eqref{eq: var bound 1}  uses the bounds
\begin{align}
\Var\left[\frac{|\bar{S}| - \hat{N}_{\bar{S}}}{N}|  F_S\right] & \leq \left( \frac{|\bar{S}|}{N \rho}\right)^2 \label{eq: var bound 1.1} \\
\Var \left[ \frac{N_{Q} - \hat{N}_Q}{N}  | F_S \right] & \leq \frac{C N^2 d^2}{\rho^2 M N^2} = \frac{Cd^2}{\rho^2 M},  \label{eq: var bound 1.2}
\end{align}
which will be proven after the main result, and also uses the law of total variance to decompose $\Var\left[ \frac{(N - \hat{N}_{Y_{ki}})^2}{N \hat{N}_{Y_{ki}}}\right]$;  
\eqref{eq: var bound 2} uses the bounds
\begin{align}
\label{eq: var bound 2.1}  \mathbb{E}\left[\Var\left(\frac{(N-\hat{N}_{Y_{ki}})^2}{N \hat{N}_{Y_{ki}}} | F_S, B\right) | F_S \right] & \leq \left(\frac{(|\bar{S}|/\rho + \epsilon_B)^2}{N(N_Q - \epsilon_B)}\right)^2 + \frac{1}{N^{2 + \delta}} \left(\frac{N^2}{\rho N}\right)^2\\
\label{eq: var bound 2.2} \Var\left( \mathbb{E}\left[ \frac{(N - \hat{N}_{Y_{ki}})^2}{N \hat{N}_{Y_{ki}}}| F_S, B\right] | F_S\right)  & \leq \frac{1}{N^{2 + \delta}} \left(\frac{N^2}{\rho N}\right)^2
 \end{align}
 which will be proven shortly; \eqref{eq: var bound 3} holds because $|\bar{S}| \leq \frac{2CM}{N}$, by definition of $\epsilon_B$, and $N_Q \geq N - \frac{2CM}{N}$.

It remains to show \eqref{eq: upper bound squared frac}, \eqref{eq: var bound 1.1} -- \eqref{eq: var bound 1.2}, and \eqref{eq: var bound 2.1} -- \eqref{eq: var bound 2.2}:
\begin{enumerate}
  \item To show \eqref{eq: upper bound squared frac}, we let $\textup{card}(\bar{S})$ denote $|\bar{S}|$, and observe that if $B=0$ then it holds that
\begin{align}
\nonumber \frac{(N - \hat{N}_{Y_{ki}})^2}{N \hat{N}_{Y_{ki}}} & = \frac{(|\bar{S}| - \hat{N}_{\bar{S}} + N_Q - \hat{N}_Q)^2}{N(\hat{N}_{\bar{S}} + \hat{N}_Q)} \\ 
\nonumber & \leq  \frac{(| \textup{card}(\bar{S})  - \hat{N}_{\bar{S}}| + |N_Q - \hat{N}_Q|)^2}{N (\hat{N}_{\bar{S}} + \hat{N}_Q)}  \\
\label{eq: substitute S and B 1.1} & \leq \frac{(|\bar{S}|/\rho + \epsilon_B)^2}{N (N_Q - \epsilon_B)}, 
\end{align}
where \eqref{eq: substitute S and B 1.1} uses that $|\hat{N}_Q - N_Q| < \epsilon_B$ when $B = 0$, and also that $\hat{N}_{\bar{S}} \geq 0$ to lower bound the denominator, and also that $\hat{N}_{\bar{S}} \leq |\bar{S}|/\rho$ which implies $| \textup{card}(\bar{S}) - \hat{N}_{\bar{S}}| \leq |\bar{S}|/\rho$.

On the other hand, if $B=1$, it holds conditional on $F_S$
\begin{align}
\label{eq: substitute S and B 2.1} \frac{(N - \hat{N}_{Y_{ki}})^2}{N \hat{N}_{Y_{ki}}} & \leq \left( \frac{N}{\rho}\right)^2  \cdot \frac{1}{N}
\end{align}
where the numerator uses $\hat{N}_{Y_{ki}} \leq N/\rho$ and the denominator uses that $\hat{N}_{Y_{ki}} \geq 1$ conditional on $F_S$. Combining \eqref{eq: substitute S and B 1.1} and \eqref{eq: substitute S and B 2.1} proves \eqref{eq: upper bound squared frac}.

\item To show \eqref{eq: var bound 1.1} we observe that $0 \leq \hat{N}_{\bar{S}} \leq \frac{|\bar{S}|}{\rho}$.

\item To show \eqref{eq: var bound 1.2} we observe that $\hat{N}_Q/N$ given by \eqref{eq: hat NQ and hat NS} is a sum of $\leq N$ terms whoses magnitudes are bounded by $(N\rho)^{-1}$ and whose dependency graph has at most $N \cdot \frac{d^2CN}{M}$ edges, where we have used \eqref{eq: d2 bound}.

\item To show \eqref{eq: var bound 2.1}, we observe that
 \begin{align}
  \nonumber & \mathbb{E}\left[\Var\left(\frac{(N-\hat{N}_{Y_{ki}})^2}{N \hat{N}_{Y_{ki}}} | F_S, B\right) | F_S \right] \\
\nonumber &  = P(B=0|F_S) \Var\left(\frac{(N-\hat{N}_{Y_{ki}})^2}{N \hat{N}_{Y_{ki}}} | F_S, B=0\right) + \\
\nonumber & \hskip1cm P(B=1|F_S)\Var\left(\frac{(N-\hat{N}_{Y_{ki}})^2}{N \hat{N}_{Y_{ki}}} | F_S, B=1\right) \\
& \leq \left(\frac{(|\bar{S}|/\rho  + \epsilon_B)^2}{N(N_Q - \epsilon_B)}\right)^2 + \frac{1}{N^{2 + \delta}}\left(\frac{N^2}{\rho N}\right)^2  \label{eq: var bound 2.1.1}
\end{align}
where \eqref{eq: var bound 2.1.1} holds by \eqref{eq: substitute S and B 1.1}, \eqref{eq: substitute S and B 2.1}, and \eqref{eq: B prob}.

\item To show \eqref{eq: var bound 2.2}, we observe that 
\begin{align}
\nonumber & \Var\left( \mathbb{E}\left[ \frac{(N - \hat{N}_{Y_{ki}})^2}{N \hat{N}_{Y_{ki}}}| F_S, B\right] | F_S\right) \\
& \leq  P(B=1|F_S)P(B=0|F_S) \cdot  \left( 2 \max \left\{ \frac{(|\bar{S}|/\rho  + \epsilon_B)^2}{N(N_Q - \epsilon_B)}\, , \, \frac{N}{\rho}\right\} \right)^2 \nonumber \\
\nonumber & \leq \frac{1}{N^{2+\delta}} \left(\frac{N}{\rho}\right)^2
\end{align}
which uses \eqref{eq: substitute S and B 1.1}, \eqref{eq: substitute S and B 2.1}, \eqref{eq: B prob},  and that the variance of $a + (a-b)U$, where $U \sim \textup{Bernoulli}(p)$ and $a,b$ are scalars, equals $p(1-p)(a-b)^2$
\end{enumerate}

This shows \eqref{eq: hajek V1 1}, \eqref{eq: hajek V1 2}, and \eqref{eq: hajek V1 3}, as required by the Lemma.

\end{proof}

\begin{proof}[Proof of Lemma \ref{le: Gamma overlap}]

To prove \eqref{eq: d2 bound}, we observe that for any $i \in [N]$, it holds that
\begin{align}
\label{eq: d2 1} \sum_{j=1}^N 1\{|\Gamma_i \cup \Gamma_j| \neq 0\} & = \sum_{j=1}^N 1\{|\eta_{g(i)} \cup \eta_{g(j)}| \neq 0\}  \\
\label{eq: d2 2} & = \sum_{j=1}^N \sum_{m=1}^M 1\{g(j)=m\} \cdot 1\{|\eta_{g(i)} \cup \eta_{m}| \neq 0\} \\
\label{eq: d2 3} & = \sum_{m=1}^M  1\{|\eta_{g(i)} \cup \eta_{m}| \neq 0\}  \sum_{j=1}^N 1\{g(j)=m\} \\
\label{eq: d2 4} & \leq \sum_{m=1}^M  1\{|\eta_{g(i)} \cup \eta_{m}| \neq 0\} \cdot \frac{CN}{M} \\
\label{eq: d2 5} & \leq d^2\frac{CN}{M}
\end{align}
where \eqref{eq: d2 1} holds by \eqref{eq: design neighborhood}, \eqref{eq: d2 2} and \eqref{eq: d2 3} are algebraic manipulations, \eqref{eq: d2 4} holds by \eqref{eq: group size}, and \eqref{eq: d2 5} holds because
\begin{align}\label{eq: eta dependency}
\sum_{m=1}^M 1\{ | \eta_{g(i)} \cap \eta_m| \neq 0\} \leq d^2
\end{align}
which holds because a group neighborhood $\eta_{g(i)}$ contains at most $d$ groups by \eqref{eq: eta size}, each of which are members of at most $d$  group neighborhoods by \eqref{eq: eta overlap}.

To prove \eqref{eq: d bound}, we observe that for any $m \in [M]$, it holds that
\begin{align}
\label{eq: d 1}\sum_{i=1}^N 1\{G_m \subseteq \Gamma_i\} & = \sum_{i=1}^N 1\{m \in \eta_{g(i)}\} \\
\label{eq: d 2} & = \sum_{i=1}^N \sum_{\ell=1}^M 1\{g(i) = \ell\} \cdot 1\{m \in \eta_{\ell}\} \\
\label{eq: d 3} & = \sum_{\ell=1}^M 1\{m \in \eta_\ell\} \cdot \sum_{i=1}^N 1\{g(i) = \ell\} \\
\label{eq: d 4} & \leq \sum_{\ell=1}^N 1\{m \in \eta_\ell\} \cdot \frac{CN}{M} \\
\label{eq: d 5} & \leq d\frac{CN}{M}
\end{align}
where \eqref{eq: d 1} holds by \eqref{eq: design neighborhood}, \eqref{eq: d 2} and \eqref{eq: d 3} are algebraic manipulations, \eqref{eq: d 4} holds by \eqref{eq: group size}, and \eqref{eq: d 5} holds by \eqref{eq: eta size}.
\end{proof} 

\subsection{Proof of Theorem \ref{th: hajek coverage}} \label{sec: coverage}

To prove Theorem \ref{th: hajek coverage}, we will use Lemmas \ref{le: Phi coverage}-\ref{le: fatou}, and also Theorem \ref{th: ross} which is a local dependence central limit theorem. 

Lemma \ref{le: Phi coverage} gives a coverage result for $\hat{\tau}_k$.
\begin{lemma} \label{le: Phi coverage}
Let Assumption \ref{as: design} hold, and let $\widetilde{V}_k(\phi)$ for $\phi \in \{0,1\}^N$ and $k=1,2$ be given by \eqref{eq: V tilde}. It holds for $k=1,2$ that
\[ P\left( \hat{\tau}_k^\Haj - \tau \geq z_{1-\frac{\alpha}{2}} \sqrt{\widetilde{V}_k} \right)\leq \frac{\alpha}{2} + o(1) \]
\end{lemma}

Lemmas \ref{le: hajek CLT} and \ref{le: fatou} are used to prove Lemma \ref{le: Phi coverage}. Lemma \ref{le: hajek CLT} gives a central limit theorem for a linearized version of $\hat{\tau}_k^\Haj$. Lemma \ref{le: fatou} is a corollary of Fatou's lemma, and will be used to show that $\hat{V}_k$ is conservative for the variance of the linearized version of $\hat{\tau}_k^\Haj$. It appears as Exercise 3.2.4 in \citeappendix[5th edition]{durrett2019probability}.

\begin{lemma} \label{le: hajek CLT}
Let Assumption \ref{as: design} hold. For $k \in [2]$, let $\mu_1$ and $\mu_0$ denote
\begin{align*}
\mu_1 &= \sum_{i \in \Iscr} \mathbb{E}[\tilde{f}_{ki}(T_i)] , & \mu_0 &= \sum_{i \in \Iscr} \mathbb{E}[1-\tilde{f}_{ki}(T_i)]. 
\end{align*}
where $\tilde{f}_{ki}$ is given by \eqref{eq: f_ki}. Let $h_i$ for $i \in [N]$ denote
\begin{align*}
h_i = \frac{1\{\bar{Z}_i=1\}}{P(\bar{Z}_i = 1| T_i)}Y_{ki} \phi_i^* - \mu_1 + \frac{1\{\bar{Z}_i=0\}}{P(\bar{Z}_i = 0| T_i)}(1-Y_{ki}) \phi_i^* - \mu_0
\end{align*}
where $\phi^*$ is given by \eqref{eq: phi star}. Let $\sigma^2 = \Var \sum_i h_i$, and suppose that $\sigma^2 \geq N^2 / M^{4/3 - \epsilon}$ for any $\epsilon > 0$. Then $\sum_i h_i / \sigma$ converges to a standard normal.
\end{lemma}

\begin{lemma}\citeappendix[Exercise 3.2.4]{durrett2019probability} \label{le: fatou}
Let $X_n$ denote a sequence of random variables that converges in distribution to $X_{\infty}$, and let $g$ denote a non-negative continuous function. Then $\liminf\limits_{n \rightarrow \infty} \mathbb{E}g(X_n) \geq \mathbb{E}g(X_{\infty})$
\end{lemma}

Theorem \ref{th: ross} is taken from \citeappendix{ross2011fundamentals}, and gives a bounded dependence central limit theorem. It will be used to prove the asymptotic normality result of Lemma \ref{le: hajek CLT}.

\begin{theorem}\citeappendix[Thm 3.6]{ross2011fundamentals} \label{th: ross}
Let $U_1,\ldots,U_n$ be random variables such that $\mathbb{E}[U_i^4] \leq \infty$, $\mathbb{E}U_i = 0$, with variance denoted by $\sigma^2 = \Var \sum_{i=1}^N U_i$, with dependency neighborhoods $N_1,\ldots,N_n$, which are subsets of $[n]$ (which need not be disjoint) such that $U_j$ and $U_i$ are independent if $U_j \notin N_i$. Let $D = \max_i |N_i|$. Let $W = \sum_{i=1}^N U_i/\sigma$. Then for $Z$ denoting a $N(0,1)$ random variable, and $d_W$ denoting Wasserstein metric, it holds that
\[ d_W(W, Z) \leq \frac{D^2}{\sigma^3}\sum_{i=1}^n \mathbb{E}\left[ |U_i|^3\right] + \frac{\sqrt{28}D^{3/2}}{\sqrt{\pi}\cdot\sigma^2} \sqrt{\sum_{i=1}^n \mathbb{E}\left[U_i^4\right]}\]
\end{theorem}

\paragraph{Proof of Theorem \ref{th: hajek coverage}} \label{sec: proofs 5}

\begin{proof}[Proof of Theorem \ref{th: hajek coverage}]

It holds that
\begin{align}
LB_{1-\alpha} - \tau &=  \max \left\{ \min_{\phi \in \{0,1\}} \hat{\tau}_1(\phi) - z_{1-\frac{\alpha}{2}} \sqrt{\widetilde{V}_1(\phi)} \ , \ \min_{\phi \in \{0,1\}} \hat{\tau}_2(\phi) - z_{1-\frac{\alpha}{2}} \sqrt{\widetilde{V}_2(\phi)} \right\} - \tau \nonumber \\
 & \leq \max \left\{ \hat{\tau}_1 - z_{1-\frac{\alpha}{2}} \sqrt{\widetilde{V}_1} \ , \
 \hat{\tau}_2 - z_{1-\frac{\alpha}{2}} \sqrt{\widetilde{V}_2} \right\} - \tau \nonumber \\
 & = \max \left\{\hat{\tau}_1  - z_{1-\frac{\alpha}{2}} \sqrt{\widetilde{V}_1} - \tau \ , \
 \hat{\tau}_2 - z_{1-\frac{\alpha}{2}} \sqrt{\widetilde{V}_2} - \tau \right\} \label{eq: coverage expression}
\end{align}
Let $\mathcal{E}_k$ denote the event
\[ \mathcal{E}_k = \left\{ \hat{\tau}_k - z_{1-\frac{\alpha}{2}}\sqrt{\widetilde{V}_k} - \tau \geq 0\right\}. \]
By \eqref{eq: coverage expression}, it follows that loss of coverage occurs if either $\mathcal{E}_1$ or $\mathcal{E}_2$ occurs.
By Lemma \ref{le: Phi coverage}, the probability of $\mathcal{E}_k$ is bounded by $\alpha/2 + o(1)$, and hence by union bound it holds that 
\begin{align} \label{eq: coverage bad event prob}
 P(\mathcal{E}_1 \cup \mathcal{E}_2) \leq \alpha + o(1)
\end{align}
Combining \eqref{eq: coverage expression} and \eqref{eq: coverage bad event prob} yields that
\begin{align}
P(LB_{1-\alpha} - \tau > 0) \leq \alpha + o(1)
\end{align}
proving that $LB_{1-\alpha} \leq \tau$ with probability converging to at least $1-\alpha$.
\end{proof}

\paragraph{Proof of Lemmas \ref{le: Phi coverage} and \ref{le: hajek CLT}}

\begin{proof}[Proof of Lemma \ref{le: Phi coverage}]
Let $h_i$ be defined as in Lemma \ref{le: hajek CLT}, and let $\sigma^2 = \Var \sum_i h_i$. Divide the sequence of experiments into two subsequences, depending on whether $\sigma^2 \geq N^2 / M^{4/3-\epsilon}$.

For the subsequence that satisfies $\sigma^2 \geq N^2 / M^{4/3-\epsilon}$, it holds by linearization arguments made in the proof of Theorem \ref{th: hajek consistency} that
\begin{align}
 \frac{\hat{\tau}_k^{\Haj} - \tau}{\sigma} & = \frac{\nabla h(\mathbb{E}\xi)^T(\xi - \mathbb{E}\xi)}{\sigma} + O_P\left(\frac{1}{\sigma}\right) \label{eq: coverage linearization 1}\\
&=  \frac{\sum_{i=1}^N h_i}{\sigma} + O_P\left(\frac{1}{\sigma}\right) \label{eq: coverage linearization 2}
 \end{align}
where \eqref{eq: coverage linearization 1} holds by \eqref{eq: hajek linearization} and \eqref{eq: remainder consistency}, while \eqref{eq: coverage linearization 2} holds by \eqref{eq: linearization terms}-\eqref{eq: epsilon} and \eqref{eq: nabla h}.  By Lemma \ref{le: hajek CLT} and \eqref{eq: coverage linearization 2} it follows that 
\begin{align*} 
\frac{\hat{\tau}_k^{\Haj} - \tau}{\sigma} \rightarrow N(0,1)
\end{align*}
from which it follows by Lemma \ref{le: fatou} that
\begin{align} \label{eq: fatou usage 1}
\liminf\limits_{N \rightarrow \infty} \Var \left[ \frac{\hat{\tau}_k^\Haj - \tau}{\sigma} \right]  \geq 1
\end{align}
Since by Theorem \ref{th: hajek variance estimation} it holds that 
\begin{align}
 \frac{\hat{V}_k}{\sigma^2} & = \frac{\mathbb{E}\hat{V}_k}{\sigma^2} + \frac{1}{\sigma^2} O_P\left(\frac{N^2 \log M}{M^{3/2}}\right)  \nonumber \\
 & \geq \Var \left( \frac{\hat{\tau}_k^\Haj}{\sigma}\right) + o_P\left(1\right) \label{eq: fatou usage 2}
\end{align}
Combining \eqref{eq: fatou usage 1} and \eqref{eq: fatou usage 2} yields
\begin{align*} 
\liminf\limits_{N \rightarrow \infty} \frac{\hat{V}_k}{\sigma^2} \geq 1 
\end{align*}
which implies that
\begin{align} \label{eq: CLT subsequence}
P\left( \hat{\tau}_k^\Haj - \tau \geq z_{1-\frac{\alpha}{2}} \sqrt{\hat{V}_k}\right) \leq \frac{\alpha}{2} + o(1)
\end{align}
for the subsequence of experiments that satisfies $\sigma^2 \geq N^2 / M^{4/3-\epsilon}$.

For the other subsequence, Chebychev inequality implies that
\begin{align*} 
P\left(\left| \sum_i h_i \right| \geq \frac{N}{M^{2/3-\epsilon}\sqrt{\alpha/2}}\right) \leq \alpha/2 
\end{align*}
where we have used that $\frac{N}{M^{2/3-\epsilon}} = \omega(\sigma)$ for $\epsilon > 0$. In turn, this implies that
\begin{align} \label{eq: using cheby 2}
P\left(\hat{\tau}_k - \tau \geq \frac{N}{M^{2/3-\epsilon}\sqrt{\alpha/2}} \right) \leq 1-\alpha + o(1)
\end{align}
where we have used that $\hat{\tau}_k - \tau = \sum_i h_i + O_P(1)$, which holds by  \eqref{eq: coverage linearization 1}. Applying \eqref{eq: using cheby 2} yields
\begin{align} \label{eq: using cheby 3}
P\left(\hat{\tau}_k - \tau \geq z_{1-\frac{\alpha}{2}} \sqrt{\Delta} \right) \leq \alpha/2 + o(1)
\end{align} 
where $\Delta = N/ (M^{2/3-\epsilon} z_{1-\frac{\alpha}{2}}^2 \cdot \alpha/2 )$. Combining \eqref{eq: CLT subsequence} and \eqref{eq: using cheby 3} yields that for all experiments, it holds that
\begin{align*}
P\left(\hat{\tau}_k - \tau \geq z_{1-\frac{\alpha}{2}} \sqrt{\max\left(\hat{V}_k, \Delta\right)} \right) \leq \alpha/2 + o(1),
\end{align*}
As $\widetilde{V}_k$ can be seen to equal $\max\left(\hat{V}_k, \Delta\right)$, this proves the lemma.

\end{proof}

\begin{proof}[Proof of Lemma \ref{le: hajek CLT}]
Let $U_m$ for $m \in [M]$ be given by
\[ U_m = \sum_{i \in G_m} h_i\]
Since $|h_i| \leq \frac{1}{\rho}$ and $|G_m| \leq CN/M$, it follows that $|U_m| \leq \frac{CN}{M \rho}$. Additionally, it can be seen that the dependency neighborhoods of $U_1,\ldots,U_M$ have maximum size $d^2$, since each $U_m$ depends on $X$ only through $\{X_{G_\ell}: \ell \in \eta_m\}$. As a result, applying Theorem \ref{th: ross} yields
\begin{align} 
\nonumber d_W\left(\frac{\sum_i h_i}{\sigma}\ ,\ N(0,1) \right) & \leq \frac{d^4}{\sigma^3}\cdot M \left(\frac{CN}{M\rho}\right)^3 + 
\frac{\sqrt{28}}{\sqrt{\pi}}\cdot \frac{d^3}{\sigma^2}\sqrt{M\left(\frac{CN}{M\rho}\right)^4} \\
& = O\left(\frac{N^3}{\sigma^3M^2}\right) + O\left(\frac{N^2}{M^{3/2}\sigma^2}\right) \\
& = o(1)\label{eq: hajek ross 2}
\end{align}
where \eqref{eq: hajek ross 2} uses the assumption of the lemma that $\sigma^2 \geq \frac{N}{M^{2/3-\epsilon}}$ for  $\epsilon > 0$.
\end{proof}

 \bibliographystyleappendix{apalike}
 {\footnotesize
 \bibliographyappendix{bibfile}
}

\end{document}